%% file: HestonVar_article_main.tex
\documentclass[final,hidelinks,onefignum,onetabnum]{siamart220329}

\usepackage[margin=1in]{geometry}

\input{HestonVaR_shared}

\ifpdf
\hypersetup{
  pdftitle={VaR-constrained portfolios in incomplete markets},
  pdfauthor={M. Escobar-Anel, Y. Havrylenko, and R. Zagst}
}
\fi

\begin{document}

\maketitle

\begin{abstract}
We solve an expected utility-maximization problem with a Value-at-risk constraint on the terminal portfolio value in an incomplete financial market due to stochastic volatility. To derive the optimal investment strategy, we use the dynamic programming approach. We demonstrate that the value function in the constrained problem can be represented as the expected modified utility function of a vega-neutral financial derivative on the optimal terminal wealth in the unconstrained utility-maximization problem. Via the same financial derivative, the optimal wealth and the optimal investment strategy in the constrained problem are linked to the optimal wealth and the optimal investment strategy in the unconstrained problem. In numerical studies, we substantiate the impact of risk aversion levels and investment horizons on the optimal investment strategy.  We observe a $20\%$ relative difference between the constrained and unconstrained allocations for average parameters in a low-risk-aversion short-horizon setting.
\end{abstract}



\section*{Acknowledgments}
 Yevhen Havrylenko and Rudi Zagst acknowledge the financial support of the ERGO Center of Excellence in Insurance at the Technical University of Munich promoted by ERGO Group. Yevhen Havrylenko expresses gratitude to the participants of the 15-th Actuarial and Financial Mathematics Conference and the participants of a Mathematical and Computational Finance seminar at the University of Calgary for their feedback on the paper. 

 \section*{Conflict of interest}
The authors declare no conflict of interest.

\section{Introduction}
In the context of incomplete financial markets, this paper establishes the relevance of dynamic programming techniques for portfolio optimization problems with terminal wealth constraints. This development allows us to find the first analytical solution to an expected utility maximizer in the presence of a Value-at-Risk constraint on the terminal wealth in a stochastic-volatility environment as per Heston model (see \cite{Heston1993}). This is an important problem in the banking and insurance sectors; not only there is ample evidence of time-dependent volatilities in financial markets, see \cite{Wiggins1987} and \cite{Taylor1994}, but also financial institutions have to comply with the minimum capital reserve based on Value-at-Risk (VaR) required by the Basel Committee on Banking (see \cite{Basak2001}), and similarly the insurance sector must provide minimum guarantees due to solvency regulations (see \cite{Basak1995}). These are effective constraints on their operating portfolios (see \cite{Boyle2007}). 

Our methodology generalizes the pioneering work of \cite{Kraft2013} to incomplete markets due to stochastic volatility by demonstrating that the optimal wealth in the constrained optimization problem can be represented as a rolling-over vega-neutral financial derivative on the optimal wealth in the unconstrained optimization problem. To prove the result, we use a convenient financial derivative and match Hamilton-Jacobi-Bellman (HJB) equations under the real-world probability measure as well as an equivalent-martingale measure (EMM). Importantly, the link between the optimal terminal wealth in the wealth-constrained portfolio optimization problem and the optimal terminal wealth in the unconstrained problem was established in \cite{Basak2001}, where the authors considered a complete financial market and used the martingale approach (see \cite{Pliska1986}, \cite{Karatzas1987}). In contrast to \cite{Kraft2013}, \cite{Basak2001} and other related papers, we demonstrate this link for constrained problems in an incomplete financial market with stochastic volatility.

In closely related literature,  \cite{EscobarAnel2022a} derives the analytical representation of the optimal investment strategy for a decision maker maximizing his/her expected power utility of terminal wealth subject to a VaR constraint on the running minimum of the wealth process. To achieve this, the author extends the methodology of \cite{Kraft2013} to path-dependent constraints in a complete Black-Scholes market. \cite{Chen2018b} solves the portfolio optimization problem with a VaR constraint on terminal wealth in a complete financial market with three assets: a risk-free asset, a stock whose price dynamics follows the Heston model, and a continuously traded financial derivative that allows the investor to hedge the variance risk. In their setting of a complete market, the authors apply the martingale approach. When the market is incomplete, the martingale approach becomes significantly more difficult even for problems without constraints on the terminal portfolio value, since the investor cannot hedge any generic contingent claim and there are infinitely many equivalent martingale measures. \cite{Ntambara2017} addresses portfolio optimization problems with a constraint limiting the present expected short-fall of terminal wealth in an incomplete financial market that consists of one risky asset and one money-market account with a stochastic interest rate following 1-factor or 2-factor Vasicek model. Using the martingale approach, duality and Malliavin calculus, the researcher derives optimal investment strategies. According to \cite{Ntambara2017}, the distribution of a so-called deflator (also known as the pricing kernel) must be obtained prior to solving the dual optimization problem and a deflator is not known when a stochastic interest rate is described by the Cox-Ingersoll-Ross (CIR) model. This aspect motivates us to tackle our wealth-constrained portfolio optimization problem via a dynamic programming approach in contrast to the martingale approach, since the stochastic volatility in our incomplete market is also governed by the CIR model. Our work provides more evidence of the usefulness of Bellman's principle of optimality for portfolio optimization problems commonly tackled via the martingale approach.

Our methodology can be extended in many directions, for instance, to other incomplete market problems, e.g., stochastic market price of risk, stochastic interest rates, or stochastic correlations; other terminal or intermediate constraints on wealth like expected shortfall; or other utilities like HARA or piece-wise concave. Each of these cases would need special considerations in terms of matching partial differential equations (PDEs) and terminal conditions. In other words, each problem requires a special crafting, i.e. an ansatz, of the financial derivatives permitting the matching of PDEs and, hence, linking constrained and unconstrained problems.

Closed-form solutions to wealth-constrained optimization problems in incomplete markets have remained elusive through the years, mainly due to the lack of techniques to tackle the problem. Next, we highlight relevant sources on utility maximization in incomplete markets for general constraints. \cite{Karatzas1991} consider an extension of the Black-Scholes market where the number of risk drivers is larger than the number of traded stocks, placing constraints on investment strategies, rather than wealth. Their idea is to complete fictitiously the financial market. This completion is based on a suitably parameterized family of fictitious completions, which correspond to exponential local martingales. The \enquote{right} completion should satisfy a certain minimality property. \cite{Gundel2007} study this approach for the only optimal terminal wealth but do not derive the corresponding investment strategies. The explicit representation of the optimal terminal portfolio value is derived via certain worst-case measures, which can be characterized as minimizers of a dual problem. In parallel, \cite{He1991} apply a martingale approach to solve a continuous-time consumption-investment problem in a setting with an incomplete market and short-sale (strategy) constraints. They introduce the notion of minimax local martingale, transforming the dynamic problem into a static problem. Showing when the minimax local measure exists and how it is characterized, they derive conditions when the optimization has a solution, then linking the optimal strategies to the solution of a quasi-linear PDE. 

Our paper is organized as follows. Section \ref{sec:constrained_portfolios_general} introduces the problem at hand, a few important well-known results, and a first general representation of the main theorem of the paper. Section \ref{sec:explicit_formulas} focuses on explicitly applying the theorem to the power-utility maximization problem subject to VaR constraints in a Heston-model-based financial market. Section \ref{sec:numerical_studies} reports details on the numerical implementation and the most significant results. Conclusions and an outlook for further research are presented in Section \ref{sec:consclusion}. Appendix \ref{app:unconstrained_problem} contains the results related to the unconstrained optimization problem. Appendix \ref{app:constrained_OP_solution_short} provides proofs of theoretical results for solving the VaR-constrained problem. Appendix \ref{app:explicit_formulas} contains the derivation of explicit formulas to calculate the price and the Greeks of a synthetic derivative linking the constrained optimization problem and the unconstrained one. Additional plots from numerical studies are provided in Appendix \ref{app:numerical_studies_turbulent_markets}.

\section{Constrained portfolio optimization problem and its solution}\label{sec:constrained_portfolios_general}

We consider an investor maximizing utility from terminal wealth at time $T$
with respect to a continuous and increasing utility (primal) function $%
U$. The price process $B(t)$ of the risk-free asset
evolves according to $dB(t)=rB(t)dt,B(0) = 1$, and the interest rate $r$ is
assumed to be constant. The price process $S(t)$ of the risky asset follows
Heston's stochastic volatility model, introduced in \cite{Heston1993}. Its dynamics under the real-world measure $\mathbb{P}$ is given by the stochastic differential equation (SDE):
\begin{align}
d S(t)& =S(t)\left[ \left( r+\overline{\lambda }v(t)\right) dt+\sqrt{v(t)}\,%
dW^{\mathbb{P}}_{1}(t)\right]  \label{Heston-stock}; \\
d v(t)& =\kappa \left( \theta -v(t)\right) dt+\sigma \sqrt{v(t)} \rho \,d%
W^{\mathbb{P}}_{1}(t)\,+\sigma \sqrt{v(t)}\sqrt{1-\rho ^{2}}dW^{\mathbb{P}}_{2}(t)
\label{Heston-vol};
\end{align}%
with starting values $ S(0) = s_0 > 0$ and $v(0) = v_0 > 0$, premium for volatility $\overline{\lambda }>0$, mean-reversion rate $\kappa >0$, long-run mean $\theta >0$, volatility of the variance $\sigma>0$ and fulfilling Feller's condition $\kappa \theta > \frac{\sigma^2}{2}$. The portfolio value process under the real-world measure $\mathbb{P}$ evolves according to:
\begin{equation*}
dX^{x_0,\pi }(t)=X^{x_0, \pi }(t)\left[ \left( r+\pi (t)\overline{\lambda }%
v(t)\right) dt+\pi (t)\sqrt{v(t)}dW^{\mathbb{P}}_{1}(t)\right], \qquad
X(0) = x_{0}>0,
\end{equation*}%
where $\pi (t)$ denotes the proportion of wealth invested in the stock at time $\tin$, with
$1-\pi (t)$ invested in the cash account, and $x_{0}$ is the initial budget.

We consider the set of EMMs that have the following Radon-Nikodym derivatives w.r.t. $\mathbb{P}$:
\begin{equation*}
\begin{aligned}
    \frac{d \mathbb{Q}(\bar{\lambda}, \lambda^v(\cdot))}{d \mathbb{P}} = \exp \biggl(&-\int \limits_{0}^{T} \lambar \sqrt{v(s)}\,d W^{\mathbb{P}}_{1}(s) - \int \limits_{0}^{T} \lambda^{v}(s) \sqrt{v(s)}\,d W^{\mathbb{P}}_{2}(s) \\
    &- \frac{1}{2}\int \limits_{0}^{T}\rBrackets{\rBrackets{\lambar \sqrt{v(s)}}^2 + \rBrackets{\lambda^{v}(s) \sqrt{v(s)}}^2}\,ds \biggr) ,
    \end{aligned}
\end{equation*}
where $\mathbb{Q}(\bar{\lambda}, \lambda^{v})$ denotes a specific EMM, $\lambda^{v}(s)$ is assumed to be dependent on $t$ and independent of $v(s)$, as per \cite{Heston1993}, and it also satisfies the Novikov's condition:

\begin{equation*}
    \mathbb{E}^{\mathbb{P}}\sBrackets{\exp \rBrackets{\frac{1}{2}\int \limits_{0}^{T}\rBrackets{\rBrackets{\lambar \sqrt{v(s)}}^2 + \rBrackets{\lambda^{v}(s) \sqrt{v(s)}}^2}\,ds}} < +\infty.
\end{equation*}

To make the notation concise, we will write only $\lambda^{v}$ and $\mathbb{Q}(\lambda^{v})$, since only $\lambda^{v}$ is a degree of freedom in the choice of the EMM. Moreover, we assume that $\lambda^{v}$ is such that
$v(t) \geq 0$ $\forall\,\tin$ under $\mathbb{Q}(\bar{\lambda}, \lambda^{v})$ (see \eqref{eq:Heston_S_under_Q} and \eqref{eq:Heston_v_under_Q} below). 

The Heston model under the EMM $\mathbb{Q}(\lambda^v)$ is given by%
\begin{align}
& dS(t)=S(t)\left[ rdt+\sqrt{v(t)}dW^{\mathbb{Q}}_{1}(t)\right] \label{eq:Heston_S_under_Q}; \\
& d v(t)=\tilde{\kappa}\left( \tilde{\theta}-v(t)\right) dt+\sigma 
\sqrt{v(t)} \rho d W^{\mathbb{Q}}_1(t)\;+\sigma \sqrt{v(t)}\sqrt{1-\rho ^{2}} d
W^{\mathbb{Q}}_{2}(t)  \label{eq:Heston_v_under_Q},
\end{align}%
where $S(0) = s_0 > 0$, $v(0) = v_0 > 0$,  $d W^{\mathbb{P}}_{1}(t)=-\overline{\lambda }\sqrt{v(t)}dt+d W^{\mathbb{Q}}_1(t)$, $d W^{\mathbb{P}}_{2}(t)=-\lambda ^{v}\sqrt{v(t)}dt+dW^{\mathbb{Q}}_{2}(t)$, $\tilde{\kappa}=\kappa +\sigma
\overline{\lambda }\rho +\sigma \lambda^v\sqrt{1-\rho ^{2}}$, $\tilde{%
\theta}=\kappa \theta /\tilde{\kappa}$. 

The wealth process under the EMM $\mathbb{Q}(\lambda^v)$ evolves according to:
\begin{equation*}
dX^{x_0,\pi }(t)=X^{x_0, \pi }(t)\left[ rdt+\pi (t)\sqrt{v(t)}dW^{\mathbb{Q}}_{1}(t)%
\right], \qquad X^{x_0, \pi }(0)=x_0 > 0.
\end{equation*}

Let $\mathcal{U}(x_0, v_0)$ be the set of all admissible unconstrained investment strategies on $[0, T]$:
\begin{align*}
     \mathcal{U}(x_0, v_0) := \Biggl\{\left. \pi := \rBrackets{\pi(t)}_{\tin} \right \rvert &\pi \text{ is progressively measurable}, X^{x_0, \pi}(0) = x_0,\,\\
     & v(0) = v_0, \, \int_0^T \rBrackets{\pi (t)X^{x_0, \pi}(t)}^2\,dt < \infty\,\mathbb{P}\text{-a.s.} \Biggr\}
\end{align*}
and $\mathcal{U}(t, x, v)$ be the set of all admissible unconstrained investment strategies $\pi$ on $[t, T]$, given that $X^{x, \pi}(t) = x$, and $v(t) = v$.
Denote by $\mathcal{A}(x_0, v_0)\subset \mathcal{U}(x_0, v_0)$ the set of all admissible investment strategies that satisfy at $t=0$ a VaR constraint $\mathcal{A}(x_0, v_0) = \left\{ \pi \in \mathcal{U}(x_0, v_0) \mid \mathbb{P}\left( X^{x_0, \pi}(T)<K\right) \leq \varepsilon \right\}$. As in \cite{Basak2001} and \cite{Kraft2013}, the VaR constraint is static, i.e., it is satisfied only at the initial time $t = 0$ and may not be satisfied in general at a later time $t > 0$. Readers interested in a dynamic version of a risk constraint are referred to \cite{pirvu2007portfolio}, where the author considers a dynamic VaR of the projected portfolio loss over infinitesimally small time periods. 

Unless otherwise stated, the decision maker maximizes the expected power utility function $U\left( x\right) ={x^{\gamma }} / {\gamma }$, $\gamma \in (-\infty, 0) \cup (0,1)$, $x>0$, evaluated at the terminal wealth $X^{x_0,\pi}(T)$. So, the main problem we study is
\begin{equation}\label{eq:OP_main}
    \mathcal{V}^{c}\left(0,x_0,v_0\right):= \underset{\pi \in \mathcal{A}(x_0, v_0)}{\max}\mathbb{E}^{\mathbb{P}}_{0, x_0, v_0}\sBrackets{U\left( X^{x_0, \pi}(T)\right)},
\end{equation}
where we write $\EVtxv{M}{\cdot} := \mathbb{E}^{\mathbb{M}}\sBrackets{\cdot|X^{x, \pi}(t) = x, v(t) = v}$ for $\mathbb{M} \in \{\mathbb{P}, \mathbb{Q} \}$. Analogously, we will use the notation $\mathbb{M}_{t,x,v}\rBrackets{\cdot} := \mathbb{M}\rBrackets{\cdot\,|\,X^{x, \pi}(t) = x, v(t) = v}$ for $\mathbb{M} \in \{\mathbb{P}, \mathbb{Q} \}$.

Since the VaR constraint is static and must be abided by the investor only at $t = 0$, \eqref{eq:OP_main} can be transformed to an equivalent problem using a proper (optimal) Lagrange multiplier $\lambda _{\varepsilon } \geq 0$:
\begin{equation}\label{eq:OP_main_Lagrangian}
    \mathcal{V}^{c}\left(0,x_0,v_0\right) = \underset{\pi \in \mathcal{U}(x_0, v_0)}{\max}\mathbb{E}^{\mathbb{P}}_{0, x_0, v_0}\sBrackets{\overline{U}\left( X^{x_0, \pi}(T)\right)},
\end{equation}
where $\overline{U}\left( x\right) =U\left( x\right) -\lambda _{\varepsilon }\Bigl(1_{\{x <  K\}}$ $- \varepsilon\Bigr)$ is a modified utility function. Static VaR constraint implies that $\lambda_\varepsilon$ is constant and the problem is time consistent, i.e., the Bellman's principle of optimality holds. If we imposed a dynamic VaR constraint, i.e., $\mathbb{P}_{t,x,v}\left( X^{x, \pi}(T) < K \right) \leq \varepsilon$ $\forall (t, x, v) \in [0, T] \times (0, +\infty) \times (0, +\infty) $, then the $\lambda_{\varepsilon}$ would be a function of $(t,x, v)$ and the respective optimization problem would be time inconsistent, i.e., the dynamic programming approach would not be applicable and a different notion of optimality would be needed, e.g, see \cite{Bjoerk2021}. 

We embed \eqref{eq:OP_main_Lagrangian} into a family of related problems by varying time $\tin$: 
\begin{equation}\label{MainControlProb} \tag{PC}
\mathcal{V}^{c}\left( t,x,v\right)  := \underset{\pi \in \mathcal{U}(t, x, v)}{\max}\EVtxv{P}{ \overline{U}\left( X^{x, \pi}(T)\right) }. 
\end{equation}
We denote by $\pic = (\pic(t))_{\tin}$ the optimal investment strategy for \eqref{MainControlProb} and by $X^{\ast}(t):= X^{x, \pic}(t), \tin,$ the corresponding optimal wealth process.

We will solve \eqref{MainControlProb} using the solution to the following unconstrained optimization problem:
\begin{equation}\label{eq:OP_unconstrained}
    \mathcal{V}^{u}(0,y_0,v_0):=\underset{\pi \in \mathcal{U}(y_0, v_0)}{\max}\mathbb{E}^{\mathbb{P}}_{0, y_0, v_0}\sBrackets{U\left( Y^{y_0, \pi}(T)\right)},
\end{equation}
where we denote by $Y^{y_0, \pi}(t), \tin$ the respective wealth process to emphasize its relation to the unconstrained problem. As in the constrained case, we can embed \eqref{eq:OP_unconstrained} into a family of similar unconstrained, time consistent, problems that start at a different initial point $(t, y, v) \in [0, T] \times (0, +\infty) \times (0, +\infty)$:
\begin{equation}\label{PowerUnconstProb} \tag{PU}
\mathcal{V}^{u}(t,y,v):=\underset{\pi \in \mathcal{U}(t, y, v)}{\max }~\EVtyv{P}{
U(Y^{y, \pi}(T))}. 
\end{equation}

Let $\piu = (\piu(t))_{\tin}$ be the optimal unconstrained investment strategy and $Y^{\ast}(t):= Y^{y_0, \piu}(t), \tin$, be the optimal unconstrained wealth process \eqref{eq:OP_unconstrained}. These objects are known, since \eqref{PowerUnconstProb} has been well studied in the literature. In particular, for Heston's models whose parameters satisfy the following condition (same as Condition (26) in \cite{Kraft2005}):
\begin{equation}
    \frac{\gamma }{1-\gamma }\overline{\lambda } \left(\frac{\kappa \rho}{\sigma} +\frac{\overline{\lambda }}{2} \right) <\frac{\kappa ^{2}}{2 \sigma^2 },\label{KraftCondition}
\end{equation}
\cite{Kraft2005} solves the unconstrained utility maximization problem using the HJB approach to derive a candidate solution and then provides a verification result. \cite{Kallsen2010} combines the martingale method, the concept of an opportunity process, and the calculus of semi-martingale characteristics for parameters that may violate Condition \eqref{KraftCondition}.

In Appendix \ref{app:unconstrained_OP_solution}, for completeness, we provide two propositions regarding the unconstrained optimization problem. In Proposition \ref{prop:unconstrained_problem_solution} we show the optimal investment strategy, the optimal wealth, and the value function in \eqref{PowerUnconstProb}, which is a concise version of the results obtained in \cite{Kraft2013} adapted to our notation. In Proposition \ref{prop:characteristic_fct_log_Y}, we derive the characteristic functions of the logarithm of the unconstrained optimal wealth $Z^\ast(t):=\ln \rBrackets{Y^{\ast}(t)},\, \tin,$ under $\mathbb{P}$ and $\mathbb{Q}(\lambda^{v})$. These characteristic functions are needed later for pricing financial derivatives on the optimal unconstrained wealth and calculating their Greeks with the help of the inverse Fourier transform.

If $\mathbb{P}_{0, x_0, v_0}\rBrackets{Y^{x_0, \piu}(T) < K} \leq \varepsilon$, then, obviously, the VaR constraint is non-binding, $\lambda_{\varepsilon}^\ast = 0$, $\pic(t) = \piu(t)$ and $X^{x_0, \pic}(t) = Y^{x_0, \piu}(t)$ $\forall \tin$. So, from now on, we assume that $\mathbb{P}_{0, x_0, v_0}\rBrackets{Y^{x_0, \piu}(T) < K} > \varepsilon$ and the investor's initial capital $x_0$ is sufficiently large to satisfy the VaR constraint on the terminal wealth. The optimal wealth $X^{x_0, \piu}$ for the constrained problem \eqref{MainControlProb} will be represented via a to-be-conjectured financial derivative on the optimal unconstrained wealth $Y^{\ast}$ and the variance process $v$. They have the following SDEs under the EMM $\mathbb{Q}(\lambda^{v})$:%
\begin{eqnarray}
dY^{\ast}(t) &=&Y^{\ast}(t)rdt+Y^{\ast}(t)\piu(t)\sqrt{v(t)}dW^{\mathbb{Q}}_{1}(t) ;\label{eq:Y_Star_SDE_under_Q}\\
d v(t) &=&\tilde{\kappa}\left( \tilde{\theta}-v(t)\right) dt+\sigma \sqrt{%
v(t)}\rho d W^{\mathbb{Q}}_{1}(t)+\sigma \sqrt{v(t)}\sqrt{1-\rho ^{2}}d W^{\mathbb{Q}}_{2}(t). \notag
\end{eqnarray}

Let $D(\cdot, \cdot)$ be a Borel-measurable payoff function\footnote{The payoff function $D(\cdot, \cdot)$ may have points of discontinuity, but the function must be Borel measurable, so that we may later apply the Feynman-Kac theorem, see, e.g., Theorem 6.4.1 in \cite{Shreve2004}).} of a financial derivative on $Y^\ast$ and $v$.  We denote  by $\dqla(t,y,v)$ the price of such a contingent claim at $\tin$ such that $\dqla(T, y, v) = D(y,v)$. At $t = 0$, this financial derivative should satisfy the budget constraint and the terminal-wealth constraint, i.e., $\dqla(0, y_0, v_0) = x_0$ and $\mathbb{P}\rBrackets{D(Y^{\ast}(T), v(T)) < K|Y^{\ast}(0) = y_0, v(0) = v_0} = \varepsilon$ respectively. 

The PDE for the price of $D(Y^\ast(T), v(T))$
\begin{equation}\label{eq:def_of_financial_derivative_on_Y}
\dqlatyv=\mathbb{E}_{t,y,v}^{\mathbb{Q}(\lambda^{v})}\left[\exp\rBrackets{-r(T-t)} D(Y^{\ast}(T), v(T))\right]
\end{equation}%
is known via the Feynman-Kac (FK) theorem:%
\begin{flalign}
\dqla_{t}&=r\dqla - ry\dqla_{y}-\tilde{\kappa}\left( \tilde{\theta}-v\right) \dqla_{v}\label{eq:dq_pde} \\
 &\quad -\frac{1}{2}v\Bigg[y^{2}(\piu)^{2}\dqla_{yy}+2\sigma \rho y\piu \dqla_{yv}+\sigma ^{2}\dqla_{vv}\Bigg];\notag \\
\dqla(T,y,v)&=D(y,v). \notag
\end{flalign}

The expected modified utility of the claim that is based on the modified utility function $\overline{U}(\cdot)$ is:
\begin{equation}\label{eq:def_expected_aux_utility}
\udptyv=\EVtyv{P}{\overline{U}(D(Y^{\ast}(T), v(T)))},
\end{equation}
where the optimal wealth process under $\mathbb{P}$ comes from
Proposition \ref{prop:unconstrained_problem_solution}. Due to the FK theorem, the investor's expected modified utility $\rBrackets{\udptyv}$ of the contingent claim satisfies the following PDE:
\begin{flalign}
&0=\udp_{t}+\left( r+\piu \overline{\lambda }v\right) y\udp_{y}+\kappa
\left( \theta -v \right) \udp_{v} +\frac{1}{2}v\Bigg[y^{2}(\piu)^{2}\udp_{yy}+2\sigma \rho y\piu \udp_{yv}+\sigma ^{2}\udp_{vv}\Bigg] \label{eq:udp_pde} ;\\
&\udp(T,y,v)=\overline{U}(D(y, v)). \label{eq:udp_terminal_condition}
\end{flalign}

We show now that the wealth of the constrained problem 
can be represented by the price $\dqla$ of a contingent claim
on $Y^\ast$ and $v$, and the value function $\mathcal{V}^{c}$ by the expected utility on the contingent claim $\udp$. The following theorem is our main result. It provides three conditions under which the PDEs and the terminal conditions associated with $\mathcal{V}^{c}\left( t, x ,v\right) $ and $\udp(t,y,v)$ coincide, with $x=\dqla \left( t,y,v\right)$.

\bigskip
\begin{theorem}[Representation of constrained-problem solution]\label{MainTheo}
$\,$\\
Assume that Condition \eqref{KraftCondition} holds and that the VaR constraint is feasible in \eqref{MainControlProb}. Let $D(\cdot,\cdot)$, $y_0$, $\lambda^v(\cdot)$ and $\lambda_{\varepsilon}$ be such that $\mathbb{P}_{0, y_0, v_0}\rBrackets{D(Y^{y_0, \piu}(T), v(T)) < K} = \varepsilon $ at $t=0$, $D(y,v)$ is non-decreasing in $y \in (0, +\infty)$ for any $v>0$ and strictly increasing on a non-empty open sub-interval of $(0, +\infty)$, and the following three conditions are satisfied at each time $t \in [0,T]$:
\begin{eqnarray}
-\frac{{y}\udp_{yy}(t,y,v)}{\udp_{y}(t,y,v)}&=&-\frac{{y}\dqla_{yy}(t,y,v)}{\dqla_{y}(t,y,v)}+{1-\gamma };
\label{cond:U_D_yy_y} \\
\frac{\udp_{yv}(t,y,v)}{\udp_{y}(t,y,v)}&=& \frac{\dqla_{yv}(t,y,v)}{\dqla_{y}(t,y,v)} + b(t)
\label{cond:U_D_yv_y} ;\\
\dqla_{v}(t,y,v) &=&0, \label{cond:D_v}
\end{eqnarray}%
where $\dqla$ is given by \eqref{eq:def_of_financial_derivative_on_Y}, $\udp$ is defined in \eqref{eq:def_expected_aux_utility}, $Y^{\ast}(t) = y$, $v(t) = v$. Then a candidate optimal for the terminal portfolio value is:
\begin{align}
    X^{x, \pi^\ast_c}(T) & =  D(Y^{y, \piu}(T), v(T)) \label{eq:optimal_terminal_wealth} \\
    \text{with} \,\,x &=  \mathbb{E}_{t,y,v}^{\mathbb{Q}(\lambda^{v})}\left[\exp\rBrackets{-r(T-t)} D(Y^{y, \piu}(T),v(T))\right] = \dqlatyv  \notag
\end{align}
and the value function and optimal investment strategy in \eqref{MainControlProb} at time $\tin$ are:
\begin{align}
    \mathcal{V}^{c}\left( t, x ,v\right) =  \EVtxv{P}{ U\left( X^{x, \pi^\ast_c}(T)\right) }   & =\EVtyv{P}{\overline{U}(D(Y^{y, \piu}(T),v(T)))}  = \udp(t,y,v) ;\label{eq:link_btw_utility_functons}\\
    \pic(t) &= \piu(t) \cdot y  \cdot \frac{\dqla_{y}(t,y,v)}{\dqla (t,y,v)}. \label{OptConsProp}
\end{align}

If $\rho = 0$, solely Conditions \eqref{cond:U_D_yy_y} and \eqref{cond:D_v} are required.
\end{theorem}
\begin{proof}
See  Appendix \ref{app:constrained_OP_solution_short}. 
\end{proof}
\bigskip

\textbf{Remarks to Theorem \ref{MainTheo}}
\begin{enumerate}
    \item We do not yet need to impose any condition on $\lambda^v$. Along with the parameters of the payof{}f function $D$, it is an important degree of freedom to ensure Conditions \eqref{cond:U_D_yy_y} -- \eqref{cond:D_v}, as we will see in the following corollaries. 
    \item Condition \eqref{cond:U_D_yy_y} is the same as in \cite{Kraft2013}. Moreover, in the absence of stochastic volatility, we recover their results for the Black-Scholes market. Recall that the relative-risk aversion (RRA) coefficient of $U(x)=x^\gamma / \gamma$ is $1 - \gamma$. Therefore, from an economic perspective, condition \eqref{cond:U_D_yy_y} means that the RRA coefficient of $\udp$, which is induced by $U$, is $1 - \gamma$, since $\dq$ can be interpreted as the value function of a risk-neutral decision maker with the RRA coefficient of $0$.  
    \item Condition \eqref{cond:U_D_yv_y} conveys a deterministic relation between the ratio of the Greeks vanna and delta for both the contingent claim value ($\dqla$) and the expected modified utility ($\udp$). This is similar to the deterministic relation on the ratio of the Greeks gamma and delta implied by Condition \eqref{cond:U_D_yy_y}.
    \item Condition \eqref{cond:D_v} means that the financial derivative with terminal payof{}f $D$ has to be vega-neutral at time $t$ and the value $v$ of the variance process (i.e., $\ddv D = \frac{\partial  D}{\partial \sqrt{v}}\frac{\partial \sqrt{v}}{\partial v} = \frac{\partial D}{\partial \sqrt{v}}\frac{1}{2\sqrt{v}}$). The complexity lies in crafting this payof{}f function $D$.
    \item If the optimal terminal wealth in the unconstrained problem with initial capital $x$ satisfies the VaR constraint, then it is obviously the optimal wealth in the VaR-constrained optimization problem.  In this case,  $D(y,v) = y$, $\lambda_{\varepsilon} = 0$ and:
    \begin{itemize}
        \item $\eqref{cond:U_D_yy_y} \iff \udp_{y}(t,y,v) = y^{\gamma - 1} G_A(t, v)$ for some function $G_A(t, v)$, which holds for $\mathcal{V}^{u}(t,y,v)$ from  \eqref{eq:unconstrained_value_function};
        \item $\eqref{cond:U_D_yv_y} \iff \udp_{y}(t,y,v) = \exp\rBrackets{b(t)v} G_B(t, y)$ for some function $G_B(t, y)$, which holds for $\mathcal{V}^{u}(t,y,v)$ from  \eqref{eq:unconstrained_value_function};
        \item \eqref{cond:D_v} holds;
        \item \eqref{eq:link_btw_utility_functons} becomes $\mathcal{V}^{c}\left( t, y ,v\right) = \udp(t,y,v) \stackrel{\lambda_{\varepsilon}= 0}{=} \mathcal{V}^{u}\left( t, y ,v\right) $;
        \item \eqref{OptConsProp}  $\pic(t) = y \rBrackets{\frac{\overline{\lambda }}{1-\gamma}+\frac{\sigma \rho }{1-\gamma }b(t)}\frac{1}{y}= \piu(t)$.
    \end{itemize}
\end{enumerate}

Next, we provide convenient sufficient conditions to facilitate the applications of Theorem \ref{MainTheo}.

\begin{lemma}[Sufficient condition for \eqref{cond:U_D_yy_y} and \eqref{cond:U_D_yv_y}]
\label{lem:sufficient_condition}
$\,$\\
Condition \eqref{cond:U_D_yy_y} is satisfied at time $t \in [0, T]$ given $Y^{\ast}(t) = y$ and $v(t) = v$, if there exists a function $H(t,v)$ such that the following sufficient condition holds:
\begin{equation}\label{eq:sufficient_condition_rho_zero}
\udp_{y}(t,y,v)=y^{\gamma -1} H(t,v) \dqla_{y}(t,y,v).
\tag{SC0}
\end{equation}

Both Condition \eqref{cond:U_D_yy_y} and Condition \eqref{cond:U_D_yv_y} are satisfied at time $t \in [0, T]$ given $Y^{\ast}(t) = y$ and $v(t) = v$, if \eqref{eq:sufficient_condition_rho_zero} holds with $H(t,v) = h(t) \exp \left(b(t)v\right)$ for some function $h(t)$, i.e.:
\begin{equation}\label{cond:SC_in_lemma}
\udp_{y}(t,y,v) = y^{\gamma -1} h(t) \exp \left(b(t)v\right) \dqla_{y}(t,y,v).
\tag{SC}
\end{equation}
\end{lemma}
\begin{proof}
See Appendix \ref{app:constrained_OP_solution_short}.
\end{proof}

\bigskip

In contrast to the sufficient condition in {\cite{Kraft2013}}, Condition \eqref{cond:SC_in_lemma} has an additional term $\exp \left(b(t)v\right)$. As we will see later, $h(t) = \exp(a(t))$ with $a(t)$ from Proposition \ref{prop:unconstrained_problem_solution}.

\section{Detailed application of Theorem \ref{MainTheo}} \label{sec:explicit_formulas}

In solving the VaR-constrained power utility problem in a complete Black-Scholes market, \cite{Kraft2013} use a contingent claim $D^{BS}(\cdot)$ with the following payof{}f as seen from time $\tin$:
\begin{align}
    X^\ast(T) & =  Y^{\ast}(T) + (K - Y^{\ast}(T)) \mathbbm{1}_{\{ k_{\varepsilon} < Y^{\ast}(T) \leq K\}}  \notag\\
    & =   Y^{\ast}(T) + \left( K - Y^{\ast}(T) \right) 1_{\left\{ Y^{\ast}(T) \leq K\right\} } \notag \\
    &\quad  - \left( k_{\varepsilon} - Y^{\ast}(T) \right) 1_{\left\{ Y^{\ast}(T) < k_{\varepsilon}\right\} }-\left( K-k_{\varepsilon}\right) 1_{\left\{ Y^{\ast}(T) <k_{\varepsilon}\right\} } =: D^{BS}(Y^{\ast}(T);k_{\varepsilon}, t) ,  \label{eq:D_conjecture_BS} 
\end{align}
where $0\leq k_{\varepsilon} \leq K$. The payof{}f \eqref{eq:D_conjecture_BS} is illustrated in Figure \ref{sfig:payoff_KS}. It consists of a long position of the optimal unconstrained wealth, a long put option, a short put option with a lower strike, and a binary put option.

The main result in \cite{Kraft2013} (Theorem 1) and their VaR application (i.e., Proposition 2) requires only one condition, their Equation (8), to ensure a successful writing of the constrained problem in terms of a derivative on the unconstrained optimal wealth (i.e., the matching of PDEs). This is the key observation behind their need for only one degree of freedom, the parameter $k_{\varepsilon}$ in their choice of contingent claim. However, their condition must be met at all times $t \in [0,T]$. It turns out that the condition leads to the same constant $k_{\varepsilon}$ at all times and for all state variables. In other words, the contingent claim is the same at all times: $D^{BS}(Y^{\ast}(T);k_{\varepsilon},t)=D^{BS}(Y^{\ast}(T))$. This further facilitates the calculation of $k_{\varepsilon}$ and helps simplify the contingent claim. This conclusion is also supported by a concavification argument, as explained by the authors.

\begin{figure}[!ht]
        \centering
        \begin{subfigure}{0.5\textwidth}
          \centering
          \includegraphics[width=\linewidth]{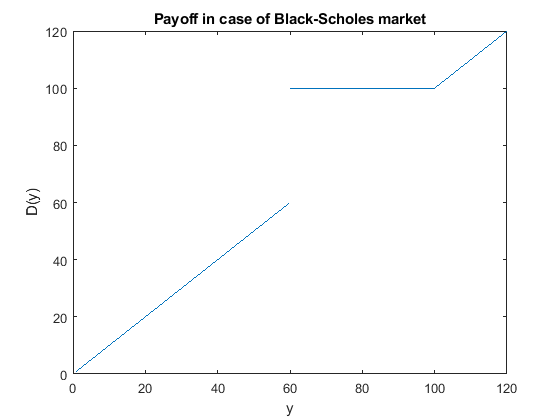}
          \caption{Example of Payoff \eqref{eq:D_conjecture_BS}, complete Black-Scholes \newline market}
          \label{sfig:payoff_KS}
        \end{subfigure}%
        \begin{subfigure}{0.5\textwidth}
          \centering
          \includegraphics[width=\linewidth]{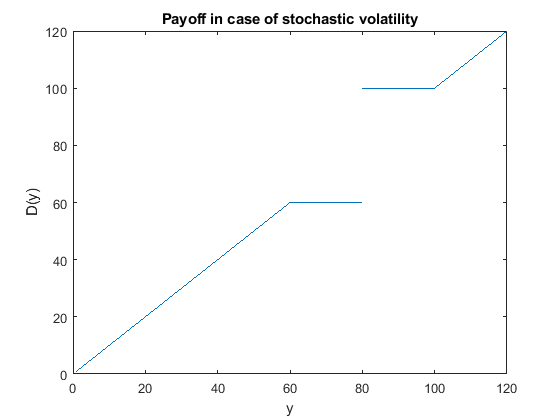}
          \caption{Example of Payoff \eqref{eq:D_conjecture_Heston}, incomplete stochastic \\ volatility market}
          \label{sfig:payoff_HM}
        \end{subfigure}
        \caption{\color{black} Comparison of payof{}f structures of to-be-conjectured $D$ in complete and incomplete markets}
        \label{fig:payoffs_KS_HM}
\end{figure}

In the presence of stochastic volatility, we have to ensure two conditions, namely Equation \eqref{cond:SC_in_lemma} (i.e., a sufficient condition for Equations \eqref{cond:D_v} and \eqref{cond:U_D_yv_y}), and Equation \eqref{cond:U_D_yy_y}. Therefore, the payof{}f structure requires a second degree of freedom, while \eqref{eq:D_conjecture_BS} is no longer viable. Our choice of the payof{}f a contingent claim at a given time $\tin$ is presented next:
\begin{align}
X^\ast(T) & =   Y^{\ast}(T)+\left( K-Y^{\ast}(T)\right) 1_{\left\{ k_{\varepsilon} \leq Y^{\ast}(T)\leq K\right\} }-\left( Y^{\ast}(T)-k_{v}\right) 1_{\left\{ k_{v}\leq Y^{\ast}(T)<k_{\varepsilon}\right\} } \notag\\
& =  Y^{\ast}(T) + \left( K - Y^{\ast}(T) \right) 1_{\left\{ Y^{\ast}(T) \leq K\right\} } \notag \\
&\quad -\left( k_{v} - Y^{\ast}(T) \right) 1_{\left\{ Y^{\ast}(T) < k_{v}\right\} }-\left( K-k_{v}\right) 1_{\left\{ Y^{\ast}(T) <k_{\varepsilon}\right\} } =: \widehat{D}\left( Y^{\ast}(T); k_{\varepsilon},k_{v}, t \right), \label{eq:D_conjecture_Heston} 
\end{align}%
with $0 \leq k_{v}\leq k_{\varepsilon}\leq K$. This payof{}f has an additional degree of freedom $k_v$, which is crafted to ensure Condition \eqref{cond:D_v} at the corresponding time point $\tin$, while $k_{\varepsilon}$ remains to ensure the sufficient condition \eqref{cond:SC_in_lemma}. {\color{black} In general, for $t_1 \in [0,T]$ and $t_2 \in [0,T]$ such that $t_1 \neq t_2$, the payof{}fs $\widehat{D}\left( Y^{\ast}(T); k_{\varepsilon},k_{v}, t_1 \right)$ and $\widehat{D}\left( Y^{\ast}(T); k_{\varepsilon},k_{v}, t_2 \right)$ are different. To emphasize this time-dependence while making the notation more concise, we write $\widehat{D}\left( Y^{\ast}(T); k_{\varepsilon, t},k_{v,t} \right) := \widehat{D}\left( Y^{\ast}(T); k_{\varepsilon},k_{v}, t \right)$, where $k_{\varepsilon, t}$ and $k_{v,t}$ are the related parameters of the contingent claim at time $\tin$.} Since $\mathbb{P}\left(X^{x_0, \pic}(T)<K\right) =\varepsilon \Leftrightarrow \mathbb{P}\left( Y^{y_0, \piu}(T) < k_{\varepsilon,0}\right)
=\varepsilon $, $\widehat{D}(\cdot; k_{\varepsilon,t},k_{v,t})$ has enough flexibility to ensure Condition \eqref{cond:D_v}, which can be seen as vega neutrality of the financial derivative. This payof{}f is illustrated in Figure \ref{sfig:payoff_HM}.

We show in the proof of Corollary \ref{cor:heston_var_rho_nonzero_solution} that $(k_{\varepsilon,t}, k_{v,t})$ must be modified at every time $\tin$ to ensure Condition \eqref{cond:SC_in_lemma}. This means that we deal with a state-dependent payof{}f. This is a financial derivative that can be hedged with a self-financing investment strategy and that has changing payof{}fs. In other words, we must use an infinite sequence, a continuum, of contingent claims to match the two conditions at all times. A similar development is needed in \cite{Kraft2013} to tackle the so-called Expected Shortfall constraint in their Section 3.3, although the authors do not dwell on it.

Each contingent claim in the above-mentioned continuum of claims has the underlying asset $Y^\ast$ and a payof{}f $\widehat{D}(\cdot;k_{\varepsilon, t}, k_{v, t})$ that has the structure of \eqref{eq:D_conjecture_Heston}. Therefore, in order to apply our main theorem, we need to show that such a continuum of contingent claims (i.e., with payof{}fs $\widehat{D}\left( \cdot; k_{\varepsilon,t }, k_{v, t} \right)$) can be modeled as a single contingent claim with a non-state-dependent payof{}f, denoted by $D(Y^{\ast}(T), v(T))$, as required by our Theorem \ref{MainTheo}. This connection is presented in the next proposition.

\begin{proposition}\label{prop:equivalence_btw_sequence_and_single_D}
Let $\widehat{G}(Y^\ast(T); k(t, Y^\ast(t), v(t)))$ be some payof{}f function of a contingent claim with a state-dependent strike denoted by $k(t, Y^\ast(t), v(t))$. Let $\widehat{D}(t, Y^\ast(t), v(t); k(t, Y^\ast(t), v(t)))$ be the price process corresponding to this payof{}f:
\begin{equation*}
    \widehat{D}(t,Y^\ast(t),v(t);k\left( t,Y^\ast(t),v(t)\right) )=E_{t,y,v}^{\mathbb{Q}}\left[ \widehat{G}\left(
Y^\ast(T);k\left( t,Y^\ast(t),v(t)\right) \right) \right].
\end{equation*}
{\color{black} Assume that $\widehat{D}$ is vega neutral at each $\tin$ and that there is a self-financing investment strategy that replicates the price process of this financial derivative.} Then $\widehat{D}$ can be interpreted as the price process of a single contingent claim with price $\widetilde{D}(t, Y^\ast(t), v(t), \widetilde{k}(t))$ and a non-state-dependent payof{}f $\widetilde{G}(Y^\ast(T), \widetilde{k}(T)) \equiv D(Y^\ast(T), v(T))$, such that for all $\tin$:
\begin{equation*}
    \widetilde{D}(t, Y^\ast(t), v(t), \widetilde{k}(t)) =  \widehat{D}(t,Y^\ast(t),v(t);k\left( t,Y^\ast(t),v(t)\right)),
\end{equation*}
where $ \widetilde{k}$ can be interpreted as an artificial asset fully explained by $(Y^* , v)$.
\end{proposition}
\begin{proof}
    See Appendix \ref{app:constrained_OP_solution_short}.
\end{proof}

Next, we apply our main theorem, using the sequence of financial derivatives with payof{}f \eqref{eq:D_conjecture_Heston} and Proposition \ref{prop:equivalence_btw_sequence_and_single_D}, and derive a more explicit representation of the solution to \eqref{MainControlProb}.

\begin{corollary}[Solution to \eqref{MainControlProb}] \label{cor:heston_var_rho_nonzero_solution}
{\color{black} Assume that Condition \eqref{KraftCondition} holds and that the VaR constraint is feasible and binding in \eqref{MainControlProb}.
Set $\lambda^v(t) = - \sigma \sqrt{1 - \rho^2} b(t)$, and let $D(\cdot,\cdot)$ be the payof{}f derived -- via Proposition \ref{prop:equivalence_btw_sequence_and_single_D} -- from a continuum of payof{}fs, denoted by $\widehat{D}(\cdot; k_{\varepsilon, t},k_{v, t})$, of the type given by Equation \eqref{eq:D_conjecture_Heston}.  Assume that $\widehat{D}(\cdot; k_{\varepsilon, t}, k_{v, t})$ is such that its degrees of freedom $\left(y_t, k_{v,t}, k_{\varepsilon, t}\right)$ satisfy the system of non-linear equations
at time $t = 0$:
\begin{equation}\label{eq:SNLE_t_zero} \tag{NLS0}
    \left\{
    \begin{aligned}
        &h_{B}(y_0, k_{v, 0}, k_{\varepsilon, 0}) := \widehat{D}(0,y_0,v_0;k_{v, 0}, k_{\varepsilon, 0}) = x_0;\\ 
        &h_{VN}(y_0, k_{v, 0}, k_{\varepsilon, 0}) := \widehat{D}_{v}(0,y_0,v_0;k_{v, 0}, k_{\varepsilon, 0}) = 0;\\ 
        &h_{VaR}(y_0, k_{v, 0}, k_{\varepsilon, 0}) := \mathbb{P}\left( Y^{\ast}(T) <  k_{\varepsilon, 0}|Y^{\ast}(0) = y_0, v(0) = v_0\right) = \varepsilon;
    \end{aligned}
    \right.
\end{equation}
and the system of non-linear equations for each time $t \in (0, T]$:
\begin{equation}\label{eq:SNLE_rho_zero}  \tag{NLS}
    \left\{
    \begin{aligned}
        &h_{B}(y_t, k_{v, t}, k_{\varepsilon, t}) = x_t;\\
        &h_{VN}(y_t, k_{v, t}, k_{\varepsilon, t})  = 0;\\
        &y_t^{\gamma - 1}\exp \left(a(t) + b(t)v_t - r(T - t) \right) \left( K - k_{v, t}\right) \frac{f^{\mathbb{Q}}_{Z^\ast(T)}(\ln k_{\varepsilon, t})}{f^{\mathbb{P}}_{Z^\ast(T)}(\ln k_{\varepsilon, t})} -  \frac{K^{\gamma} - k_{v,t}^{\gamma }}{\gamma } = \lambda^\ast_{\varepsilon};
    \end{aligned}
    \right.
\end{equation}
for the Lagrange multiplier given by:
\begin{equation}\label{eq:lambda_e_for_rho_any}
    \lambda _{\varepsilon}^\ast = y_0^{\gamma - 1}\exp \left(a(0) + b(0)v_0 - rT\right) \left(K-k_{v,0}\right) \frac{f^{\mathbb{Q}}_{Z^\ast(T)}(\ln k_{\varepsilon, 0})}{f^{\mathbb{P}}_{Z^\ast(T)}(\ln k_{\varepsilon, 0})} -  \frac{K^{\gamma} - k_{v,0}^{\gamma }}{\gamma }.
\end{equation}

Here $x_t$ is the realized value of $X(t)$ at $\tin$, $v_t$ is the realized value of $v(t)$ at $\tin$, $f^{\mathbb{M}}_{Z^\ast(T)}(\cdot)$ is the conditional density function of $Z^\ast(T):=\ln(Y^{\ast}(T))$ under the measure $\mathbb{M} \in \{\mathbb{P},\, \mathbb{Q}\}$ given $Y^{\ast}(t) = y_t$ and $v(t) = v_t$. Then, the optimal terminal portfolio value is given by \eqref{eq:optimal_terminal_wealth}, the value function is given by \eqref{eq:link_btw_utility_functons}, and the solution to \eqref{MainControlProb} is given by \eqref{OptConsProp}.}
\end{corollary}
\begin{proof}
See  Appendix \ref{app:constrained_OP_solution_short}.
\end{proof}

\textbf{ Remarks to Corollary \ref{cor:heston_var_rho_nonzero_solution}}
\begin{enumerate}
    \item The tuple $\left(y_t,k_{v,t},k_{\varepsilon,t}\right)$ needs to be updated at every $\tin$ in order to produce the right strategy. In a Black-Scholes market, \cite{Kraft2013} does not need to update $k_{\varepsilon}$ and $y$.
    \item The conditional density functions $f^{\mathbb{M}}_{Z^\ast(T)}(\cdot)$ can be calculated using the inversion of the characteristic functions of $Z^\ast(T)$ provided in Proposition \ref{prop:characteristic_fct_log_Y}.
    \item The investor's value function $\mathcal{V}^{c}$, the price of the financial derivative $\dqla$ and its Greeks  $\dqla_{y}$ as well as $\dqla_{v}$ can be numerically computed using the Carr-Madan approach to pricing options. We provide the corresponding formulas for $h_{B}(y_t, k_{v,t},\,k_{\varepsilon,t})$, $h_{VaR}(y_t, k_{v,t}, k_{\varepsilon,t})$, $h_{VN}(y_t, k_{v,t}, k_{\varepsilon,t})$ in Appendix \ref{app:explicit_formulas}.
    \item $\lambda^v$ is the same as the worst-case shadow price in the unconstrained problem that leads to $Y^\ast$. As we represent $X^\ast$ as a synthetic financial derivative on $Y^\ast$, it is not a surprise that $\lambda^v$ appears in the optimal constrained solution too.
\end{enumerate}

\section{Numerical studies}\label{sec:numerical_studies}
In this section, first, we explain how we choose the model parameters and provide details on the solution procedure for the system of non-linear equations \eqref{eq:SNLE_rho_zero}.  Second, we conduct sensitivity analysis of the optimal constrained and unconstrained investment strategy w.r.t. key parameters such as the relative risk aversion (RRA) coefficient and the investment horizon.

\subsection{Model parameterization and numerical procedure}
{\color{black}
We choose the parameters of the Heston model as in Table 4 in \cite{Escobar2016}, the row corresponding to the average case of the table mentioned. There, the authors provide model parameterization under an EMM. In particular, we set: $\kappa = 3.6129$, $\theta = 0.0291$, $\sigma = 0.3$, $\rho = -0.4$, $v_0 = 0.03$. $\lambar = 1$, $r = 3\%$. Under these parameters, $\lambda^{v} = 0.0238$, leading to $\tilde{\kappa} = 3.5$, $\tilde{\theta} = 0.03$. We set $\gamma = -2$, which corresponds to the RRA coefficient of $3$, as also considered in \cite{Chen2018b}.  We assume that the investor's time horizon is $T = 3$, his/her initial wealth is $x_0 = 100$, and the VaR constraint is specified by $K = 100$ and $\varepsilon = 1\%$ in the base case.

Solving the system of non-linear equations \ref{eq:SNLE_rho_zero} requires numerical methods. First, we need to find $A^{\mathbb{M}}$, $B^{\mathbb{M}}$, $\mathbb{M} \in \{\mathbb{P}, \mathbb{Q}\}$ appearing in the characteristic functions of $Z^\ast(T)$. As we mentioned in Remark 2 to Propostion \ref{prop:characteristic_fct_log_Y}, the ODEs for $B^{\mathbb{M}}$ have time-dependent complex-valued coefficients and are of Riccati type. To compute the solutions to those equations, we use a Matlab function \texttt{ode45} that is based on an explicit Runge-Kutta method. We chose a time grid of $10001$ points, which corresponds to a time discretization step of $3 \cdot  10^{-4}$.  Second, we compute the LHS of \ref{eq:SNLE_rho_zero} using the Carr-Madan approach, see Appendix \ref{app:explicit_formulas} for explicit formulas.
Regarding dampening factors in this approach, we use $2$ for plain vanilla put options (the  $2$-nd and $3$-rd terms in the financial derivative $D$) and $0.5$ for a digital put option (the $4$-th term in $D$). Finally,  the solution of \ref{eq:SNLE_rho_zero} is computed by minimizing the sum of squared absolute errors, which is done with the help of a Matlab function \texttt{fmincon} with sequential quadratic programming as the underlying non-linear optimization algorithm.
}

\subsection{Numerical results}
{\color{black}
In this subsection, we first compute and interpret the optimal constrained investment strategy in the base case of $\varepsilon = 1\%$. Second, we conduct a sensitivity analysis of $\pic(0)$ and the optimal parameters of the synthetic derivative $D$ are with respect to $\varepsilon$. Third, we examine the impact of the RRA coefficient and the investment horizon on the optimal constrained investment strategy. 

In the base case, the optimal unconstrained investment strategy at time $t = 0$ is equal to $33.71\%$. The optimal constrained investment strategy at time $t = 0$ equals $31.72\%$.  The optimal terminal wealth in the constrained problem equals a financial derivative on the optimal unconstrained wealth with the following parameters: $y^\ast_0 = 99.5$, $k^\ast_{v,0} = 68.55$, $k^\ast_{\varepsilon,0} = 87.96$. It means that the optimal terminal wealth in the constrained optimization problem given the starting value $x_0 = 100$ is equal to a financial derivative consisting of:
\begin{enumerate}
	\item a long position in the optimal unconstrained wealth $Y^*(T)$ with $Y^*(0) = y^\ast_0  = 99.5$;
	\item a long position in one put option on the optimal unconstrained wealth $Y^*(T)$ and with strike $K = 100$;
	\item a short position in one put option on  $Y^*(T)$ and with strike $k^\ast_{v,0} = 68.55$;
	\item a short position in $K - k^\ast_{v,0} = 31.45$ digital put options on the optimal unconstrained wealth $Y^*(T)$ and with strike $k^\ast_{\varepsilon,0} = 87.96$.
\end{enumerate}

Next, we analyze the impact of $\varepsilon$.  Denote by $\varepsilon_u := \mathbb{P}\rBrackets{Y^{x, \piu}(T) < K} \approx 12\%$ the probability that the optimal terminal unconstrained portfolio value is below $K$. Consider Figure \ref{fig:impact_of_epsilon_VaR_eq}, which consists of two subfigures.
In Subfigure \ref{sfig:piStar_vs_eps_VaR_eq}, we see that for increasing $\varepsilon$ the optimal constrained investment strategy becomes closer to the unconstrained one. This is intuitive, since the closer $\varepsilon$ is to $\varepsilon_u$ , the more the optimal constrained investment strategy should resemble the optimal unconstrained one. As Subfigure \ref{sfig:DQ_Params_vs_eps_VaR_eq} indicates, the larger $\varepsilon < \varepsilon_u$, the larger the optimal initial capital of the underlying of the financial derivative $D$ (the optimal unconstrained portfolio) and the higher the thresholds $k_{\varepsilon}$ and $k_{v}$. This is consistent with our previous finding that $\varepsilon$ closer to $\varepsilon_u$ leads to the optimal constrained investment becoming the unconstrained one. The same holds for the optimal terminal wealth, since increasing $k_{\varepsilon}$ and $k_{v}$ mean that the optimal payof{}f of the derivative $D$ is closer to the identity function (cf. Remark to Theorem \ref{MainTheo}).}

\begin{figure}[!ht]
        \centering
        \begin{subfigure}{0.5\textwidth}
          \centering
          \includegraphics[width=\linewidth]{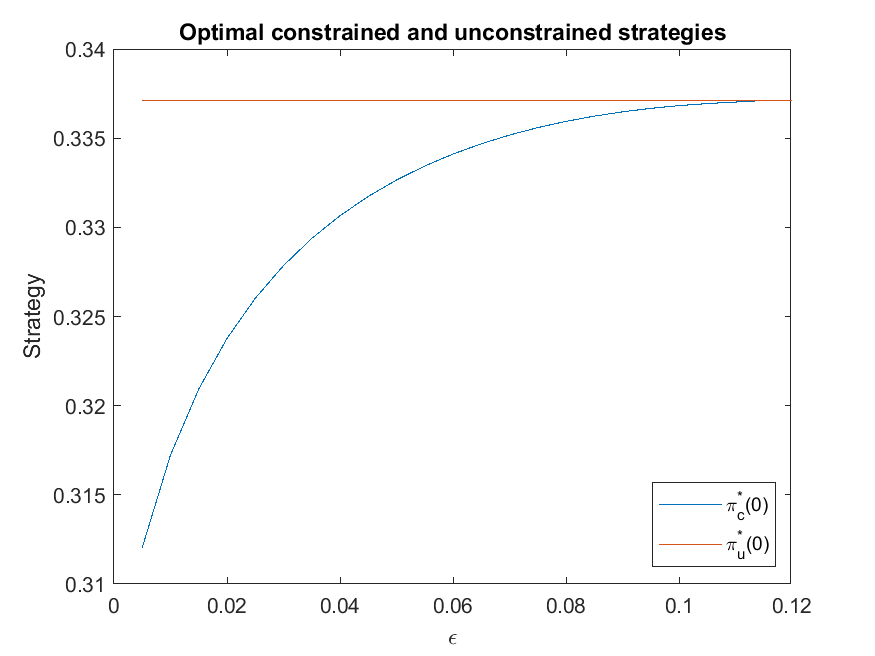}
          \caption{$\pic(0)$ vs $\varepsilon$}
          \label{sfig:piStar_vs_eps_VaR_eq}
        \end{subfigure}%
        \begin{subfigure}{0.5\textwidth}
          \centering
          \includegraphics[width=\linewidth]{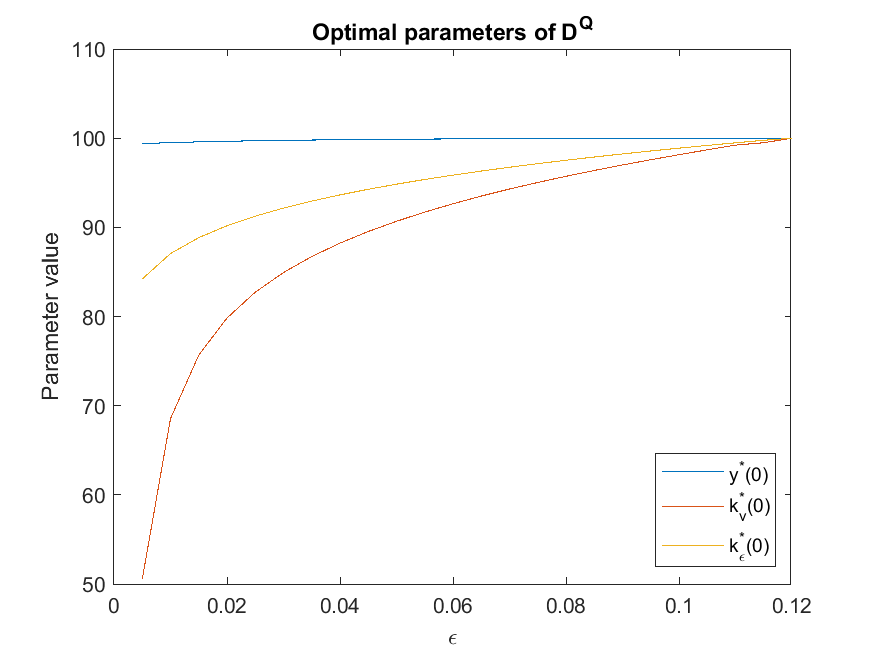}
          \caption{$\dq$ parameters vs $\varepsilon$}
          \label{sfig:DQ_Params_vs_eps_VaR_eq}
        \end{subfigure}
        \caption{The impact of VaR-probability threshold on the solution to Problem \eqref{MainControlProb}} 
        \label{fig:impact_of_epsilon_VaR_eq}
\end{figure}

{\color{black} Next, we investigate in Figure \ref{fig:impact_of_RRA_and_T} the influence of the investor's risk aversion and time horizon on the optimal investment strategies. In Subfigure \ref{sfig:piStar_vs_RRA}, we see that both constrained and unconstrained investment strategies are decreasing in the RRA coefficient $1 - \gamma$. The difference between these strategies shrinks as the investor becomes more risk averse, going from a relative difference of approximately $14\%$ ($\pic(0)=44\%$ and $\piu(0)=50\%$ for an RRA of $2$) to a relative difference of $0\%$ ($\pic(0)=\piu(0)=20\%$ for an RRA of $2$). 

In Subfigure  \ref{sfig:piStar_vs_T}, we see that the optimal constrained strategy is increasing in the investment horizon and is approaching the unconstrained one. A constrained decision maker with a $1$-year investment horizon will allocate approximately $28\%$ of his/her money to the risky asset, while an unconstrained investor would allocate almost $33.7\%$ to the risky asset, which means a relative difference of approximately $20\%$. However, over a longer period of time, e.g., $5$ years, the investor allocates more money to the risky asset while still ensuring the desired VaR constraint, i.e., he/she invests $33\%$ of the money in $S_1$. For $T = 10$, the probability that the optimal terminal unconstrained wealth is smaller than $K = 100$ is around $1\%$, which is why the optimal unconstrained strategy and the optimal constrained strategy for $\varepsilon = 1\%$ almost coincide.}

\begin{figure}[!ht]
        \centering
        \begin{subfigure}{0.5\textwidth}
          \centering
          \includegraphics[width=\linewidth]{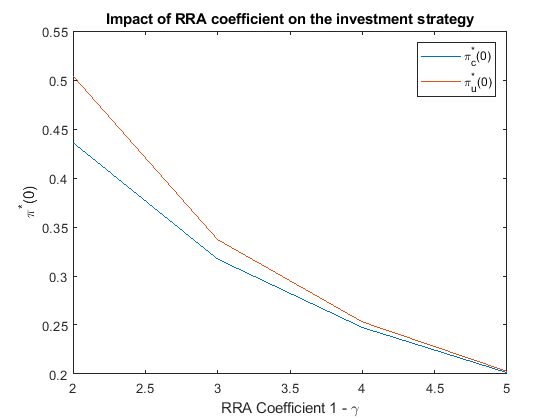}
          \caption{$\pic(0)$ and $\piu(0)$ vs $RRA = 1 - \gamma$}
          \label{sfig:piStar_vs_RRA}
        \end{subfigure}%
        \begin{subfigure}{0.5\textwidth}
          \centering
          \includegraphics[width=\linewidth]{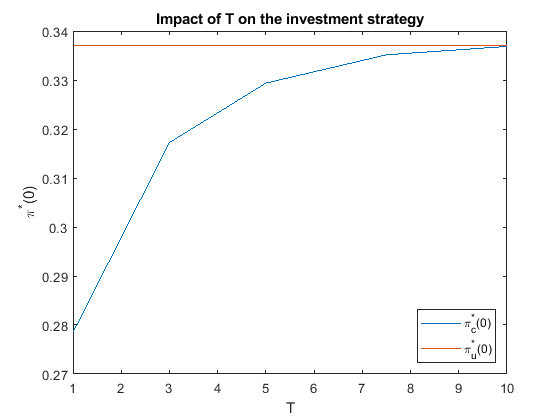}
          \caption{$\pic(0)$ and $\piu(0)$  vs $T$}
          \label{sfig:piStar_vs_T}
        \end{subfigure}
        \caption{The impact of risk aversion  and  time horizon on the optimal investment strategies} 
        \label{fig:impact_of_RRA_and_T}
\end{figure}

We also studied the impact of $\rho$, as well as the simultaneous decrease of the parameters $\sigma$ and $\kappa$ on the optimal constrained investment strategy. These last two parameters can be considered as a measure for the magnitude of the market incompleteness, since their joint decrease would eliminate the stochastic volatility. The impact of the parameters is very similar between constrained and unconstrained solutions; e.g., correlation decreases allocations almost with the same slope, while the joint decrease of $\kappa,\sigma$ causes similar decreasing effects in allocation (sightly more pronounced on constrained allocations); this is why we did not report the figures here. In Appendix \ref{app:numerical_studies_turbulent_markets} of supplementary materials, we provide such plots under the parameterization of the model that corresponds to a more turbulent market and a less risk-averse investor. In that case, the described sensitivity effects are more pronounced, and the differences between constrained and unconstrained strategies are larger.

\section{Conclusion}\label{sec:consclusion}
In this paper, we solve a VaR-constrained utility maximization problem under the Heston model via dynamic programming. Our methodology extends the methodology of \linebreak\cite{Kraft2013} from a complete market to an incomplete market; hence, it opens the door to studying other types of terminal wealth constraints, e.g., expected shortfall, or incomplete market setting, e.g. stochastic interest rates, stochastic price of risk.

The key idea is to link the solution of the constrained optimization problem with the solution to the unconstrained one via a synthetic derivative. For the VaR-constrained problem, this derivative is based on plain-vanilla put options and a digital put option, whose strikes must be determined numerically. In numerical studies, we find that, for investors with low risk aversion and short investment horizons, the relative difference between optimal constrained and unconstrained allocations could be substantial, e.g., $20\%$ for normal market parameters.
 
\bibliographystyle{apalike}
\bibliography{HestonVar_article_main}

\appendix

\section{Results on unconstrained problem}\label{app:unconstrained_problem}
\begin{proposition} \label{prop:unconstrained_problem_solution}
Assume that Condition \eqref{KraftCondition} is satisfied. Then the optimal investment strategy
for \eqref{PowerUnconstProb} is given by:
\begin{equation}
\piu(t)=-\frac{\overline{\lambda }\mathcal{V}_{y}^{u}+\sigma \rho \mathcal{V}_{yv}^{u}}{x\mathcal{V}_{yy}^{u}}=%
\frac{\overline{\lambda }}{1-\gamma }+\frac{\sigma \rho }{1-\gamma }b(t)
\label{eq:pi_star_unconstrained}
\end{equation}%
with $k_{0}=(\gamma \overline{\lambda }^{2})/(1-\gamma )$, $k_{1}=\kappa
-(\gamma \overline{\lambda }\sigma \rho )/(1-\gamma )$, $k_{2}=\sigma
^{2}+(\gamma \sigma ^{2}\rho ^{2})/(1-\gamma )$, $k_{3}=\sqrt{%
k_{1}^{2}-k_{0}k_{2}}$ and%
\begin{equation}
b(t)=k_{0}\frac{\exp\rBrackets{k_{3}(T-t)}-1}{\exp\rBrackets{k_{3}(T-t)}(k_{1}+k_{3})-k_{1}+k_{3}}.
\label{bsolution}
\end{equation}%
The value function is given by
\begin{equation}\label{eq:unconstrained_value_function}
\mathcal{V}^{u}(t,y,v) =\frac{y^{\gamma }}{\gamma }\exp (a(t)+b(t)v)
\end{equation}%
where $b(t)$ is defined by \eqref{bsolution} and
\begin{equation*}
    a(t) = \gamma r(T-t) + \frac{2\theta \kappa }{k_{2}} \ln \left( \frac{2k_{3}\exp\rBrackets{\frac{1%
}{2}(k_{1}+k_{3})(T-t)}}{2k_{3}+\left( k_{1}+k_{3}\right) (\exp\rBrackets{k_{3}(T-t)}-1)}%
\right).
\end{equation*}
The optimal wealth $Y^{\ast}$ has the following dynamics under $\mathbb{P}$:
\begin{equation}\label{eq:Y_Star_SDE_under_P}
\begin{aligned}
\!\mathrm{d}Y^{\ast}(t)=Y^{\ast}(t)\biggl[ &\biggl( r+\left( \frac{\overline{\lambda }}{%
1-\gamma }+\frac{\sigma \rho b(t)}{1-\gamma }\right) \overline{\lambda }%
v(t)\biggr) \!\mathrm{d}t\\
&+\left( \frac{\overline{\lambda }}{1-\gamma }+%
\frac{\sigma \rho b(t)}{1-\gamma }\right) \sqrt{v(t)}\,\mathrm{d}W^{\mathbb{P}}_1(t)\biggr],\,\,\,Y^{\ast}(0)=y>0.
\end{aligned}
\end{equation}
\end{proposition}
\label{app:unconstrained_OP_solution}

\begin{proof}[Proof of Proposition \ref{prop:unconstrained_problem_solution}]$\,$\\
For the readability of this proof, we denote  $\mathcal{V}:=\vfu$. We face a two-dimensional control problem \eqref{PowerUnconstProb} with state process $(Y,v)$ and consider the HJB equation:%
\begin{equation*}
0=\mathcal{V}_{t}+\frac{1}{2}\sigma ^{2}v\mathcal{V}_{vv}+\kappa (\theta -v)\mathcal{V}_{v}+\underset{\pi }%
{\max }\left\{ \underbrace{y(r+\pi \overline{\lambda }v)\mathcal{V}_{y}+\frac{1}{2}\pi
^{2}y^{2}v\mathcal{V}_{yy}+\pi y\sigma v\rho \mathcal{V}_{yv}}_{g(\pi )}\right\}
\end{equation*}%
and boundary condition $\mathcal{V}(T,y,v)=\frac{y^{\gamma }}{\gamma }$. Eliminating $\max$ results in a first-order condition for $\pi $:
\begin{equation}
\piu=-\frac{y\overline{\lambda }v\mathcal{V}_{y}+y\sigma v\rho \mathcal{V}_{yv}}{%
y^{2}v\mathcal{V}_{yy}}=-\frac{\overline{\lambda }\mathcal{V}_{y}+\sigma \rho \mathcal{V}_{yv}}{y\mathcal{V}_{yy}}=-%
\overline{\lambda }\frac{\mathcal{V}_{y}}{y\mathcal{V}_{yy}}-\sigma \rho \frac{\mathcal{V}_{yv}}{y\mathcal{V}_{yy}}
\end{equation}%
under the assumption that $\mathcal{V}_{yy} < 0$. Substituting the expression for $\piu$ back into the HJB equation leads to the following non-linear PDE for the value function.%
\begin{eqnarray}
0 &=&\mathcal{V}_{t}+\frac{1}{2}\sigma ^{2}v\mathcal{V}_{vv}+\kappa (\theta -v)\mathcal{V}_{v}+yr\mathcal{V}_{y}-y%
\frac{\overline{\lambda }\mathcal{V}_{y}+\sigma \rho \mathcal{V}_{yv}}{y\mathcal{V}_{yy}}\overline{\lambda
}v\mathcal{V}_{y} \notag \\
&&+\frac{1}{2}\frac{(\overline{\lambda }\mathcal{V}_{y}+\sigma \rho \mathcal{V}_{yv})^{2}}{%
y^{2}\mathcal{V}_{yy}^{2}}y^{2}v\mathcal{V}_{yy}-\frac{\overline{\lambda }\mathcal{V}_{y}+\sigma \rho
\mathcal{V}_{yv}}{y\mathcal{V}_{yy}}y\sigma v\rho \mathcal{V}_{yv} \notag \\
 &=& \mathcal{V}_{t} + \kappa \theta \mathcal{V}_{v} + y r \mathcal{V}_{y} + v \rBrackets{\frac{1}{2}\sigma^{2} \mathcal{V}_{vv} - \kappa \mathcal{V}_{v} - \frac{1}{2}\frac{\rBrackets{\lambar \mathcal{V}_{y} + \sigma \rho \mathcal{V}_{yv}}^{2}}{\mathcal{V}_{yy}}}. \label{eq:vu_pde}
\end{eqnarray}%
To find the solution, we use the separation ansatz
\begin{equation*}
\mathcal{V}(t,y,v)=\frac{y^{\gamma }}{\gamma }h(t,v),\,\,h(T,v)=1.
\end{equation*}%
In this case, $\piu(t) = \frac{\lambar}{1 - \gamma} + \frac{\sigma \rho }{1 - \gamma}\frac{h_v}{h} $. We substitute the ansatz into the HJB equation and conclude that:
\begin{equation}\label{eq:PDE_h}
0=h_{t}+\kappa \theta h_{v}+\gamma rh+v\left( \frac{1}{2}\sigma
^{2}h_{vv}-\kappa h_{v}+\frac{1}{2}\frac{\gamma (\overline{\lambda }h+\sigma
\rho h_{v})^{2}}{(1-\gamma )h}\right).
\end{equation}%
The structure implies that $h(t,v)$ is exponentially affine:
\begin{equation*}
h(t,v)=\exp (a(\tau (t))+b(\tau (t))v) =: h,
\end{equation*}%
with time horizon $\tau (t)=T-t$ and, therefore, using boundary condition $h(T,z)=1\quad \forall z$, we get:
\begin{equation*}
    a(0)=a(\tau (T))=0,b(0)=b(\tau (T))=0.
\end{equation*}
Using this structure of $h(t,v)$ and rearranging to emphasize the linearity in $v$, we obtain the following:
\begin{align*}
0=&-a^{\prime }(\tau)h+b(\tau)\kappa \theta h+\gamma rh+v\Bigr[ -b^{\prime }(\tau)h+b^{2}(\tau)\left(
\frac{1}{2}\sigma ^{2}h+\frac{\gamma \sigma ^{2}\rho ^{2}h}{2(1-\gamma )}%
\right) \\
&+b(\tau)\left( -\kappa h+\frac{\gamma \overline{\lambda }\sigma \rho h}{%
1-\gamma }\right) +\frac{\gamma \overline{\lambda }^{2}h}{2(1-\gamma )}%
\biggr]  
\end{align*}%
Cancelling $h$ out leads to Riccati equations for $a$ and $b$:
\begin{align}
a^{\prime }(\tau )& =\kappa \theta b(\tau )+\gamma r \label{eq:PDE_a_tau}; \\
b^{\prime }(\tau )& =\frac{1}{2}\underbrace{\left( \sigma ^{2}+\frac{\gamma \sigma
^{2}\rho ^{2}}{1-\gamma }\right) }_{k_{2}}b^{2}(\tau )-\underbrace{\left( \kappa -%
\frac{\gamma \overline{\lambda }\sigma \rho }{1-\gamma }\right) }_{k_{1}}b(\tau)+%
\frac{1}{2}\underbrace{\frac{\gamma \overline{\lambda }^{2}}{1-\gamma }}%
_{k_{0}}  =\frac{1}{2}k_{2}b(\tau )^{2}-k_{1}b(\tau )+\frac{1}{2}%
k_{0}  \label{eq:PDE_B(t)au};
\end{align}%
and boundary conditions $a(0)=0,b(0)=0$ with constants $k_{0},k_{1},k_{2}$
that have to satisfy $k_{1}^{2}-k_{0}k_{2}>0$. Then according to \cite{Kraft2005,  Kallsen2010} the solution is given by:
\begin{align}
a(\tau )=& \gamma r\tau +\frac{2\theta \kappa }{k_{2}}\ln \left( \frac{%
2k_{3}\exp\rBrackets{\frac{1}{2}(k_{1}+k_{3})\tau }}{2k_{3}+\left( k_{1}+k_{3}\right)
(\exp\rBrackets{k_{3}\tau }-1)}\right) \label{eq:a_tau_explicit}\\
b(\tau )& =k_{0}\frac{\exp\rBrackets{k_{3}\tau }-1}{\exp\rBrackets{k_{3}\tau
}(k_{1}+k_{3})-k_{1}+k_{3}}~.\label{eq:B(t)au_explicit}
\end{align}%
with $k_{3}=\sqrt{k_{1}^{2}-k_{0}k_{2}}$~. For the system to be
well-defined, we have to check whether our constants fulfill $%
k_{1}^{2}-k_{0}k_{2}>0$. Therefore, we formulate the following requirement
on the parameters:
\begin{align*}
k_{1}^{2}-k_{0}k_{2}&=\kappa ^{2}-2\kappa \frac{\gamma }{1-\gamma }\overline{%
\lambda }\sigma \rho +\frac{\gamma ^{2}}{(1-\gamma )^{2}}\overline{\lambda }%
^{2}\sigma ^{2}\rho ^{2}-\frac{\gamma }{1-\gamma }\overline{\lambda }%
^{2}\sigma ^{2}-\frac{\gamma ^{2}}{(1-\gamma )^{2}}\overline{\lambda }%
^{2}\sigma ^{2}\rho ^{2}\\
& =\kappa ^{2}-\frac{\gamma }{1-\gamma }\overline{%
\lambda }\sigma (2\kappa \rho +\overline{\lambda }\sigma )>0 \Leftrightarrow \frac{\gamma }{1-\gamma }\overline{%
\lambda } \left(\frac{\kappa \rho}{\sigma} +\frac{\overline{\lambda }}{2} \right) <\frac{\kappa ^{2}}{2 \sigma^2 }
\end{align*}%
which is exactly what \cite{Kraft2005} requires in their Equation (26). Note that the ansatz satisfies the assumption $\mathcal{V}_{yy} < 0$ as for $\gamma < 1$ we have $\underbrace{(\gamma - 1)}_{<0} \underbrace{y^{\gamma - 2} h(t, v)}_{>0} < 0.$
\end{proof}

\bigskip

\begin{proposition}\label{prop:characteristic_fct_log_Y}
The logarithm of the unconstrained optimal wealth has characteristic functions of the form:%
\begin{align*}
\phi^{Z^\ast(T),\mathbb{M}} (u;t,z,v)&=\mathbb{E}_{t,z,v}^{\mathbb{M}}\left[ \exp (iuZ^\ast(T))\right] \\
&=\exp \left(A^{\mathbb{M}}(T-t,u)+B^{\mathbb{M}}(T-t,u)v+iuz\right),
\end{align*}

where $\mathbb{M} \in \{\mathbb{P}, \mathbb{Q} \}$ and $A^{\mathbb{M}}$ and $B^{\mathbb{M}}$ satisfy ordinary differential equations (ODEs):
\begin{equation}\label{eq:AP_BP_pde}
    \begin{aligned}
        0 &=-B_{\tau}^{\mathbb{P}}(\tau, u) + \left( \pi^\ast(\tau) \sigma \rho iu-\kappa \right) B^{\mathbb{P}}(\tau, u) +\frac{1}{2}\sigma^{2}\rBrackets{B^{\mathbb{P}}(\tau, u)}^{2}  -\frac{1}{2}\rBrackets{\pi^\ast(\tau)}^{2}\left( u^{2}+iu\right) + \pi^\ast(\tau) \overline{\lambda } iu;  \\
        0 &=-A_{\tau}^{\mathbb{P}}(\tau, u) + riu+\kappa \theta B^{\mathbb{P}}(\tau, u).
    \end{aligned}
\end{equation}
and
\begin{equation}\label{eq:AQ_BQ_pde}
    \begin{aligned}
        0 &=-B_{\tau}^\mathbb{Q}(\tau, u) + \left( \pi^\ast(\tau) \sigma \rho iu-\tilde{\kappa} \right) B^\mathbb{Q}(\tau, u) + \frac{1}{2}\sigma^{2} \rBrackets{B^\mathbb{Q}(\tau, u)}^{2} -\frac{1}{2}\rBrackets{\pi^\ast(\tau)}^{2}\left( u^{2}+iu\right);
        \\
        0 &=-A_{\tau}^\mathbb{Q}(\tau, u) + riu+ \tilde{\kappa} \tilde{\theta} B^\mathbb{Q}(\tau, u).
    \end{aligned}
\end{equation}%

respectively, where $\tau := T - t$.
\end{proposition}
\begin{proof}[Proof of Proposition \ref{prop:characteristic_fct_log_Y}]$\,$\\

Applying It\^{o}'s lemma to the wealth process $Y^*$ and the logarithmic function, we obtain the dynamics of $Z^*$ under the measure $\mathbb{P}$:
\begin{align*}
\mathrm{d}Z^\ast(t)& =\left( r+\left( \piu(t) \overline{\lambda }-\frac{1}{2}\rBrackets{\piu
(t)}^{2}\right) v(t)\right) \,\mathrm{d}t+\piu(t) \sqrt{v(t)}\,\mathrm{d}W^{\mathbb{P}}_{1}(t); \\
\mathrm{d}v(t)& =\kappa \left( \theta -v(t)\right) \!\mathrm{d}t+\sigma \rho
\sqrt{v(t)}\,\mathrm{d}W^{\mathbb{P}}_{1}(t)\,+\sigma \rho  \sqrt{v(t)}\sqrt{1-\rho
^{2}}\,\mathrm{d}W^{\mathbb{P}}_{2}(t).
\end{align*}

To make the notation within the proof concise, we write $\pi^\ast(t)$ for $\piu(t)$.
According to the Feynman-Kac theorem, the characteristic function satisfies the following relations under $\mathbb{P}$:%
\begin{align*}
\phi^{Z^\ast(T),\mathbb{P}} \left( u;t,z,v\right) & =\mathbb{E}^{\mathbb{P}}_{t,z,v}\left[ \exp (iuZ^\ast(T))\right];\\
0 &=\phi^{Z^\ast(T),\mathbb{P}}_{t}+(r+\left( \pi^\ast(t) \overline{\lambda }-\frac{1}{2}(\pi^\ast(t))^{2}\right)
v)\phi^{Z^\ast(T),\mathbb{P}} _{z}+\kappa (\theta -v)\phi^{Z^\ast(T),\mathbb{P}} _{v}\\
& \quad +\frac{1}{2} (\pi^\ast(t))^{2}v\phi^{Z^\ast(T),\mathbb{P}} _{zz}+\pi^\ast(t)
v\sigma \rho \phi^{Z^\ast(T),\mathbb{P}} _{zv}+\frac{1}{2}\sigma ^{2}v\phi^{Z^\ast(T),\mathbb{P}} _{vv}; \\
\phi^{Z^\ast(T),\mathbb{P}} (u;T,z,v) &= \exp \left( iuz\right).
\end{align*}%
Using the ansatz for the characteristic function:
\begin{equation*}
\phi ^{Z^\ast(T)\mathbb{P}}(u;t,z,v)=\exp \left(
A^{\mathbb{P}}(T-t,u)+B^{\mathbb{P}}(T-t,u)v+iuz\right),
\end{equation*}
changing the variable $\tau = T -t$, substituting and grouping under $\mathbb{P}$,we receive%
\begin{eqnarray*}
0&=&-A_{\tau}^{\mathbb{P}}(\tau, u) - B_{\tau}^{\mathbb{P}}(\tau, u)v + \rBrackets{r+\left( \pi^\ast(\tau) \overline{\lambda }-\frac{1}{2}\rBrackets{\pi^\ast(\tau)}^{2}\right) v}iu+\kappa (\theta -v)B^{\mathbb{P}}(\tau, u) \\
& &-\frac{1}{2}\rBrackets{\pi^\ast(\tau)}^{2} v u^{2}  + \pi^\ast(\tau) v\sigma \rho iu B^{\mathbb{P}}(\tau, u)+\frac{1}{2}\sigma ^{2}v \rBrackets{B^{\mathbb{P}}(\tau, u)}^{2}
\end{eqnarray*}%
and, thus,%
\begin{equation*}
    \begin{aligned}
        0 &=-B_{\tau}^{\mathbb{P}}(\tau, u) + \left( \pi^\ast(\tau) \sigma \rho iu-\kappa \right) B^{\mathbb{P}}(\tau, u) +\frac{1}{2}\sigma^{2}\rBrackets{B^{\mathbb{P}}(\tau, u)}^{2}-\frac{1}{2}\rBrackets{\pi^\ast(\tau)}^{2}\left( u^{2}+iu\right) +\pi^\ast(\tau) \overline{\lambda } iu;\\
        0 &=-A_{\tau}^{\mathbb{P}}(\tau, u) + riu+\kappa \theta B^{\mathbb{P}}(\tau, u).
    \end{aligned}
\end{equation*}%

Analogously, we obtain  the dynamics of $Z^*$ under the measure $\mathbb{Q}(\lambda^{v})$, writing $\mathbb{Q}$ for short:%
\begin{align*}
\mathrm{d}Z^\ast(t)& =\left( r-\frac{1}{2}\rBrackets{\piu(t)}^{2} v(t)\right) \mathrm{d}t+\piu(t)
\sqrt{v(t)}\, \mathrm{d}W^{\mathbb{Q}}_{1}(t); \\
\,\mathrm{d}v(t)& =\tilde{\kappa}\left( \tilde{\theta}-v(t)\right) \,%
\mathrm{d}t + \sigma \rho \sqrt{v(t)}\,\mathrm{d}W^{\mathbb{Q}}_{1}(t)\, +\sigma
\sqrt{v(t)}\sqrt{1-\rho ^{2}}\,\mathrm{d}W^{\mathbb{Q}}_{2}(t),
\end{align*}
with $\tilde{\kappa}=\kappa +\sigma \overline{\lambda }\rho +\sigma \lambda
^{v}\sqrt{1-\rho ^{2}}$ and $\tilde{\theta}=\kappa \theta /\tilde{\kappa}$.These parameters may be time-dependent due to $\lambda^{v}$.

Again using Feynman-Kac theorem and the ansatz
$$\phi^{Z^\ast(T),\mathbb{Q}}(u;t,z,v)=\exp \left(A^\mathbb{Q}(T-t,u)+B^\mathbb{Q}(T-t,u)v+iuz\right), $$
we obtain
\begin{eqnarray*}
0&=&-A_{\tau}^\mathbb{Q}(\tau, u)-B_{\tau}^\mathbb{Q}(\tau, u)v+\left( r-\frac{1}{2}\rBrackets{\pi^\ast(\tau)}^{2}v\right) iu + \tilde{\kappa}(\tilde{\theta} - v)B^\mathbb{Q}(\tau, u) \\
& &  - \frac{1}{2}\rBrackets{\pi^\ast(\tau)} ^{2}vu^{2} + \pi^\ast(\tau) v\sigma \rho iu B^\mathbb{Q} (\tau, u) +\frac{1}{2} \sigma ^{2}v \rBrackets{B^\mathbb{Q} (\tau, u)}^{2}.
\end{eqnarray*}%
Hence:
\begin{align*}
        0 &=-B_{\tau}^\mathbb{Q}(\tau, u) + \left( \pi^\ast(\tau) \sigma \rho iu-\tilde{\kappa} \right) B^\mathbb{Q}(\tau, u) + \frac{1}{2}\sigma^{2} \rBrackets{B^\mathbb{Q}(\tau, u)}^{2} -\frac{1}{2}\rBrackets{\pi^\ast(\tau)}^{2}\left( u^{2}+iu\right);
        \\
        0 &=-A_{\tau}^\mathbb{Q}(\tau, u) + riu+ \tilde{\kappa} \tilde{\theta} B^\mathbb{Q}(\tau, u).
\end{align*}%
\end{proof}

\textbf{Remark to Proposition \ref{prop:characteristic_fct_log_Y}.} The characteristic functions of $ \ln \rBrackets{Y^\ast(t)}$ have the same structural form as the characteristic functions of $\ln \rBrackets{S(t)}$. The latter function is known in closed form. The ODEs for $B^{\mathbb{P}}$ and $B^{\mathbb{Q}}$ are of Riccati type, as in the case of the characteristic functions of $\ln \rBrackets{S^\ast(T)}$. However, here, the coefficients of the Riccati ODEs for $B^{\mathbb{P}}$ and $B^{\mathbb{Q}}$ are time-dependent. Therefore, we solve them numerically. The analytical derivation of the solutions to these ODEs is beyond the scope of this paper.

\section{Proofs of main results}\label{app:constrained_OP_solution_short}

\begin{proof}[Proof of Theorem \ref{MainTheo}]

Our proof is based on the fact that two functions are equal if they satisfy the same PDEs with the same terminal conditions. In the following, we:
\begin{enumerate}
    \item use the dynamic programming approach to derive the HJB PDE of $\vfc(t,x,c)$, simplify it under the assumption that $\vfc_{xx}(t,x,v) < 0$ and get the optimal investment strategy $\pic$ in terms of the (to be found) function $\vfc(t,x,v)$;
    \item consider the PDE of $\udp(t, y, v)$ obtained via the Feynman-Kac (FK) theorem  and change of variables from $(t,y,v)$ to $(t,x,v)$ via $x = \dqla(t, y, v)$, i.e., $\vcansatz(t, x, v) := \udp(t, \rBrackets{\dqla}^{-1}(t, x, v), v)$ is our ansatz for the value function in the constrained optimization problem;
    \item simplify the PDE from Step 2 using the assumption \eqref{cond:D_v} that $\dq_v(t,y,v) = 0$ and using the PDE of $\dq(t,y,v)$ obtained via the FK theorem
    \item show that the resulting PDE in Step 3 coincides with the PDE of $\vfc(t,x,c)$:
    \begin{enumerate}
        \item for case $\rho = 0$ if Condition \eqref{cond:U_D_yy_y} holds;
        \item for case $\rho \neq 0$ if both Conditions \eqref{cond:U_D_yy_y}, \eqref{cond:U_D_yv_y} hold;
    \end{enumerate}
    \item show that the terminal conditions in the PDEs from Step 1 and Step 4 coincide and that $\vcansatz_{xx}(t,x,v) < 0$, which implies that $\vcansatz(t,x,v)$ solves the HJB PDE of $\vfc(t,x,c)$ and enables the calculation of $\pic$ from Step 1.
\end{enumerate}

To make the derivations in this theorem more readable, we omit the arguments of the functions $\mathcal{V}^{c}(t, x, v)$, \linebreak $\dqla(t, y, v)$, $\udp(t, y, v)$. We also omit the parameter $\lambda^{v}$ of the EMM $\mathbb{Q}(\lambda^{v})$.

\textbf{Step 1. HJB PDE of $\vfc$.} Similarly to the unconstrained Problem \eqref{PowerUnconstProb}, we face a two-dimensional control problem with state process $(X,v)$ and consider the HJB PDE:%
\begin{equation}\label{eq:HJB_PDE_vfc}
0=\mathcal{V}^{c}_{t}+\frac{1}{2}\sigma ^{2}v\mathcal{V}^{c}_{vv}+\kappa (\theta -v)\mathcal{V}^{c}_{v}+\underset{\pi }%
{\max }\left\{ x(r+\pi \overline{\lambda }v)\mathcal{V}^{c}_{x}+\frac{1}{2}\pi
^{2}x^{2}v\mathcal{V}^{c}_{xx}+\pi x \sigma v\rho \mathcal{V}^{c}_{xv}\right\}
\end{equation}%
and the boundary condition $\mathcal{V}^{c}(T,x,v)=\overline{U}(x)$. Eliminating $\max $ results in a first-order condition for $\pi $:
\begin{equation}\label{eq:pi_star_control}
\pic =-\frac{x\overline{\lambda }v\mathcal{V}^{c}_{x}+x\sigma v\rho \mathcal{V}^{c}_{xv}}{%
x^{2}v\mathcal{V}^{c}_{xx}}=-\frac{\overline{\lambda }\mathcal{V}^{c}_{x}+\sigma \rho \mathcal{V}^{c}_{xv}}{x\mathcal{V}^{c}_{xx}}
\end{equation}%
under the assumption that $\mathcal{V}^{c}_{xx} < 0$. Analogously to \eqref{eq:vu_pde}, we substitute the expression for $\pic$ back into the HJB PDE \eqref{eq:HJB_PDE_vfc} and get the following PDE for the value function $\mathcal{V}^{c}$:%
\begin{align}
& \mathcal{V}_{t}^{c}+xr\mathcal{V}_{x}^{c}+\kappa \theta \mathcal{V}_{v}^{c}+v\left( \frac{1}{2}\sigma
^{2}\mathcal{V}_{vv}^{c}-\kappa \mathcal{V}_{v}^{c}-\frac{1}{2}\frac{(\overline{\lambda }%
\mathcal{V}_{x}^{c}+\sigma \rho \mathcal{V}_{xv}^{c})^{2}}{\mathcal{V}_{xx}^{c}}\right) =0;
\label{eq:vc_pde} \\
& \mathcal{\mathcal{V}}^{c}(T,x,v)=\overline{U}\left( x\right) \label{eq:vc_terminal_condition}.
\end{align}%

\textbf{Steps 2-4. PDE of $\udp$ and a change of variables.} Recall from \eqref{eq:udp_pde} and \eqref{eq:udp_terminal_condition} that the FK representation of $\udp$ is given by:
\begin{align*}
0&=\udp_{t}+\left( r+\piu \overline{\lambda }v\right) y\udp_{y}+\kappa
\left( \theta -v\right) \udp_{v} + \frac{1}{2}v\Bigg[y^{2}(\piu)^{2}\udp_{yy}+2\sigma \rho y\piu \udp_{yv}+\sigma ^{2}\udp_{vv}\Bigg] \\
\udp(T,y,v)&=\overline{U}(D(y,v)).
\end{align*}

We change variables as follows:
\begin{equation}
    t = t, \quad x = \dq(t,y,v),\quad v = v.
\end{equation}
This change of variables leads to an equivalent PDE $\forall \, (t, y, v) \in [0, T] \times (0, +\infty) \times (0, +\infty)$, since:
\begin{align*}
\begin{vmatrix}
\frac{\partial t}{\partial t} & \frac{\partial t}{\partial y} & \frac{\partial t}{\partial v}  \\
\frac{\partial x}{\partial t} & \frac{\partial x}{\partial y} & \frac{\partial x}{\partial v}  \\
\frac{\partial v}{\partial t} & \frac{\partial v}{\partial y} & \frac{\partial v}{\partial v}  \notag
\end{vmatrix}
=
\begin{vmatrix}
1 & 0 & 0  \\
\dq_t & \dq_y & \dq_v  \\
0 & 0 & 1 \notag
\end{vmatrix}
= \dq_y \neq 0 \quad \forall \, (t, y, v) \in [0, T] \times (0, +\infty) \times (0, +\infty)
\end{align*}
under the assumption of $D(t,y, v)$ being non-decreasing in $y \in (0, +\infty)$ with a strictly increasing part for any $t \in [0, T], v \in (0, +\infty)$. The condition above is needed to ensure that the change of variables $(t,x,v) \leftrightarrow (t,y,v)$ is bijective, which is necessary for the equivalence of the respective PDEs on the whole domain.

Using the ansatz
\begin{equation} \label{eq:vc_ansatz}
    \udptyv = \vcansatz(t,\dq(t,y,v),v),
\end{equation}
we compute the corresponding derivatives that appear in the PDE of $\udp$:

\begin{equation}\label{eq:udp_vcansatz_derivatives}
    \begin{aligned}
    \udp_{t}& =\vcansatz_{t} + \vcansatz_{x}\dq _{t}, \\
    \udp_{y}& =\vcansatz_{x}\dq _{y}, \\
    \udp_{v}& =\vcansatz_{x}\dq _{v} + \vcansatz_{v} \stackrel{\eqref{cond:D_v}}{=} \vcansatz_{v}, \\
    \udp_{yy}& =\vcansatz_{xx}(\dq _{y})^{2}+\vcansatz_{x}\dq _{yy}, \\
    \udp_{yv}& =\vcansatz_{xv}\dq _{y}+\vcansatz_{x}\dq _{yv}+\vcansatz_{xx}\dq _{y}\dq _{v} \stackrel{\eqref{cond:D_v}}{=}\vcansatz_{xv}\dq _{y}+\vcansatz_{x}\dq _{yv}, \\
    \udp_{vv}& =2\vcansatz_{xv}\dq _{v}+\vcansatz_{x}\dq _{vv}+\vcansatz_{vv}+\vcansatz_{xx}\left(
    \dq _{v}\right) ^{2} \stackrel{\eqref{cond:D_v}}{=}\vcansatz_{x}\dq _{vv}+\vcansatz_{vv}.
    \end{aligned}%
\end{equation}
Next we substitute these derivatives into the PDE of $\udp$, also
use the PDE for $\dq _{t}$ to simplify the equation, and then we cancel out terms and insert the assumption $\dq (t,y, v)=x$.
\begin{align*}
0\stackrel{\eqref{eq:udp_pde}}{=}& \udp_{t}+yr\udp_{y}+\kappa \theta \udp_{v}+v\biggl( \frac{1}{2}\sigma
^{2}\udp_{vv}-\kappa \udp_{v}+y\overline{\lambda }\piu \udp_{y}\\
& +\frac{1}{2} y^{2}(\piu)^{2}\udp_{yy}+\sigma \rho y \piu \udp_{yv}\biggr) \\
\stackrel{\eqref{eq:udp_vcansatz_derivatives}}{=} & \vcansatz_{t} + \vcansatz_{x}\dq _{t}+yr\vcansatz_{x}\dq _{y}+\kappa \theta \vcansatz_{v} + v%
\Bigg[\frac{1}{2}\sigma ^{2}\left( \vcansatz_{x}\dq _{vv}+\vcansatz_{vv}\right)
-\kappa \vcansatz_{v}+y\overline{\lambda }\piu\vcansatz_{x}\dq _{y} \\
& +\frac{1}{2}y^{2}(\piu)^{2}\left( \vcansatz_{xx}(\dq
_{y})^{2}+\vcansatz_{x}\dq _{yy}\right) +\sigma \rho y\piu\left(
\vcansatz_{xv}\dq _{y}+\vcansatz_{x}\dq _{yv}\right) \Bigg] \\
\stackrel[\eqref{eq:dq_pde}]{\eqref{cond:D_v}}{=}& \vcansatz_{t}+\vcansatz_{x}rx+\kappa \theta \vcansatz_{v}+v\biggl( \frac{1}{2}\sigma
^{2}\vcansatz_{vv}-\kappa \vcansatz_{v}+y\overline{\lambda }\piu\vcansatz_{x}\dq
_{y}\\
& +\frac{1}{2}y^{2}(\piu)^{2}\left( \vcansatz_{xx}(\dq _{y})^{2}\right)
+\sigma \rho y\piu\left( \vcansatz_{xv}\dq _{y}\right) \biggr) \\
\stackrel{(i)}{=}&\vcansatz_{t}+xr\vcansatz_{x}+\kappa \theta \vcansatz_{v}+v\left[ \frac{1}{2}\sigma^{2}\vcansatz_{vv}-\kappa \vcansatz_{v} - \frac{1}{2}\frac{(\overline{\lambda }%
\vcansatz_{x} + \sigma \rho \vcansatz_{xv})^{2}}{\vcansatz_{xx}}\right] + v\\
&\cdot \underbrace{\Biggl( y\overline{\lambda }\piu \vcansatz_{x}\dq _{y}+\frac{1}{2}y^{2}(\piu)^{2}\vcansatz_{xx}\rBrackets{\dq_{y}}^{2}+\sigma \rho y\piu\vcansatz_{xv}\dq _{y} + \frac{1}{2}\frac{(\overline{\lambda }%
\vcansatz_{x} + \sigma \rho \vcansatz_{xv})^{2}}{\vcansatz_{xx}}\Biggr)}_{C},
\end{align*}%
where in $(i)$ we added and subtracted the term $-v\frac{1}{2}\frac{(\overline{\lambda }%
\vcansatz_{x} + \sigma \rho \vcansatz_{xv})^{2}}{\vcansatz_{xx}}$.

We show now that under Conditions \eqref{cond:U_D_yy_y} and \eqref{cond:U_D_yv_y}, the term $C$ is zero. Expanding the brackets in the last term of $C$ we get:
\begin{equation*}
\begin{aligned}
C &= y\overline{\lambda }\piu \vcansatz_{x} \dq _{y}+\frac{1}{2}%
y^{2}(\piu)^{2}\vcansatz_{xx}\rBrackets{\dq _{y}}^{2}+\sigma \rho y \piu \vcansatz_{xv}\dq _{y} + \frac{1}{2}\overline{\lambda }^{2}\frac{\left( \vcansatz_{x}\right)^{2}}{\vcansatz_{xx}} + \overline{\lambda }\sigma \rho \frac{\vcansatz_{x} \vcansatz_{xv}}{\vcansatz_{xx}}+\frac{1}{2}\sigma ^{2}\rho ^{2}\frac{\left( \vcansatz_{xv}\right)^{2}}{\vcansatz_{xx}}.
\end{aligned}
\end{equation*}

Using \eqref{eq:udp_vcansatz_derivatives}, we obtain:%
\begin{equation} \label{eq:vcansatz_derivatives}
\begin{aligned}
&\vcansatz_{x}=\frac{\udp_{y}}{\dq _{y}},\vcansatz_{xx} = \frac{\udp_{yy}\dq _{y} - \udp_{y}\dq_{yy}}{\rBrackets{\dq _{y}}^{3}}, \vcansatz_{xv} = \frac{1}{\dq _{y}}\left( \udp_{yv}-\vcansatz_{x}\dq_{yv}\right) =\frac{\udp_{yv}\dq _{y}-\udp_{y}\dq _{yv}}{\rBrackets{\dq _{y}}^{2}}.
\end{aligned}
\end{equation}%
Inserting these expressions in $C$, we get:%
\begin{eqnarray*}
C & = &y\overline{\lambda }\piu\udp_{y}+\frac{1}{2}y^{2}(\piu)^{2}%
\frac{\udp_{yy}\dq _{y}-\udp_{y}\dq _{yy}}{\dq _{y}}+\sigma \rho y\piu%
\frac{\udp_{yv}\dq _{y}-\udp_{y}\dq _{yv}}{\dq _{y}} \\
&&+\frac{1}{2}\overline{\lambda }^{2}\frac{\rBrackets{\udp_{y}}^{2}\dq _{y}}{\udp_{yy}\dq
_{y}-\udp_{y}\dq _{yy}} + \overline{\lambda }\sigma \rho \udp_{y}\frac{\udp_{yv}\dq
_{y}-\udp_{y}\dq _{yv}}{\udp_{yy}\dq _{y}-\udp_{y}\dq _{yy}} \\
&& +\frac{1}{2}\sigma
^{2}\rho ^{2}\frac{\left( \udp_{yv}\dq _{y}-\udp_{y}\dq _{yv}\right) ^{2}}{\dq
_{y}\left( \udp_{yy}\dq _{y}-\udp_{y}\dq _{yy}\right) }
\end{eqnarray*}
\begin{eqnarray*}
& = &y\overline{\lambda }\piu\udp_{y}+\frac{1}{2}y^{2}(\piu)^{2}\udp_{y}\left( \frac{\udp_{yy}}{\udp_{y}}-\frac{\dq _{yy}}{\dq _{y}}\right)
+\sigma \rho y\piu\udp_{y}\left( \frac{\udp_{yv}}{\udp_{y}}-\frac{\dq _{yv}}{%
\dq _{y}}\right) \\
&&+\frac{1}{2}\overline{\lambda }^{2}\frac{\udp_{y}}{\left( \frac{\udp_{yy}}{\udp_{y}%
}-\frac{\dq _{yy}}{\dq _{y}}\right) } + \overline{\lambda }\sigma \rho \udp_{y}%
\frac{\left( \frac{\udp_{yv}}{\udp_{y}}-\frac{\dq _{yv}}{\dq _{y}}\right) }{\left(
\frac{\udp_{yy}}{\udp_{y}}-\frac{\dq _{yy}}{\dq _{y}}\right) } + \frac{1}{2}\sigma
^{2}\rho ^{2}\udp_{y}\frac{\left( \frac{\udp_{yv}}{\udp_{y}}-\frac{\dq _{yv}}{\dq _{y}%
}\right) ^{2}}{\left( \frac{\udp_{yy}}{\udp_{y}}-\frac{\dq _{yy}}{\dq _{y}}\right)
}
\end{eqnarray*}%
Denoting%
\begin{equation*}
A=\left( \frac{\udp_{yy}}{\udp_{y}}-\frac{\dq _{yy}}{\dq _{y}}\right) \text{ and }%
B=\left( \frac{\udp_{yv}}{\udp_{y}}-\frac{\dq _{yv}}{\dq _{y}}\right),
\end{equation*}
we get:%
\begin{equation} \label{KeyPDEF}
C = \udp_{y}\rBrackets{y\overline{\lambda }\piu+\frac{1}{2}y^{2}(\piu)^{2}A+\sigma
\rho y\piu B + \frac{1}{2}\overline{\lambda }^{2}\frac{1}{A} + \overline{%
\lambda }\sigma \rho \frac{B}{A} + \frac{1}{2}\sigma ^{2}\rho ^{2}\frac{B^{2}}{%
A}}.
\end{equation}%

If $\rho = 0$, the term $B$ disappears (i.e., no condition on $B$ is required) and \eqref{KeyPDEF} becomes:
\begin{equation*}
    \begin{aligned}
        C =  \udp_{y}\rBrackets{y\overline{\lambda }\piu+\frac{1}{2}y^{2}(\piu)^{2}A + \frac{1}{2}\overline{\lambda }^{2}\frac{1}{A}} \stackrel{!}{=}0 & \stackrel{\udp_y > 0}{\iff} \frac{1}{2 A} \left( \overline{\lambda } + y\piu A \right)^2 \stackrel{!}{=}0 \stackrel{\eqref{eq:pi_star_unconstrained}}{\iff} A \stackrel{!}{=} -\frac{1- \gamma}{y},\\
    \end{aligned}
\end{equation*}
i.e., Condition \eqref{cond:U_D_yy_y} of this theorem. Thus, we conclude that $\vcansatz$ satisfies the PDE \eqref{eq:vc_pde}.

If $\rho \neq 0$, we insert $A = -\frac{1-\gamma }{y}$ into \eqref{KeyPDEF} and get:
\begin{eqnarray*}
C &=& \udp_{y}\Biggl( y\overline{\lambda }\piu+\frac{1}{2}y^{2}(\piu)^{2}\rBrackets{-\frac{1-\gamma }{y}}+\sigma
\rho y\piu B + \frac{1}{2}\overline{\lambda }^{2}\rBrackets{-\frac{y}{1-\gamma }} \\
&& + \overline{\lambda }\sigma \rho B \rBrackets{-\frac{y}{1-\gamma }} + \frac{1}{2}\sigma ^{2}\rho ^{2}B^{2} \rBrackets{-\frac{y}{1-\gamma }} \Biggr)\\
&\stackrel{\eqref{eq:pi_star_unconstrained}}{=}& \frac{\udp_{y}y}{1 - \gamma}\Biggl( \overline{\lambda }(\lambar + \sigma \rho b(t)) - \frac{1}{2}(\lambar + \sigma \rho b(t))^{2} +\sigma \rho (\lambar + \sigma \rho b(t)) B - \frac{1}{2}\overline{\lambda }^{2} - \overline{\lambda }\sigma \rho B  - \frac{1}{2}\sigma ^{2}\rho ^{2}B^{2} \Biggr) \\
& = & \frac{\udp_{y}y}{1 - \gamma} \Biggl( \lambar^{2} + \lambar \sigma \rho b(t) - \frac{1}{2}\lambar^{2}  - \lambar \sigma \rho b(t) - \frac{1}{2}(\sigma \rho b(t))^{2} + \sigma \rho \lambar B + (\sigma \rho)^{2} b(t) B \\
&& - \frac{1}{2}\lambar^{2}  - \lambar \sigma \rho B  - \frac{1}{2}\sigma ^{2}\rho ^{2}B^{2} \Biggr)\\
& = & \frac{\udp_{y}y}{1 - \gamma} \rBrackets{ - \frac{1}{2}(\sigma \rho b(t))^{2} + (\sigma \rho)^{2} b(t) B - \frac{1}{2}\sigma ^{2}\rho ^{2}B^{2}}\\
& = & \frac{\udp_{y}y}{1 - \gamma} \frac{\sigma^{2} \rho^{2}}{2} \rBrackets{b(t) - B}^{2}.
\end{eqnarray*}%
Hence, if $\rho \neq 0$, $A = -\frac{1 - \gamma}{y}$ and $B = b(t)$, i.e., Conditions \eqref{cond:U_D_yy_y} and \eqref{cond:U_D_yv_y} hold, then $C = 0$. Thus, we conclude that $\vcansatz$ satisfies PDE \eqref{eq:vc_pde}.

\textbf{Step 5. Concluding the value function and the optimal investment strategy.}\\
Having shown that $\vcansatz$ satisfies the HJB PDE of $\vfc$ for any $\rho \in [-1,1]$, we now show that $\vcansatz$ satisfies the terminal condition of the HJB PDE of $\vfc$:
\begin{equation*}
\vcansatz(T, \dq(T,y,v), v) \stackrel{\eqref{eq:vc_ansatz}}{=} \udp(T, y, v) \stackrel{\eqref{eq:udp_terminal_condition}}{=} \overline{U}(D(y,v)) \stackrel{\eqref{eq:def_of_financial_derivative_on_Y}}{=} \overline{U}(\dq(T, y, v)),
\end{equation*}%
i.e., \eqref{eq:vc_terminal_condition} holds with $x=\dq \left( T,y,v\right)$.

Next, we prove that $\vcansatz$ satisfies the assumption of concavity in $x$. Observe that:
\begin{equation}\label{eq:vcansatz_xx_sign}
\begin{aligned}
    \vcansatz_{xx} &\stackrel{\eqref{eq:vcansatz_derivatives}}{=} \frac{\udp_{yy}\dq _{y} - \udp_{y}\dq_{yy}}{\rBrackets{\dq _{y}}^{3}} \stackrel{\text{Def.}\,A}{=} \rBrackets{\dq _{y}}^{-3} A \udp_{y} \dq _{y} \stackrel{\eqref{cond:U_D_yy_y}}{=} \underbrace{\rBrackets{\dq _{y}}^{-2}}_{>0} \underbrace{\rBrackets{-\frac{1 - \gamma}{y}}}_{ < 0 } \udp_{y},
    \end{aligned}
\end{equation}
since $y > 0,\gamma < 1$, and $D(\cdot,v)$ is assumed to be non-decreasing on $(0, +\infty)$ with a strictly increasing part. If $\udp_{y} > 0$, then $\vcansatz_{xx} < 0$.

Take any $y > 0$ and $\Delta y > 0$. Obviously, for any $\omega \in \Omega$ the following holds:
\begin{equation*}
Y^\ast_{\omega}(T)|_{Y^{\ast}_{\omega}(t) = y + \Delta y} > Y^{\ast}_{\omega}(T)|_{Y^{\ast}_{\omega}(t) = y}.
\end{equation*}
Denote by $(\underline{d}, \overline{d}) \subset (0, +\infty)$ the sub-interval where $D(\cdot,v)$ is strictly increasing. Denote $\mathcal{S}(y) = \Bigl\{\omega \in \Omega: Y^{\ast}_{\omega}(T) \in (\underline{d}, \overline{d}) | Y^{\ast}_{\omega}(t) = y \Bigr\}$.
Then, according to \eqref{Heston-stock} and \eqref{Heston-vol}, $\mathbb{P}(\mathcal{S}(y)) > 0\,\,\forall y > 0$. The function $\overline{U}(D(y, v)) = U(D(y,v)) - \lambda_{\varepsilon} (\mathbbm{1}_{\{D(y,v) < K\}} - \varepsilon)$ is strictly increasing in $y$ because $U(\cdot)$ is strictly increasing, $\lambda_\varepsilon \geq 0$, and $\mathbbm{1}_{\{D(y,v) < K\}}$ is a non-increasing function as a superposition $\alpha(\beta(y;v))$ of a non-increasing function $\alpha(x) = \mathbbm{1}_{\{x < K\}}$ and a non-decreasing function $\beta(y; v) = D(y,v)$. Using these properties and the linearity of the expectation operator, we obtain that $\udp$ is strictly increasing in $y$ as follows:
\begin{flalign*}
    \udp(t, y + \Delta y, v) & = \mathbb{E}^{\mathbb{P}}_{t, y + \Delta y, v} \sBrackets{\overline{U}(D(Y^{\ast}(T),v(T)))} \\
    & =  \mathbb{E}^{\mathbb{P}}_{t, y + \Delta y, v} \sBrackets{\overline{U}(D(Y^{\ast}(T),v(T)))\mathbbm{1}_{\left\{ \mathcal{S}(y + \Delta y)\right\} }}  \\
    & \quad + \mathbb{E}^{\mathbb{P}}_{t, y + \Delta y, v} \sBrackets{\overline{U}(D(Y^{\ast}(T),v(T)))\mathbbm{1}_{\left\{\Omega \setminus \mathcal{S}(y + \Delta y)\right\}}}\\
    & >  \mathbb{E}^{\mathbb{P}}_{t,y, v} \sBrackets{\overline{U}(D(Y^{\ast}(T),v(T)))\mathbbm{1}_{\left\{ \mathcal{S}(y + \Delta y)\right\} }}  + \mathbb{E}^{\mathbb{P}}_{t,y, v} \sBrackets{\overline{U}(D(Y^{\ast}(T),v(T)))\mathbbm{1}_{\left\{\Omega \setminus \mathcal{S}(y + \Delta y)\right\}}}\\
    & =  \mathbb{E}^{\mathbb{P}}_{t,y, v} \sBrackets{\overline{U}(D(Y^{\ast}(T),v(T)))} = \udp(t, y, v)
\end{flalign*}

So $\udp$ is strictly increasing in $y$. Therefore, $\udp_{y} > 0$, and via \eqref{eq:vcansatz_xx_sign} we obtain $\vcansatz_{xx} < 0$.

Since $\vcansatz$ satisfies the PDE of $\vfc$, the corresponding terminal condition, and $\vcansatz_{xx} < 0$, we conclude that it is a candidate for the value function in the constrained optimization problem. Thus, we can now calculate the candidate for the optimal investment strategy. Plugging
\begin{eqnarray*}
\frac{\vcansatz_{xv}}{\vcansatz_{xx}}\stackrel{\eqref{eq:vcansatz_derivatives}}{=}\frac{\udp_{yv}\dq _{y}-\udp_{y}\dq_{yv}}{\rBrackets{\dq _{y}}^{2}}\frac{\rBrackets{\dq _{y}}^{3}}{\udp_{yy}\dq _{y}-\udp_{y}\dq_{yy}}=\frac{B}{A}\dq _{y} = -\frac{y b(t)}{1- \gamma}\dq_{y}.
\end{eqnarray*}%
and $\vcansatz_{x}$ as well as $\vcansatz_{xx}$ from \eqref{eq:vcansatz_derivatives} into \eqref{eq:pi_star_control}, we obtain the optimal control in the constrained portfolio optimization problem:
\begin{equation*}
\pic(t) = -\frac{\overline{\lambda }\vcansatz_{x}}{x\vcansatz_{xx}}-\frac{\sigma
\rho \vcansatz_{xv}}{x\vcansatz_{xx}}\stackrel{}{=}\frac{y\overline{\lambda }}{1-\gamma }\frac{\dq _{y}}{\dq }+\frac{y\sigma \rho }{1-\gamma }b(t)\frac{\dq _{y}}{\dq } = \piu(t)\frac{y \dq _{y}}{\dq }.
\end{equation*}
\end{proof}

\textbf{Remark}
The above proof uses $D(\cdot,\cdot)$ to ensure a matching of the terminal condition and the necessary Conditions \eqref{cond:U_D_yy_y}--\eqref{cond:D_v}. The choice of $\lambda^{v}$ is crucial to ensure the Conditions \eqref{cond:U_D_yy_y}--\eqref{cond:D_v}.

\bigskip

\begin{proof}[Proof of Lemma \ref{lem:sufficient_condition}]

If $\udp_{y} = y^{\gamma -1}H(t,v)\dq _{y}$, then:
\begin{equation*}
\begin{aligned}
    \frac{\udp_{yy}}{\udp_{y}}-\frac{\dq _{yy}}{\dq _{y}} = \frac{\left( \gamma
-1\right) H(t,v)\dq _{y}y^{\gamma -2} +H(t,v)\dq
_{yy}y^{\gamma - 1} }{H(t,v)\dq _{y}y^{\gamma -1}}-\frac{\dq _{yy}}{\dq _{y}} = -\frac{1-\gamma }{y},
\end{aligned}
\end{equation*}
i.e., Condition \eqref{cond:U_D_yy_y} holds.

If $H(t,v) = h(t)\exp(b(t)v)$, where $H(t, v)$ does not depend on $y$, then we also have the following:
\begin{eqnarray*}
\frac{\udp_{yv}}{\udp_{y}}-\frac{\dq _{yv}}{\dq _{y}} =\frac{b(t)\dq
_{y}y^{\gamma -1}H(t, v) + \dq _{yv}y^{\gamma - 1}H(t, v) }{\dq _{y}y^{\gamma -1}H(t,v) }-
\frac{\dq _{yv}}{\dq _{y}}=b(t),
\end{eqnarray*}
i.e.,  both Conditions \eqref{cond:U_D_yy_y} and \eqref{cond:U_D_yv_y} are satisfied.
\end{proof}

\begin{proof}[Proof of Proposition \ref{prop:equivalence_btw_sequence_and_single_D}]
{\color{black} Without loss of generality, we assume that $r=0$ and consider the following model under an EMM $\mathbb{Q}$:
\begin{eqnarray*}
dA(t) & = &A(t)\sigma_{A}(t)\sqrt{v(t)}dW_1^{\mathbb{Q}}; \\
dv(t) & = &\tilde{\kappa}(t) \left( \tilde{\theta}(t) - v(t)\right) dt + \sigma_{v}\sqrt{v(t)}dW_3^{\mathbb{Q}},
\end{eqnarray*}
where $W_3^{\mathbb{Q}} = \rho W_1^{\mathbb{Q}} + \sqrt{1 - \rho^2}W_2^{\mathbb{Q}}$, $A=(A(t))_{\tin}$ is the price process of a generic asset, and $\sigma_{A}(t), \tilde{\kappa}(t), \tilde{\theta}(t)$ are deterministic functions of time, whose argument we drop in the rest of this proof to make notation easier.

Consider a generic contingent claim with value:
\begin{eqnarray*}
C(t,A(t),v(t);k)=\mathbb{E}_{t}^{\mathbb{Q}}\left[ \widehat{G}\left( A(T),k\right) \right],
\end{eqnarray*}
where $k$ is assumed to be a scalar parameter for simplicity (e.g., strike an of an option), but it could be a vector of parameters (e.g., strikes of multiple options constituting the contingent claim $C$).

By the FK theorem, the price process $C$ of a contingent claim with a fixed $k$ satisfies the following PDE and terminal condition:%
\begin{eqnarray*}
C_{t} +\tilde{\kappa} \left( \tilde{\theta} -v\right) C_{v} +\frac{1}{2}\sigma _{A}^{2}A^{2} v C_{AA}+\rho \sigma _{A}\sigma
_{v} A v C_{Av}+\frac{1}{2}\sigma _{v}^{2} v C_{vv} &=& 0; \\
C(T,A,v;k) &=&\widehat{G}\left( A,k\right).
\end{eqnarray*}

If we roll over this contingent claim $C$, it creates a new product that can be interpreted as a continuum of financial derivatives. This product has the following price at $\tin$:
\begin{equation*}
    \widehat{D}(t,A(t),v(t);k\left( t,A(t),v(t)\right) )=\mathbb{E}_{t}^{\mathbb{Q}}\left[ \widehat{G}\left(A(T),k\left( t,A(t),v(t)\right) \right) \right],
\end{equation*}
where $k\left( t,A,v\right) $ is now seen as a function of $(t, A, v)$, and it is assumed to be such that $\widehat{D}$ is attainable, i.e., the financial derivative $\widehat{D}$ can be hedged by a self-financing portfolio. This means:%
\begin{equation}\label{eq:D_hat_attainability}
\frac{d\widehat{D}(t,A(t),v(t);k\left( t,A(t),v(t)\right) )}{\widehat{D}%
(t,A(t),v(t);k\left( t,A(t),v(t)\right) )} =\pi \left(
t,A(t),v(t)\right) \frac{dA(t)}{A(t)} =  \pi \left( t,A(t),v(t)\right) A(t)\sigma _{A}\sqrt{v} dW_1^{\mathbb{Q}}
\end{equation}
for some function $\pi \left( t,A,v\right)$. To make notation less cumbersome, we omit time when referring to a process at time $t$.

Applying It{\^o}'s lemma to $\widehat{D}(t,A,v;k\left( t,A,v\right) )$, we get:
\begin{equation}\label{eq:D_hat_Itos_lemma}
\begin{aligned}
    d\widehat{D}(t,A,v;k\left( t,A,v\right) ) =&\left( \widehat{%
D}_{t}+\widehat{D}_{k}k_{t} + \frac{1}{2} A^{2} \sigma _{A}^{2} v\left( \widehat{D}%
_{AA}+2\widehat{D}_{Ak}k_{A}+\widehat{D}_{kk}k_{A}^{2}+\widehat{D}%
_{k}k_{AA}\right) \right) dt \\
&+\left( \left( \widehat{D}_{v}+\widehat{D}_{k}k_{v}\right) \tilde{\kappa} \left(
\tilde{\theta} -v\right) +\frac{1}{2}\sigma _{v}^{2}v\left( \widehat{D}_{vv}+2%
\widehat{D}_{vk}k_{v}+\widehat{D}_{kk}k_{v}^{2}+\widehat{D}_{k}k_{vv}\right)
\right) dt \\
&+\left( \rho \sigma _{v}\sigma _{A} A v\left( \widehat{D}_{vA}+\widehat{D}%
_{vk}k_{A}+\widehat{D}_{kA}k_{v}+\widehat{D}_{kk}k_{A}k_{v}+\widehat{D}%
_{k}k_{vA}\right) \right) dt \\
&+\left( \widehat{D}_{A}+\widehat{D}_{k}k_{A}\right) A \sigma _{A}\sqrt{v}%
dW_1^{\mathbb{Q}}+\left( \widehat{D}_{v}+\widehat{D}_{k}k_{v}\right) \sigma _{v}\sqrt{v}dW_3^{\mathbb{Q}}.
\end{aligned}
\end{equation}


Matching the SDEs \eqref{eq:D_hat_attainability} and \eqref{eq:D_hat_Itos_lemma}, we must ensure that the terms related to $dt$, $dW^\mathbb{Q}_{1}$ and $dW^\mathbb{Q}_{2}$ are equal. The equality of diffusion terms $dW^\mathbb{Q}_{1}$ and $dW^\mathbb{Q}_{2}$ implies that: 
\begin{eqnarray*}
\widehat{D} \pi \left( t,A,v\right) A \sigma _{A}\sqrt{v}  &=& \rBrackets{\widehat{D}_{A}+\widehat{D}_{k}k_{A}} A \sigma _{A}\sqrt{v} \Longleftrightarrow \widehat{D} \pi \left( t,A,v\right)  = \widehat{D}_{A}+\widehat{D}_{k}k_{A};\\
\left( \widehat{D}_{v}+\widehat{D}_{k}k_{v}\right)\sigma _{v}\sqrt{v}  &=& 0 \Longleftrightarrow
k_{v}=-\frac{\widehat{D}_{v}}{\widehat{D}_{k}}.
\end{eqnarray*}

The previous equation is a condition on strike $k$ due to the incompleteness of the financial market. Since our rolling derivative is constructed to be vega neutral at all $\tin$, we naturally have $\widehat{D}_{v}+\widehat{D}_{k}k_{v}=0$. 

The equality of the drift terms and the terminal conditions implies:%
\begin{eqnarray*}
&&\left( \widehat{D}_{t}+\widehat{D}_{k}k_{t}+\frac{1}{2}A^{2} \sigma _{A}^{2} v\left(
\widehat{D}_{AA}+2\widehat{D}_{Ak}k_{A}+\widehat{D}_{kk}k_{A}^{2}+\widehat{D}%
_{k}k_{AA}\right) \right) + \\
&&\left( \left( \widehat{D}_{v}+\widehat{D}_{k}k_{v}\right) \tilde{\kappa} \left(
\tilde{\theta} -v\right) +\frac{1}{2}\sigma _{v}^{2}v\left( \widehat{D}_{vv}+2%
\widehat{D}_{vk}k_{v}+\widehat{D}_{kk}k_{v}^{2}+\widehat{D}_{k}k_{vv}\right)
\right) + \\
&&\left( \rho \sigma _{v}\sigma _{A} A v \left( \widehat{D}_{vA}+\widehat{D}%
_{vk}k_{A}+\widehat{D}_{kA}k_{v}+\widehat{D}_{kk}k_{A}k_{v}+\widehat{D}%
_{k}k_{vA}\right) \right) = 0; \\
&&\widehat{D}(T,A,v;k\left( t,A,v\right) )=\widehat{G}\left(
A,k\left( t,A,v\right) \right).
\end{eqnarray*}

We can rewrite the previous PDE in the following way:%
\begin{eqnarray*}
&&\widehat{D}_{t}+\left( k_{t}+k_{v}\tilde{\kappa} \left( \tilde{\theta} -v\right) +\frac{1}{2}%
A^{2}\sigma _{A}^{2}v k_{AA} + \frac{1}{2}\sigma _{v}^{2}v_t k_{vv}+\rho \sigma
_{v}\sigma _{A} A v k_{vA}\right) \widehat{D}_{k}\\
&& +\left( \frac{1}{2}\sigma
_{A}^{2} A^{2} k_{A}^{2}+\frac{1}{2}\sigma _{v}^{2}k_{v}^{2}+\rho \sigma _{v}\sigma
_{A} A k_{A} k_{v}\right) v\widehat{D}_{kk}\\
&&+\left( \sigma _{A}^{2} A^{2} k_{A}+\rho \sigma _{v}\sigma _{A} A k_{v}\right) v%
\widehat{D}_{Ak}+\left( \rho \sigma _{v}\sigma _{A} A k_{A}+\sigma
_{v}^{2}k_{v}\right) v\widehat{D}_{vk} \\
&& +\rho \sigma _{v}\sigma _{A} A v \widehat{D%
}_{vA} + \frac{1}{2}\sigma _{A}^{2} A^{2} v \widehat{D}_{AA}+\tilde{\kappa} \left( \tilde{\theta}
-v\right) \widehat{D}_{v}+\frac{1}{2}\sigma _{v}^{2}v\widehat{D}_{vv}=0.
\end{eqnarray*}


This is a FK formula for the price of a financial derivative with three underlying assets $(A(t),v(t),\widetilde{k}(t))$, one of which is perfectly correlated to the others:%
\begin{eqnarray*}
dA(t) &=&A(t)\sigma _{A}\sqrt{v(t)}dW_1^{\mathbb{Q}}; \\
dv(t) &=&\tilde{\kappa} \left( \tilde{\theta} -v(t)\right) dt+\sigma _{v}\sqrt{v(t)}dW_3^{\mathbb{Q}}; \\
d\widetilde{k}(t) &=&\left( k_{t}+k_{v}\tilde{\kappa} \left( \tilde{\theta} -v(t)\right) +\frac{1}{2}\sigma
_{A}^{2}v(t)k_{AA}+\frac{1}{2}\sigma _{v}^{2} \rBrackets{A(t)}^2 v(t) k_{vv}+\rho \sigma _{v}\sigma
_{A} A(t) v(t)k_{vA}\right) dt \\
&&+\sigma _{A}k_{A} A(t) \sqrt{v(t)}dW_1^{\mathbb{Q}}+\sigma _{v}k_{v}\sqrt{v(t)}dW_3^{\mathbb{Q}}.
\end{eqnarray*}


Therefore, $\widehat{D}$ can be interpreted as a single financial derivative $\widetilde{D}$ on three underlying assets:%
\[
\widetilde{D}\left( t,A,v,\widetilde{k}\right) =\mathbb{E}_{t}^{\mathbb{Q}}\left[ \widetilde{G}\left(
A(T),\widetilde{k}(T)\right) \right] = \mathbb{E}_{t}^{\mathbb{Q}}\left[ \widehat{G}\left( A(T),k\left(
t,A,v\right) \right) \right] =\widehat{D}(t,A,v;k\left(
t,A,v\right) ).
\]


In this derivative, $\widetilde{D}\left( t,A,v,\widetilde{k}\right) $, the process $\widetilde{k}$ is an explicit function of time, asset price, and variance, that is, $\widetilde{k}(t)=k\left( t,A(t),v(t)\right)$. Therefore, $\widetilde{D}\left( t,A,v,\widetilde{k}\right) $ can be interpreted as $D^{\mathbb{Q}}\left( t,A,v\right) $, i.e., the financial derivative invoked in Theorem \ref{MainTheo}, and the payof{}f $\widetilde{G}\left(A(T),\widetilde{k}(T)\right)$ can be seen as $G\left(A(T),v(T)\right)$ for an implied function $G$.
\[
\widetilde{D}\left( t,A,v,\widetilde{k}\right) =\mathbb{E}_{t}^{\mathbb{Q}}\left[ G\left(
A(T),\widetilde{k}(T)\right) \right] = \mathbb{E}_{t}^{\mathbb{Q}}\left[ G\left( A(T),v(T) \right) \right] =D^{\mathbb{Q}}\left( t,A,v\right).
\]

Now we apply FK theorem again and get:%
\begin{eqnarray*}
D_{t}^{\mathbb{Q}}+\tilde{\kappa} \left( \tilde{\theta} -v\right) D_{v}^{\mathbb{Q}}+\frac{1}{2}\sigma _{A}^{2}{A}^{2}D_{AA}^{\mathbb{Q}}+\rho \sigma _{A}\sigma _{v} A v D_{Av}^{\mathbb{Q}}+%
\frac{1}{2}\sigma _{v}^{2}vD_{vv}^{\mathbb{Q}} &=&0;
\\
D^{\mathbb{Q}}\left( T,A,v\right)^{\mathbb{Q}}  &=& G\left( A,v\right).
\end{eqnarray*}

These calculations indicate that a rolling-over contingent claim with a changing payof{}f
\linebreak $\widehat{G}\left( A(T),k\left( t,A,v\right) \right)$ can be interpreted as a
single financial derivative with a new payof{}f $G\left( A(T),v(T) \right) $. In other words, the financial derivative from Theorem \ref{MainTheo} with payof{}f $D(A(T),v(T))=G(A(T),v(T))$ can be constructed from a continuum of derivatives with payof{}fs $\widehat{D}(A(T),k(t,A,v))=\widehat{G}(A(T),k(t,A,v))$ as prescribed in Corollary \ref{cor:heston_var_rho_nonzero_solution}, where $A(t) = Y^\ast(t)$ and $k(t, A(t), v(t)) = (k_{\varepsilon}(t, Y^\ast(t), v(t)), k_{v}(t, Y^\ast(t), v(t)))^\top =: (k_{\varepsilon, t}, k_{v,t})^\top,\,\tin$.}
\end{proof}

\bigskip

\begin{proof}[Proof of Corollary \ref{cor:heston_var_rho_nonzero_solution}]
Here we prove that for the Heston model and power-utility function there exist $D$ and $\lambda^v$ such that the VaR constraint is satisfied at $t=0$ and Conditions \eqref{cond:U_D_yy_y}, \eqref{cond:U_D_yv_y} and \eqref{cond:D_v} hold for all $\tin$. Then we apply Theorem \ref{MainTheo} to derive the optimal solution to \eqref{MainControlProb} and provide more explicit formulas for computing the optimal solution and the value function.

Recall that $D$ can be constructed, thanks to Proposition \ref{prop:equivalence_btw_sequence_and_single_D}, via a continuum of derivatives $
\widehat{D}$ depending only on the unconstrained optimal wealth process $Y^\ast$ and with time-changing (state-dependent) strike prices $k_{\varepsilon}$, $k_{v}$. Therefore, we show that at each $\tin$, the degrees of freedom $k_{\varepsilon}$ and $k_{v}$ of the payof{}f $\widehat{D}(\cdot;k_{\varepsilon}, k_{v})$ ensure the conditions necessary for the application of Theorem \ref{MainTheo}. For convenience, we state here the related payof{}f structure as per \eqref{eq:D_conjecture_Heston} and suppress the hat in $\widehat{D}$ to simplify the notation:
\begin{equation*}
    D\left(
Y^{\ast}(T)\right) =Y^{\ast}(T)+\left( K-Y^{\ast}(T)\right) 1_{\left\{ k_{\varepsilon}\leq Y^{\ast}(T)\leq
K\right\} }-\left( Y^{\ast}(T)-k_{v}\right) 1_{\left\{ k_{\varepsilon
}\leq Y^{\ast}(T)<k_{\varepsilon}\right\} }
\end{equation*}
 with $0 \leq k_{v}\leq k_{\varepsilon}\leq K$. Therefore, we can rewrite $D$ as follows:%
\begin{eqnarray}
D(y) &=&y+\left( K-y\right) 1_{\left\{ y \leq K\right\} }-\left( K-y\right)
1_{\left\{ y<k_{\varepsilon}\right\} }+\left( k_{v}-y\right) 1_{\left\{
y<k_{\varepsilon}\right\} } -\left( k_{v}-y\right) 1_{\left\{
y<k_{v}\right\} } \notag\\
&=&y+\left( K-y\right) 1_{\left\{ y \leq K\right\} }-\left( k_{\varepsilon
}-y\right) 1_{\left\{ y<k_{v}\right\} }-\left( K-k_{\varepsilon
}\right) 1_{\left\{ y<k_{\varepsilon}\right\} } \notag\\
&=:&D_{1}(y) + D_{2}(y) - D_{3}(y) - D_{4}(y) \notag
\end{eqnarray}

Observe that:
\begin{align*}
\left\{y \in \mathbb{R}: D(y) < K\right\} &=\left\{y \in \mathbb{R}: y+\left( K-y\right) 1_{\left\{ k_{\varepsilon}\leq
y\leq K\right\} }+\left( k_{v}-y\right) 1_{\left\{ k_{\varepsilon
}\leq y<k_{\varepsilon}\right\} } <  K\right\} =\left\{y \in \mathbb{R}: y < k_{\varepsilon}\right\}
\end{align*}%

We can also rewrite $\ubard(y):=\overline{U}(D(y))$ as follows:
\begin{eqnarray*}
\overline{U}(D(y)) &=&U(D(y))-\lambda _{\varepsilon }\rBrackets{1_{\left\{ D(y)<K\right\}
} - \varepsilon}\notag\\
&=&\frac{1}{\gamma }\left( y+\left( K-y\right) 1_{\left\{
y<K\right\} }-\left( k_{v}-y\right) 1_{\left\{
y<k_{v}\right\} }-\left( K-k_{v}\right) 1_{\left\{
y<k_{\varepsilon}\right\} }\right) ^{\gamma } -\lambda _{\varepsilon }1_{\left\{
y<k_{\varepsilon}\right\} }  +  {\color{black} \lambda _{\varepsilon }\varepsilon} \notag \\
&=&\frac{y^{\gamma }}{\gamma }+\frac{1}{\gamma }\left( K^{\gamma }-y^{\gamma }\right) 1_{\left\{ y \leq K\right\} }-\frac{1}{%
\gamma }\left( k_{v}^{\gamma }-y^{\gamma }\right) 1_{\left\{
y<k_{v}\right\} } -\frac{1}{\gamma }\left( \left( K^{\gamma}-k_{v}^{\gamma }\right) 1_{\left\{ y<k_{\varepsilon}\right\} } + \gamma \lambda _{\varepsilon }1_{\left\{ y<k_{\varepsilon}\right\} }\right) +  {\color{black} \lambda _{\varepsilon }\varepsilon } \notag \\
&=:&\ubardi{1}(y)+\ubardi{2} (y) - \ubardi{3}(y)-\ubardi{4}(y) + {\color{black} \lambda _{\varepsilon }\varepsilon} \notag
\end{eqnarray*}

\bigskip

The proof contains three Parts.\\

\begin{description}
\item[Part 1.] First, we show that Conditions \eqref{cond:U_D_yy_y} and \eqref{cond:U_D_yv_y} hold. By Lemma \ref{lem:sufficient_condition}, it is sufficient to show that \eqref{cond:SC_in_lemma} holds: $\udp_{y}=y^{\gamma -1} h(t) \exp \left(b(t)v\right) \dq_{y}$. This involves checking three cases, as the second and third terms are structurally the same, whereas the fifth term is independent of $y$:
    \begin{description}
    \item[Term 1] $D_1$ and $\overline{U}^{D}_1$,
    \item[Terms 2 and 3] $D_2$ and $\overline{U}^{D}_2$, $D_3$ amd $\overline{U}^{D}_3$. This involve writing the sufficient condition in terms of expectations leading to a new representation \eqref{eq:suffiecient_condition_as_financial_derivatives_SupMat}, then proving the equality via four steps:
        \begin{description}
        \item[Step 1] use FK theorem to derive the PDE of LHS of \eqref{eq:suffiecient_condition_as_financial_derivatives_SupMat};
        \item[Step 2] use FK theorem to derive the PDE of expectation term in the RHS of \eqref{eq:suffiecient_condition_as_financial_derivatives_SupMat};
        \item[Step 3] show that the terminal value of the LHS is equal to the value of the RHS, i.e., check that the terminal conditions of the corresponding PDEs are equal;
        \item[Step 4] show that
        RHS of \eqref{eq:suffiecient_condition_as_financial_derivatives_SupMat} solves the PDE for LHS of\linebreak \eqref{eq:suffiecient_condition_as_financial_derivatives_SupMat}.
        \end{description}
    \item[Term 4] $D_4$ and $\overline{U}^{D}_4$
    \end{description}
\item[Part 2.] Addressing Condition \eqref{cond:D_v}
\item[Part 3.] Addressing the VaR constraint and applying Theorem \ref{MainTheo}
\end{description}

\bigskip
We write for $i \in \{1, 2, 3, 4 \}$:
\begin{align*}
\overline{U}^{(i)}(t,y,v) &:= \EVtyv{P}{\ubardi{i}\rBrackets{Y^{\ast}(T)}};\\
D^{(i)}(t,y,v) & := \EVtyv{Q}{\exp\rBrackets{-r(T - t)}D_i\rBrackets{Y^{\ast}(T)}}.
\end{align*}

\textbf{Part 1. Term 1.\\}
For the first term of the modified utility function and the related (first) piece of the financial derivative on the optimal unconstrained wealth, we can check the sufficient condition \eqref{cond:SC_in_lemma} by explicitly calculating its LHS and RHS.

In LHS, $\overline{U}^{(1)}$ is the optimum of the objective function in \eqref{PowerUnconstProb}, which is known due to Proposition \ref{prop:unconstrained_problem_solution}:
\begin{equation*}
    \overline{U}^{(1)}=\frac{y^{\gamma }}{\gamma }\exp(a(t)+b(t)v) \quad \Rightarrow \quad \overline{U}_{y}^{(1)}=y^{\gamma -1}\exp (a(t)+b(t)v).
\end{equation*}

Regarding RHS, $D_{1}(y)=y$ and $\exp(-rt)Y^\ast(t)$ is a martingale under any EMM $\mathbb{Q}$. Thus, $D^{(1)}= y \Rightarrow D^{(1)}_{y}= 1$ and we conclude that for any $\rho \in [-1,1]$ and any $\mathbb{Q}$ the following holds:
\begin{equation*}
    \overline{U}_{y}^{(1)}=y^{\gamma -1}\exp (a(t)+b(t)v) \cdot 1= y^{\gamma -1} \underbrace{\exp \left(a(t)\right)}_{=h(t)}\exp\left( b(t)v\right) D_{y}^{(1)}.
\end{equation*}

\textbf{Part 1. Terms 2 and 3.\\}
We show now that the same relation holds for the second and third terms of the modified utility function, i.e., the utility of a put option on the unconstrained optimal wealth is linked to a price under the suitable $\mathbb{Q}(\lambda^{v})$ of a put option on the unconstrained optimal wealth.  For simplicity of presentation, we will write $\mathbb{Q}$ instead of $\mathbb{Q}(\lambda^v)$.

Recall that the expected values can be computed via the inverse Fourier
transform:%
\begin{align*}
\mathbb{E}_{t,z,v}^{\mathbb{M}}\left[ g\left( Z^\ast(T)\right) \right] &=\int g\left(
x\right) \left( \frac{1}{2\pi }\int \exp\rBrackets{-iux}\phi^{Z^\ast(T),\mathbb{M}} (u;t,z,v)du\right) dx \notag \\
&=\frac{1}{2\pi }\int \int g\left( x\right) \exp \left( -iu\left(
x-z\right) +A^{\mathbb{M}}(T-t,u) + B^{\mathbb{M}}(T-t,u)v\right) dudx,  \notag
\end{align*}%
where $\phi^{Z^\ast(T),\mathbb{M}}$ is the characteristic function of $Z$ under the measure $\mathbb{M} \in \{{\mathbb{P}}, {\mathbb{Q}} \}$ given in Proposition \ref{prop:characteristic_fct_log_Y}.

Changing variables, $Z^\ast(T) = \ln \rBrackets{Y^{\ast}(T)}$, $z=x-\ln y$, and using the inverse
Fourier Transform of $Z^\ast(T)$, we obtain $\forall \, i \in \{1,2,3, 4\}$:%
\begin{align}
\overline{U}^{(i)} &= \mathbb{E}^{\mathbb{P}}_{t,y,v}\left[ \overline{U}^{D}_{i}(Y^{\ast}(T))\right] \notag\\
& =\int \overline{U}^{D}_{i}(\exp\rBrackets{x})\left( \frac{1}{2\pi }\int \exp\rBrackets{-iux}\phi^{Z^\ast(T),\mathbb{P}} (u;t,\ln y,v)du\right) dx
\notag \\
&= \frac{1}{2\pi }\int \int \overline{U}^{D}_{i}(\exp\rBrackets{x})\exp \Bigl( -iu\left( x-\ln y\right)
+A^{\mathbb{P}}(T-t,u) +B^{\mathbb{P}}(T-t,u)v\Bigr) du\,dx \notag \\
&= \frac{1}{2\pi }\int \int \overline{U}^{D}_{i}(y\exp\rBrackets{z})\exp \left(
-iuz+A^{\mathbb{P}}(T-t,u)+B^{\mathbb{P}}(T-t,u)v\right) du\,dz \label{eq:U_i_Fourier_Transform_SupMat}\\
D^{(i)} &= \exp\rBrackets{-r(T-t)}\mathbb{E}_{t,y,v}^{\mathbb{Q}}\left[ D_{i}(Y^{\ast}(T))\right] \notag\\ 
&= \frac{\exp\rBrackets{-r(T-t)}}{2\pi } \cdot \int \int D_{i}(y\exp\rBrackets{z})\exp \left(
-iuz+A^{\mathbb{Q}}(T-t,u)+B^{\mathbb{Q}}(T-t,u)v\right) du\,dz \label{eq:Pi_i_Fourier_Transform_SupMat}.
\end{align}%

For $\overline{U}^{D}_{2}(y)=\frac{1}{\gamma }\left( K^{\gamma }-y^{\gamma }\right)
1_{\left\{ y<K\right\} }$ with $K > 0$ a given parameter, we receive, using \eqref{eq:U_i_Fourier_Transform_SupMat}:%
\begin{eqnarray*}
\overline{U}^{(2)} 
&=&\frac{1}{2\pi }\frac{1}{\gamma }\int \int \left( K^{\gamma }-\exp\rBrackets{\gamma
(z+\ln y)}\right) 1_{\left\{ z<\ln K-\ln y\right\} }  \cdot \exp \left( -iuz+A^{\mathbb{P}}(T-t,u)+B^{\mathbb{P}}(T-t,u)v\right) dudz \\
&=&\frac{1}{2\pi }\frac{K^{\gamma }}{\gamma }\int \int 1_{\left\{ z<\ln
K-\ln y\right\} }\exp \Bigl( -iuz+A^{\mathbb{P}}(T-t,u) +B^{\mathbb{P}}(T-t,u)v\Bigr) dudz \\
&&-\frac{1}{2\pi }\frac{1}{\gamma }\int \int 1_{\left\{ z<\ln K-\ln
y\right\} }\exp \Bigl( \gamma \ln y+\gamma
z-iuz+A^{\mathbb{P}}(T-t,u) + B^{\mathbb{P}}(T-t,u)v\Bigr) dudz\\
&=&\frac{1}{2\pi }\frac{K^{\gamma }}{\gamma }\int \limits_{-\infty}^{\ln\rBrackets{K/y}} \int\limits_{-\infty}^{+\infty} \exp \left( -iuz+A^{\mathbb{P}}(T-t,u)+B^{\mathbb{P}}(T-t,u)v\right) dudz \\
&&-\frac{1}{2\pi }\frac{1}{\gamma } \int \limits_{-\infty}^{\ln\rBrackets{K/y}} \int\limits_{-\infty}^{+\infty} y^{\gamma}\exp \left(\gamma
z-iuz+A^{\mathbb{P}}(T-t,u)+B^{\mathbb{P}}(T-t,u)v\right) dudz
\end{eqnarray*}

Next we state the Leibniz integral rule (LIR), as we will use it several times. For $g(\alpha, x), l(\alpha), m(\alpha)$ continuously differentiable functions it holds:
\begin{align}\label{eq:Leibniz_integral_rule}\tag{LIR}
    \frac{\partial}{\partial \alpha} \rBrackets{\int \limits_{l(\alpha)}^{m(\alpha)}g(\alpha,x) \,dx} = &g(\alpha,m(\alpha))m'(\alpha) - g(\alpha,l(\alpha))l'(\alpha) +  \rBrackets{\int \limits_{l(\alpha)}^{m(\alpha)} \rBrackets{\frac{\partial}{\partial \alpha}g(\alpha,x)} \,dx}.\notag
\end{align}

Taking the derivative of $\overline{U}^{(2)}$ yields:
\begin{flalign*}
    \overline{U}^{(2)}_y &= \ddy \Biggl({\frac{1}{2\pi }\frac{K^{\gamma }}{\gamma } \int \limits_{-\infty}^{\ln\rBrackets{K/y}} \underbrace{\int \limits_{-\infty}^{+\infty} \exp \left( -iuz+A^{\mathbb{P}}(T-t,u)+B^{\mathbb{P}}(T-t,u)v\right) du}_{=:g_1(y, z)} dz}\Biggr) \\
    & \quad - \ddy \Biggl({\frac{1}{2\pi }\frac{1}{\gamma }\int \limits_{-\infty}^{\ln\rBrackets{K/y}} \underbrace{\int \limits_{-\infty}^{+\infty} y^{\gamma} \exp \left(\gamma z - iuz+A^{\mathbb{P}}(T-t,u)+B^{\mathbb{P}}(T-t,u)v\right) du}_{=:g_2(y, z)} dz}\Biggr)\\
     & \stackrel{LIR}{=} \frac{1}{2\pi }\frac{K^{\gamma }}{\gamma } \Biggl( g_1(y, \ln\rBrackets{K/y}) \rBrackets{-\frac{1}{y}} - \lim_{c \downarrow -\infty} \biggl(  g_1(y, c) \underbrace{\frac{\partial c}{\partial y} }_{=0}\biggr)  + \int \limits_{-\infty}^{\ln\rBrackets{K/y}} \underbrace{\ddy g_1(y, z) }_{=0}dz \Biggr)\\
    & \quad - \frac{1}{2\pi }\frac{1}{\gamma } \Biggl( g_2(y, \ln\rBrackets{K/y})\rBrackets{-\frac{1}{y}} - \lim_{c \downarrow -\infty} \biggl(  g_2(y, c) \underbrace{\frac{\partial c}{\partial y} }_{=0}\biggr) + \int \limits_{-\infty}^{\ln\rBrackets{K/y}}  \ddy g_2(y, z)dz  \Biggr)\\
    & = - \frac{1}{2\pi }\frac{K^{\gamma }}{\gamma }\frac{1}{y} \int \limits_{-\infty}^{+\infty} \exp \left( -iu \ln\rBrackets{K/y}+A^{\mathbb{P}}(T-t,u)+B^{\mathbb{P}}(T-t,u)v\right) du \\
    &\quad + \frac{1}{2\pi }\frac{1}{\gamma } y^{\gamma - 1} \int \limits_{-\infty}^{+\infty} \exp \left( \gamma \ln \rBrackets{K/y} - iu\ln\rBrackets{K/y} + A^{\mathbb{P}}(T-t,u)+B^{\mathbb{P}}(T-t,u)v\right) du\\
    &\quad - \frac{1}{2\pi }\frac{1}{\gamma } \int \limits_{-\infty}^{\ln\rBrackets{K/y}} \int \limits_{-\infty}^{+\infty} \gamma y^{\gamma - 1} \exp \left(  \gamma z - iuz+A^{\mathbb{P}}(T-t,u)+B^{\mathbb{P}}(T-t,u)v\right) du\,dz\\
    & \stackrel{(a)}{=} \cancel{-\frac{1}{2\pi }\frac{K^{\gamma }}{\gamma }\frac{1}{y} \int \limits_{-\infty}^{+\infty} \exp \left( -iu \ln\rBrackets{K/y}+A^{\mathbb{P}}(T-t,u)+B^{\mathbb{P}}(T-t,u)v\right) du} \\
    &\quad + \cancel{\frac{1}{2\pi }\frac{K^{\gamma }}{\gamma }\frac{1}{y} \int \limits_{-\infty}^{+\infty} \exp \left( -iu \ln\rBrackets{K/y}+A^{\mathbb{P}}(T-t,u)+B^{\mathbb{P}}(T-t,u)v\right) du}\\
    &\quad - \frac{1}{2\pi }\cancel{\frac{1}{\gamma }} \int \limits_{-\infty}^{\ln\rBrackets{K/y}} \int \limits_{-\infty}^{+\infty} \cancel{\gamma} y^{\gamma - 1} \exp \left( \gamma z -iuz+A^{\mathbb{P}}(T-t,u)+B^{\mathbb{P}}(T-t,u)v\right) du\,dz\\
    & = -\frac{y^{\gamma -1}}{2\pi }\int \int 1_{\left\{ z<\ln K-\ln y\right\} }\exp \left(\gamma z-iuz+A^{\mathbb{P}}(T-t,u)+B^{\mathbb{P}}(T-t,u)v\right) dudz,
\end{flalign*}
where in (a) we used $\exp\rBrackets{\gamma \ln(K/y)}=\frac{K^{\gamma}}{y^{\gamma}}$.

Next we reconstruct the stochastic representation of $\overline{U}^{(2)}_{y}$:
\begin{equation*}
    \begin{aligned}
        \overline{U}^{(2)}_y &\stackrel{x = z + \ln(y)}{=} -\frac{y^{\gamma - 1 }}{2\pi }\int \int 1_{\left\{ x<\ln K\right\} }\exp \Bigl(
  \gamma (x -\ln(y)) - iu(x - \ln(y)) +A^{\mathbb{P}}(T-t,u)+B^{\mathbb{P}}(T-t,u)v\Bigr) dudx \\
        & \stackrel{-\gamma \ln(y) = \ln\left( y^{-\gamma} \right)}{=} -\frac{y^{\gamma - 1 }y^{-\gamma}}{2\pi }\int \int 1_{\left\{ x<\ln K\right\} }\exp \left(\gamma x - iux\right) \exp \Bigl(i u \ln(y) +A^{\mathbb{P}}(T-t,u) \\
        & \qquad \qquad + B^{\mathbb{P}}(T-t,u)v \Bigr) dudx \\
        & = -y^{- 1 } \int  1_{\left\{ x<\ln K\right\} }\exp \left(
        \gamma x\right) \underbrace{\left(\frac{1}{2\pi }\int \exp \left(- iux\right) \phi^{Z^\ast(T),\mathbb{P}}(u;t, \ln(y), v) du \right)}_{\text{density of}\,Z^\ast(T)\,\text{evaluated at}\, x} dx \\
        & = -y^{-1} \mathbb{E}^{\mathbb{P}}\left[\exp\rBrackets{\gamma Z^\ast(T)} 1_{\left\{ Z^\ast(T) <\ln K\right\} } | Z^\ast(t) = \ln(y), v(t) = v \right] \\
        & \stackrel{Z^\ast(t) := \ln(Y^\ast(t))}{=}-y^{-1} \mathbb{E}^{\mathbb{P}}\left[(Y^{\ast}(T))^{\gamma} 1_{\left\{ Y^{\ast}(T) < K\right\} } | Y^\ast(t) = y, v(t) = v \right].
    \end{aligned}
\end{equation*}

Applying the previous result for $\gamma = 1$ under the measure $\mathbb{Q}$ instead of $\mathbb{P}$, we receive the following expression for $D_{2}(x)=\left( K-y\right) 1_{\left\{
y<K\right\} }$ with $K > 0$ a given parameter:

\begin{eqnarray*}
D^{(2)} &=&\frac{1}{2\pi }\exp\rBrackets{-r(T-t)}\int \int D_2(y\exp\rBrackets{z})\exp \left( -iuz+A^{\mathbb{Q}}(T-t,u)+B^{\mathbb{Q}}(T-t,u)v\right) dudz \\
&=& -y^{-1} \mathbb{E}^{\mathbb{Q}}\left[\exp\rBrackets{-r(T - t)} Y^{\ast}(T) 1_{\left\{ Y^{\ast}(T) < K\right\} } | Y^{\ast}(t) = y, v(t) = v \right]
\end{eqnarray*}

Therefore, proving Condition \eqref{cond:SC_in_lemma} for the second and the third terms of the auxiliary utility function is equivalent to proving the following condition:
\begin{align*}
    -y^{-1} \mathbb{E}^{\mathbb{P}}&\left[\rBrackets{Y^{\ast}(T)}^{\gamma} 1_{\left\{ Y^{\ast}(T) < K\right\} } | Y^{\ast}(t) = y, v(t) = v \right] \stackrel{!}{=}  y^{\gamma - 1} \exp(a(t) + b(t)v) \left(-y^{-1}\right) \\
    & \cdot \mathbb{E}^{\mathbb{Q}}\left[\exp\rBrackets{-r(T - t)} Y^{\ast}(T) 1_{\left\{ Y^{\ast}(T) < K\right\} } | Y^{\ast}(t) = y, v(t) = v \right],
\end{align*}
which, in turn, is equivalent to the following one:
\begin{equation}\label{eq:suffiecient_condition_as_financial_derivatives_SupMat}\tag{ESC Put}
\begin{aligned}
    &\underbrace{\mathbb{E}^{\mathbb{P}}_{t,y,v}\left[\rBrackets{Y^{\ast}(T)}^{\gamma} 1_{\left\{ Y^{\ast}(T) < K\right\} } \right]}_{=:g^{\mathbb{P}}(t,y, v)} \\
    & \qquad  \stackrel{!}{=}  y^{\gamma - 1} \exp(a(t) + b(t)v)  \underbrace{\mathbb{E}^{\mathbb{Q}}_{t,y,v}\left[\exp\rBrackets{-r(T - t)} Y^{\ast}(T) 1_{\left\{ Y^{\ast}(T) < K\right\} }  \right]}_{=:g^{\mathbb{Q}}(t,y, v)}
    \end{aligned}
\end{equation}
ESC stands for equivalent sufficient condition.

We prove now \eqref{eq:suffiecient_condition_as_financial_derivatives_SupMat} via four steps.

\bigskip

\textbf{Part 1. Terms 2 and 3. Step 1. FK PDE for LHS of \eqref{eq:suffiecient_condition_as_financial_derivatives_SupMat}}

Recall that under the measure $\mathbb{P}$ we have:
\begin{align*}
dY^{\ast}(t)&=Y^{\ast}(t)\left[ \left( r+ \piu(t) \overline{\lambda }v(t)\right)
dt+\piu(t) \sqrt{v(t)}dW^{\mathbb{P}}_{1}(t)\right];\\
d v(t)& =\kappa \left( \theta -v(t)\right) dt+\sigma \rho \sqrt{v(t)}\,d W^{\mathbb{P}}_{1}(t) \,+\sigma \sqrt{v(t)}\sqrt{1-\rho ^{2}}d W^{\mathbb{P}}_{2}(t).
\end{align*}%
with $\piu(t) = \frac{\overline{\lambda }}{1-\gamma }+%
\frac{\sigma \rho b(t)}{1-\gamma}$.

Then $\mathbb{E}^{\mathbb{P}}_{t,y,v}\left[\rBrackets{Y^{\ast}(T)}^{\gamma} 1_{\left\{ Y^{\ast}(T) < K\right\} } \right] = g^{\mathbb{P}}(t, y, v)$ has the following FK representation:
\begin{align*}
0 &= g^{\mathbb{P}}_{t}  + y(r + \piu(t) \overline{\lambda }v)g^{\mathbb{P}}_{y} + \kappa \left( \theta -v\right)g^{\mathbb{P}}_{v} + \frac{1}{2} v y^2 \rBrackets{\piu(t)}^2 g^{\mathbb{P}}_{yy}  + \frac{1}{2} v \sigma^2 g^{\mathbb{P}}_{vv} + \rho \sigma y v \piu(t) g^{\mathbb{P}}_{yv};\\
y^{\gamma} 1_{\left\{ y < K\right\} } &= g^{\mathbb{P}}(T, y, v).
\end{align*}%

\bigskip
\textbf{Part 1. Terms 2 and 3. Step 2. FK PDE for $\mathbb{Q}$-expectation in RHS of \eqref{eq:suffiecient_condition_as_financial_derivatives_SupMat}}

Recall that under the measure $\mathbb{Q}$ we have:
\begin{eqnarray*}
dY^{\ast}(t) &=&Y^{\ast}(t)rdt+Y^{\ast}(t)\piu(t)\sqrt{v(t)}dW^{\mathbb{Q}}_{1}(t) \\
d v(t) &=&\tilde{\kappa}\left( \tilde{\theta}-v(t)\right) dt+\sigma \sqrt{%
v(t)}\rho d W^{\mathbb{Q}}_{1}(t)+\sigma \sqrt{v(t)}\sqrt{1-\rho ^{2}}d W^{\mathbb{Q}}_{2}(t)
\end{eqnarray*}

Then $\mathbb{E}^{\mathbb{Q}}\left[\exp\rBrackets{-r(T - t)} Y^{\ast}(T) 1_{\left\{ Y^{\ast}(T) < K\right\} } | Y^{\ast}(t) = y, v(t) = v \right] = g^{\mathbb{Q}}(T, y, v)$ has the following FK representation:
\begin{align*}
0 &= g^{\mathbb{Q}}_{t} - r g^{\mathbb{Q}} + y r g^{\mathbb{Q}}_{y} + \tilde{\kappa} \left( \tilde{\theta} -v\right)g^{\mathbb{Q}}_{v} + \frac{1}{2} v y^2 \rBrackets{\piu(t)}^2 g^{\mathbb{Q}}_{yy}  + \frac{1}{2} v \sigma^2 g^{\mathbb{Q}}_{vv} + \rho \sigma y v \piu(t) g^{\mathbb{Q}}_{yv};\\
y 1_{\left\{ y < K\right\} } &= g^{\mathbb{Q}}(T, y, v).
\end{align*}%

\bigskip
\textbf{Part 1. Terms 2 and 3. Step 3. Equality of terminal conditions}$\,$

Consider the ansatz $g^{\mathbb{P}}(t,y,v) = y^{\gamma - 1} \exp(a(t) + b(t)v)g^{\mathbb{Q}}(t,y,v)$ with $a(T) = b(T) = 0$. Then:
\begin{align*}
    g^{\mathbb{P}}(T, y, v) &= y^{\gamma} 1_{\left\{ y < K\right\} } = y^{\gamma - 1}  y 1_{\left\{ y < K\right\}} = y^{\gamma - 1} y 1_{\left\{ y < K\right\}} \exp(a(T) + b(T)v) \\
    & = y^{\gamma - 1} \exp(a(T) + b(T)v) g^{\mathbb{Q}}(T,y, v),
\end{align*}
i.e., the LHS and RHS coincide at time $t=T$.

\bigskip
\textbf{Part 1. Terms 2 and 3. Step 4. Verifying $g^{\mathbb{P}}(t,y,v) = y^{\gamma - 1} \exp(a(t) + b(t)v)g^{\mathbb{Q}}(t,y,v)$ via PDEs}

 Let us calculate the necessary partial derivatives of $g^{\mathbb{P}}$, which appear in its FK PDE:
\begin{flalign*}
    g^{\mathbb{P}}_{t} &=\frac{\partial}{\partial t} \rBrackets{ y^{\gamma - 1} \exp(a(t) + b(t)v)g^{\mathbb{Q}}(t,y,v) } = y^{\gamma - 1} \exp(a(t) + b(t)) \rBrackets{a'(t)+b'(t)v}g^{\mathbb{Q}} \\
    & \quad +  y^{\gamma - 1} \exp(a(t) + b(t))g^{\mathbb{Q}}_{t} = y^{\gamma - 1} \exp(a(t) + b(t))\rBrackets{ \rBrackets{a'(t)+b'(t)v}g^{\mathbb{Q}} + g^{\mathbb{Q}}_{t}}\\
    g^{\mathbb{P}}_{y} & = \exp(a(t) + b(t)v)\rBrackets{(\gamma - 1) y^{\gamma - 2}g^{\mathbb{Q}} + y^{\gamma - 1} g^{\mathbb{Q}}_{y}} = y^{\gamma - 2} \exp(a(t) + b(t)v) \rBrackets{(\gamma - 1) g^{\mathbb{Q}} + y g^{\mathbb{Q}}_{y}};\\
    g^{\mathbb{P}}_{v}
    & = y^{\gamma - 1}\rBrackets{\frac{\partial \exp(a(t) + b(t)v)}{\partial v}  g^{\mathbb{Q}} + \exp(a(t) + b(t)v) g^{\mathbb{Q}}_{v}} = y^{\gamma - 1}  \exp(a(t) + b(t)v) \rBrackets{b(t) g^{\mathbb{Q}} + g^{\mathbb{Q}}_v};\\
    g^{\mathbb{P}}_{yy} &= \frac{\partial}{\partial y} \rBrackets{ g^{\mathbb{P}}_{y}} = \exp(a(t) + b(t)v) \frac{\partial}{\partial y} \rBrackets{(\gamma - 1) y^{\gamma - 2}g^{\mathbb{Q}} + y^{\gamma - 1} g^{\mathbb{Q}}_{y} }\\
    & = \exp(a(t) + b(t)v)\rBrackets{(\gamma - 1)\rBrackets{(\gamma - 2)y^{\gamma - 3}g^{\mathbb{Q}} + y^{\gamma - 2}g^{\mathbb{Q}}_{y}} + \rBrackets{(\gamma - 1)y^{\gamma - 2}g^{\mathbb{Q}}_{y}} + y^{\gamma - 1}g^{\mathbb{Q}}_{yy}}\\
    & = y^{\gamma - 3} \exp(a(t) + b(t)v)  \rBrackets{(\gamma - 1)(\gamma - 2)g^{\mathbb{Q}} + 2(\gamma - 1) y g^{\mathbb{Q}}_{y} + y^2g^{\mathbb{Q}}_{yy}}\\
    g^{\mathbb{P}}_{vv} &= \frac{\partial}{\partial v} \rBrackets{ g^{\mathbb{P}}_{v}} = \frac{\partial}{\partial v} \rBrackets{ y^{\gamma - 1}  \exp(a(t) + b(t)v) \rBrackets{b(t) g^{\mathbb{Q}} + g^{\mathbb{Q}}_v}}&&\\
    & = y^{\gamma - 1} \rBrackets{ \exp(a(t) + b(t)v) b(t)\rBrackets{b(t) g^{\mathbb{Q}} + g^{\mathbb{Q}}_v} + \exp(a(t) + b(t)v)\rBrackets{b(t) g^{\mathbb{Q}}_v + g^{\mathbb{Q}}_{vv}}}&&\\
    & = y^{\gamma - 1} \exp(a(t) + b(t)v) \rBrackets{  \rBrackets{b(t)}^2  g^{\mathbb{Q}} + 2 b(t) g^{\mathbb{Q}}_v + g^{\mathbb{Q}}_{vv} } &&
    \end{flalign*}
    \begin{flalign*}
    g^{\mathbb{P}}_{yv} &= \frac{\partial}{\partial y} \rBrackets{ g^{\mathbb{P}}_{v}} = \frac{\partial}{\partial y} \rBrackets{ y^{\gamma - 1}  \exp(a(t) + b(t)v) \rBrackets{b(t) g^{\mathbb{Q}} + g^{\mathbb{Q}}_v} } &&\\
    & = \exp(a(t) + b(t)v) \rBrackets{ (\gamma - 1) y^{\gamma - 2} \rBrackets{b(t) g^{\mathbb{Q}} + g^{\mathbb{Q}}_v} + y^{\gamma - 1} \rBrackets{b(t) g^{\mathbb{Q}}_{y} + g^{\mathbb{Q}}_{yv}}} &&\\
    & = y^{\gamma - 2} \exp(a(t) + b(t)v) \rBrackets{ (\gamma - 1)  b(t) g^{\mathbb{Q}} + (\gamma - 1)  g^{\mathbb{Q}}_{v} + y b(t) g^{\mathbb{Q}}_{y} + y g^{\mathbb{Q}}_{yv}}.
\end{flalign*}

We plug those partial derivatives in the LHS PDE, i.e., FK PDE of $g^{\mathbb{P}}$, and get:
\begin{align*}
    0 &= y^{\gamma - 1} \exp(a(t) + b(t))\rBrackets{ \rBrackets{a'(t)+b'(t)v}g^{\mathbb{Q}} + g^{\mathbb{Q}}_{t}}  \\
    & \quad + y(r + \piu(t) \overline{\lambda }v) y^{\gamma - 2} \exp(a(t) + b(t)v) \rBrackets{(\gamma - 1) g^{\mathbb{Q}} + y g^{\mathbb{Q}}_{y}} \\
     &\quad + \kappa \left( \theta -v\right)y^{\gamma - 1}  \exp(a(t) + b(t)v) \rBrackets{b(t) g^{\mathbb{Q}} + g^{\mathbb{Q}}_v} \\
     &\quad + \frac{1}{2} v y^2 \rBrackets{\piu(t)}^2 y^{\gamma - 3} \exp(a(t) + b(t)v)  \rBrackets{(\gamma - 1)(\gamma - 2)g^{\mathbb{Q}} + 2(\gamma - 1) y g^{\mathbb{Q}}_{y} + y^2g^{\mathbb{Q}}_{yy}}\\
     &\quad + \frac{1}{2} v \sigma^2 y^{\gamma - 1} \exp(a(t) + b(t)v) \rBrackets{  \rBrackets{b(t)}^2  g^{\mathbb{Q}} + 2 b(t) g^{\mathbb{Q}}_v + g^{\mathbb{Q}}_{vv} }\\
     &\quad + \rho \sigma y v \piu(t) y^{\gamma - 2} \exp(a(t) + b(t)v) \rBrackets{ (\gamma - 1)  b(t) g^{\mathbb{Q}} + (\gamma - 1)  g^{\mathbb{Q}}_{v} + y b(t) g^{\mathbb{Q}}_{y} + y g^{\mathbb{Q}}_{yv}}.
\end{align*}

Since $\forall\, y > 0 , v>0$, we have $y^{\gamma - 1} \exp(a(t) + b(t)v) > 0$ and can divide by this term both sides of the PDE:
\begin{align*}
    0 &=  \rBrackets{a'(t)+b'(t)v}g^{\mathbb{Q}} + \underline{g^{\mathbb{Q}}_{t}}  + (r + \piu(t) \overline{\lambda }v) \rBrackets{(\gamma - \underline{1) g^{\mathbb{Q}}} + \underline{y g^{\mathbb{Q}}_{y}}}&\\
     &\quad + \kappa \left( \theta -v\right) \rBrackets{b(t) g^{\mathbb{Q}} + \underline{g^{\mathbb{Q}}_v}} &\\
     &\quad + \frac{1}{2} v  \rBrackets{\piu(t)}^2  \rBrackets{(\gamma - 1)(\gamma - 2)g^{\mathbb{Q}} + 2(\gamma - 1) y g^{\mathbb{Q}}_{y} + \underline{y^2g^{\mathbb{Q}}_{yy}}} &\\
     &\quad + \frac{1}{2} v \sigma^2  \rBrackets{  \rBrackets{b(t)}^2  g^{\mathbb{Q}} + 2 b(t) g^{\mathbb{Q}}_v + \underline{g^{\mathbb{Q}}_{vv}} } &\\
     &\quad + \rho \sigma  v \piu(t)  \rBrackets{ (\gamma - 1)  b(t) g^{\mathbb{Q}} + (\gamma - 1)  g^{\mathbb{Q}}_{v} + y b(t) g^{\mathbb{Q}}_{y} + \underline{y g^{\mathbb{Q}}_{yv}}}, &
\end{align*}
where we underlined terms related to the $g^{\mathbb{Q}}$ PDE. Collecting these terms, we get:
\begin{align*}
    0 &=  \rBrackets{a'(t)+b'(t)v}g^{\mathbb{Q}} + r\gamma g^{\mathbb{Q}} + \piu(t) \overline{\lambda }v \rBrackets{(\gamma - 1) g^{\mathbb{Q}} + y g^{\mathbb{Q}}_{y} }\\
    &\quad + \kappa \left( \theta -v\right) b(t) g^{\mathbb{Q}} \\
    &\quad + \frac{1}{2} v  \rBrackets{\piu(t)}^2  \rBrackets{(\gamma - 1)(\gamma - 2)g^{\mathbb{Q}} + 2(\gamma - 1) y g^{\mathbb{Q}}_{y}}\\
    &\quad + \frac{1}{2} v \sigma^2  \rBrackets{  \rBrackets{b(t)}^2  g^{\mathbb{Q}} + 2 b(t) g^{\mathbb{Q}}_v }\\
    &\quad + \rho \sigma  v \piu(t)  \rBrackets{ (\gamma - 1)  b(t) g^{\mathbb{Q}} + (\gamma - 1)  g^{\mathbb{Q}}_{v} + y b(t) g^{\mathbb{Q}}_{y} }\\
    &\quad + \biggl[ \underline{g^{\mathbb{Q}}_{t} - r g^{\mathbb{Q}} + r y g^{\mathbb{Q}}_{y} +\kappa \left( \theta -v\right) g^{\mathbb{Q}}_v + \frac{1}{2} v   \rBrackets{\piu(t)}^2 y^2 g^{\mathbb{Q}}_{yy} + \frac{1}{2} v \sigma^2 g^{\mathbb{Q}}_{vv} + \rho \sigma  v \piu(t) y g^{\mathbb{Q}}_{yv} }\biggr].
\end{align*}
Next we use the link between the variance process parameters under the different measures according to \eqref{eq:Heston_v_under_Q}:
\begin{equation*}
    \kappa \left( \theta -v\right) \stackrel{(i)}{=} \tilde{\kappa} \tilde{\theta}- \kappa v \stackrel{(ii)}{=} \tilde{\kappa} \tilde{\theta}- \rBrackets{\tilde{\kappa} - \sigma \lambar \rho - \sigma \lambda^{v}\sqrt{1 - \rho^2}}v = \tilde{\kappa} \left( \tilde{\theta} -v\right) + \sigma \lambar \rho v + \sigma \lambda^{v}\sqrt{1 - \rho^2}v,
\end{equation*}
where $(i)$ refers to  $\tilde{\theta} = \theta \kappa / \tilde{\kappa}$, $(ii)$  refers to $\tilde{\kappa} = \kappa + \sigma \lambar \rho + \sigma \lambda^{v}\sqrt{1 - \rho^2}$. Taking this as well as PDE of $g^{\mathbb{Q}}$ into account, we get:
\begin{align*}
    0 &=  \rBrackets{a'(t)+b'(t)v}g^{\mathbb{Q}} + r\gamma g^{\mathbb{Q}} + \piu(t) \overline{\lambda }v \rBrackets{(\gamma - 1) g^{\mathbb{Q}} + y g^{\mathbb{Q}}_{y} }\\
    &\quad + \kappa \left( \theta -v\right) b(t) g^{\mathbb{Q}} \\
    &\quad + \frac{1}{2} v  \rBrackets{\piu(t)}^2  \rBrackets{(\gamma - 1)(\gamma - 2)g^{\mathbb{Q}} + 2(\gamma - 1) y g^{\mathbb{Q}}_{y}}\\
    &\quad + \frac{1}{2} v \sigma^2  \rBrackets{  \rBrackets{b(t)}^2  g^{\mathbb{Q}} + 2 b(t) g^{\mathbb{Q}}_v }\\
    &\quad + \rho \sigma  v \piu(t)  \rBrackets{ (\gamma - 1)  b(t) g^{\mathbb{Q}} + (\gamma - 1)  g^{\mathbb{Q}}_{v} + y b(t) g^{\mathbb{Q}}_{y} }\\
    &\quad + \sigma \lambar \rho v g^{\mathbb{Q}}_v + \sigma \lambda^{v}\sqrt{1 - \rho^2} v g^{\mathbb{Q}}_v.
\end{align*}

Using the ODEs for $a(\tau), b(\tau)$ from \eqref{eq:PDE_a_tau} \eqref{eq:PDE_B(t)au} and the relation $\tau = T - t$, we conclude that:
\begin{align}
a^{\prime }(t )& =- \kappa \theta b(t) - \gamma r; \notag \\
b^{\prime }(t)& = - \frac{1}{2}{\left( \sigma ^{2}+\frac{\gamma \sigma
^{2}\rho ^{2}}{1-\gamma }\right) }b^{2}(t) + {\left( \kappa -%
\frac{\gamma \overline{\lambda }\sigma \rho }{1-\gamma }\right) }b(t) - %
\frac{1}{2}{\frac{\gamma \overline{\lambda }^{2}}{1-\gamma }}.  \notag
\end{align}%

Plugging the representation of $a'(t)$ and $b'(t)$ in the key relation we want to prove, we get:
\begin{align*}
    0 &=  \rBrackets{- \kappa \theta b(t) - \gamma r + v\cdot \rBrackets{- \frac{1}{2}\left( \sigma ^{2}+\frac{\gamma \sigma^{2} \rho^{2}}{1-\gamma }\right)(b(t))^2 + \left( \kappa - \frac{\gamma \overline{\lambda }\sigma \rho }{1-\gamma }\right)b(t) - \frac{1}{2} \frac{\gamma \overline{\lambda }^{2}}{1-\gamma } }}g^{\mathbb{Q}}\\
    &\quad + r\gamma g^{\mathbb{Q}} + \piu(t) \overline{\lambda }v \rBrackets{(\gamma - 1) g^{\mathbb{Q}} + y g^{\mathbb{Q}}_{y} }\\
    &\quad + \kappa \left( \theta -v\right) b(t) g^{\mathbb{Q}} \\
    &\quad + \frac{1}{2} v  \rBrackets{\piu(t)}^2  \rBrackets{(\gamma - 1)(\gamma - 2)g^{\mathbb{Q}} + 2(\gamma - 1) y g^{\mathbb{Q}}_{y}}\\
    &\quad + \frac{1}{2} v \sigma^2  \rBrackets{  \rBrackets{b(t)}^2  g^{\mathbb{Q}} + 2 b(t) g^{\mathbb{Q}}_v }\\
    &\quad + \rho \sigma  v \piu(t)  \rBrackets{ (\gamma - 1)  b(t) g^{\mathbb{Q}} + (\gamma - 1)  g^{\mathbb{Q}}_{v} + y b(t) g^{\mathbb{Q}}_{y} }\\
    &\quad + \sigma \lambar \rho v g^{\mathbb{Q}}_v + \sigma \lambda^{v}\sqrt{1 - \rho^2} v g^{\mathbb{Q}}_v.
\end{align*}
Next we indicate terms to be cancelled out directly and plug in the representation of $\piu(t) = \frac{\overline{\lambda }}{1-\gamma } + \frac{\sigma \rho b(t)}{1-\gamma}$:
\begin{align*}
    0 &=  \rBrackets{- \cancel{\kappa \theta b(t)} - \cancel{\gamma r} + v\cdot \rBrackets{- \frac{1}{2}\left(  \cancel{\sigma ^{2}}+\frac{\gamma \sigma^{2} \rho^{2}}{1-\gamma }\right)(b(t))^2 + \left( \cancel{\kappa} - \frac{\gamma \overline{\lambda }\sigma \rho }{1-\gamma }\right)b(t) - \frac{1}{2} \frac{\gamma \overline{\lambda }^{2}}{1-\gamma } }}g^{\mathbb{Q}}\\
    &\quad + \cancel{r\gamma g^{\mathbb{Q}}} + \rBrackets{\frac{\overline{\lambda }}{1-\gamma } + \frac{\sigma \rho b(t)}{1-\gamma}} \overline{\lambda }v \rBrackets{(\gamma - 1) g^{\mathbb{Q}} + y g^{\mathbb{Q}}_{y} }\\
    &\quad + \kappa \left( \cancel{\theta} - \cancel{v}\right) b(t) g^{\mathbb{Q}} \\
    &\quad + \frac{1}{2} v  \rBrackets{\frac{\overline{\lambda }}{1-\gamma } + \frac{\sigma \rho b(t)}{1-\gamma}}^2  \rBrackets{(\gamma - 1)(\gamma - 2)g^{\mathbb{Q}} + 2(\gamma - 1) y g^{\mathbb{Q}}_{y}}\\
    &\quad + \frac{1}{2} v \sigma^2  \rBrackets{ \cancel{ \rBrackets{ b(t)}^2  g^{\mathbb{Q}}} + 2 b(t) g^{\mathbb{Q}}_v }\\
    &\quad + \rho \sigma  v \rBrackets{\frac{\overline{\lambda }}{1-\gamma } + \frac{\sigma \rho b(t)}{1-\gamma}} \rBrackets{ (\gamma - 1)  b(t) g^{\mathbb{Q}} + (\gamma - 1)  g^{\mathbb{Q}}_{v} + y b(t) g^{\mathbb{Q}}_{y} }\\
    &\quad + \sigma \lambar \rho v g^{\mathbb{Q}}_v + \sigma \lambda^{v}\sqrt{1 - \rho^2} v g^{\mathbb{Q}}_v.
\end{align*}
Next, we use that $\frac{-\gamma}{1- \gamma}  = 1 - \frac{1}{1 - \gamma},\,\frac{\gamma - 1}{1 - \gamma} = -1, \frac{(\gamma - 1)(\gamma - 2)}{(1 - \gamma)(1 - \gamma)} = 1 + \frac{1}{1 - \gamma}$, expand several brackets with multiple summation terms, and move $y,v$ to the beginning of the corresponding product where they appear:

\begin{align*}
    0 &=   v\cdot \rBrackets{\frac{1}{2} \rBrackets{1 - \frac{1}{1 - \gamma}}\sigma^{2} \rho^{2} (b(t))^2 + \rBrackets{1 - \frac{1}{1 - \gamma}}\overline{\lambda }\sigma \rho b(t) + \frac{1}{2} \rBrackets{1 - \frac{1}{1 - \gamma}} \overline{\lambda }^{2} }g^{\mathbb{Q}}\\
    &\quad - v \rBrackets{\overline{\lambda } + \sigma \rho b(t)} \overline{\lambda } g^{\mathbb{Q}} + v y \overline{\lambda } (1 - \gamma)^{-1}\rBrackets{\overline{\lambda } + \sigma \rho b(t)}  g^{\mathbb{Q}}_{y} \\
    &\quad + v \frac{1}{2}  \rBrackets{1 + \frac{1}{1 - \gamma} } \rBrackets{\overline{\lambda } + \sigma \rho b(t)}^{2} g^{\mathbb{Q}} + v y (\gamma - 1)^{-1} \rBrackets{\overline{\lambda } + \sigma \rho b(t)}^2 g^{\mathbb{Q}}_{y}\\
    &\quad + v \sigma^2   b(t) g^{\mathbb{Q}}_v - v  \sigma  \rho \rBrackets{\overline{\lambda } + \sigma \rho b(t)}  b(t) g^{\mathbb{Q}}\\
    & \quad - v \sigma  \rho \rBrackets{\overline{\lambda } + \sigma \rho b(t)}  g^{\mathbb{Q}}_{v} + v y  \sigma  \rho (1 - \gamma)^{-1} \rBrackets{\overline{\lambda } + \sigma \rho b(t)}  b(t) g^{\mathbb{Q}}_{y} + v \lambar \sigma \rho  g^{\mathbb{Q}}_v + \sigma \lambda^{v}\sqrt{1 - \rho^2} v g^{\mathbb{Q}}_v.
\end{align*}

The above equality is true for any $y>0, v>0$ if the the terms next to $v g^{\mathbb{Q}}$, $v g^{\mathbb{Q}}_{v}$, $v y g^{\mathbb{Q}}_{y}$ are $0$.

\textbf{Coefficient next to $v g^{\mathbb{Q}}$}

Collecting all terms next to $v g^{\mathbb{Q}}$ yields:
\begin{align*}
     0 & = \frac{1}{2} \rBrackets{1 - \frac{1}{1 - \gamma}}\sigma^{2} \rho^{2} (b(t))^2 + \rBrackets{1 - \frac{1}{1 - \gamma}}\overline{\lambda }\sigma \rho b(t) + \frac{1}{2} \rBrackets{1 - \frac{1}{1 - \gamma}} \overline{\lambda }^{2} - \rBrackets{\overline{\lambda } + \sigma \rho b(t)} \overline{\lambda }\\
     & \quad + \frac{1}{2}  \rBrackets{1 + \frac{1}{1 - \gamma} } \rBrackets{\overline{\lambda } + \sigma \rho b(t)}^{2} - \sigma  \rho \rBrackets{\overline{\lambda } + \sigma \rho b(t)}  b(t)\\
     & = \frac{1}{2} \rBrackets{1 - \frac{1}{1 - \gamma}}\sigma^{2} \rho^{2} (b(t))^2 + \rBrackets{1 - \frac{1}{1 - \gamma}}\overline{\lambda }\sigma \rho b(t) + \frac{1}{2} \rBrackets{1 - \frac{1}{1 - \gamma}} \overline{\lambda }^{2} - \overline{\lambda }^2 - \overline{\lambda } \sigma \rho b(t) \\
     & \quad + \frac{1}{2}  \rBrackets{1 + \frac{1}{1 - \gamma} } \rBrackets{\overline{\lambda }^{2} + 2 \lambar \sigma \rho b(t) + \rBrackets{\sigma \rho b(t)}^{2}} - \overline{\lambda } \sigma  \rho b(t) -  \sigma^2 \rho^2 (b(t))^2.\\
\end{align*}

We show that the above equality is true by showing that the coefficients next to $(b(t))^2$, $b(t)^1$ and $b(t)^0$ are all equal to 0.

For the coefficient next to $(b(t))^2$ we obtain:
\begin{align*}
     \frac{1}{2} \rBrackets{1 - \frac{1}{1 - \gamma}}\sigma^{2} \rho^{2} + \frac{1}{2}  \rBrackets{1 + \frac{1}{1 - \gamma} } \sigma^2 \rho^2 - \sigma ^2 \rho^2 = \sigma^{2} \rho^{2}\rBrackets{\frac{1}{2} - \frac{1}{2 (1 - \gamma)} + \frac{1}{2} + \frac{1}{2 (1 - \gamma)} - 1} = 0.
\end{align*}

For the coefficient next to $(b(t))^1$ we obtain:
\begin{align*}
    \rBrackets{1 - \frac{1}{1 - \gamma}}& \overline{\lambda }\sigma \rho - \overline{\lambda }\sigma \rho +  \frac{1}{2}  \rBrackets{1 + \frac{1}{1 - \gamma} } 2 \lambar \sigma  \rho - \overline{\lambda } \sigma  \rho = \overline{\lambda } \sigma  \rho \rBrackets{1 - \frac{1}{1 - \gamma}  - 1 + \rBrackets{1 + \frac{1}{1 - \gamma} } - 1} = 0.
\end{align*}

For the coefficient next to $(b(t))^0$ we obtain:
\begin{align*}
    \frac{1}{2} \rBrackets{1 - \frac{1}{1 - \gamma}} \overline{\lambda }^{2} - \overline{\lambda }^2 + \frac{1}{2}  \rBrackets{1 + \frac{1}{1 - \gamma} } \overline{\lambda }^{2} = \overline{\lambda }^{2} \rBrackets{\frac{1}{2} - \frac{1}{2}\frac{1}{1 - \gamma} - 1 + \frac{1}{2} + \frac{1}{2} \frac{1}{1 - \gamma}} = 0.
\end{align*}

Hence, the coefficient next to $v g^{\mathbb{Q}}$ is $0$, i.e. $v g^{\mathbb{Q}}$ vanishes in the relation we are proving.

\textbf{Coefficient next to $v y g^{\mathbb{Q}}_{y}$}
The coefficient next to $v y g^{\mathbb{Q}}_{y}$ is equal to:
\begin{align*}
    \overline{\lambda } & (1 - \gamma)^{-1}\rBrackets{\overline{\lambda } + \sigma \rho b(t)} + (\gamma - 1)^{-1} \rBrackets{\overline{\lambda } + \sigma \rho b(t)}^2 + \sigma  \rho (1 - \gamma)^{-1} \rBrackets{\overline{\lambda } + \sigma \rho b(t)}  b(t)\\
    & = (1 - \gamma)^{-1} \rBrackets{\overline{\lambda }\rBrackets{\overline{\lambda } + \sigma \rho b(t)} - \rBrackets{\overline{\lambda } + \sigma \rho b(t)}^2 + \sigma  \rho  b(t) \rBrackets{\overline{\lambda } + \sigma \rho b(t)} }\\
    & = (1 - \gamma)^{-1} \rBrackets{\rBrackets{\overline{\lambda } + \sigma \rho b(t)}^2 - \rBrackets{\overline{\lambda } + \sigma \rho b(t)}^2} = 0.
\end{align*}
Hence, the coefficient next to $v y g^{\mathbb{Q}}_{y}$ is $0$, i.e., $v y g^{\mathbb{Q}}_{y}$ vanishes in the relation we are proving.

\textbf{Coefficient next to $v g^{\mathbb{Q}}_{v}$}
The coefficient next to $v g^{\mathbb{Q}}_{v}$ is equal to:
\begin{align*}
    \sigma^2  b(t) - \sigma  \rho \rBrackets{\overline{\lambda } + \sigma \rho b(t)}  + \lambar \sigma \rho + \lambda^v \sigma \sqrt{1- \rho^2} &= \sigma^2  b(t) - \sigma  \rho \overline{\lambda } - \sigma^2 \rho^2 b(t)  + \lambar \sigma \rho + \lambda^v \sigma \sqrt{1- \rho^2}\\
    &= b(t) \sigma^2 \rBrackets{1 - \rho^2} + \lambda^v \sigma \sqrt{1- \rho^2}
\end{align*}

The coefficient next to $v g^{\mathbb{Q}}_{v}$ is equal to zero if $\lambda^v = - \sigma \sqrt{1- \rho^2} b(t)$. This is equivalent to picking a convenient change of measure on the variance process.

So for $\lambda^v = - \sigma \sqrt{1- \rho^2} b(t)$  \eqref{eq:suffiecient_condition_as_financial_derivatives_SupMat} holds also for the 2nd and 3rd piece of the modified utility function:
\begin{eqnarray}
    \frac{\partial}{\partial y}\EVtyv{P}{\frac{1}{\gamma }\left( K^{\gamma }-\rBrackets{Y^{\ast}(T)}^{\gamma }\right) 1_{\left\{ Y^{\ast}(T) < K\right\} } } & = & \frac{\partial}{\partial y}\EVtyv{Q}{\exp(-r(T - t))\left( K-Y^{\ast}(T)\right) 1_{\left\{ Y^{\ast}(T) < K\right\} }}\notag \\
    &&\cdot y^{\gamma -1}\exp \left(a(t) + b(t)v\right)\quad \forall K > 0,\, \gamma < 1 \notag
\end{eqnarray}

\textbf{Part 1. Term 4. i.e. binary option\\}
Now we derive the relationship between $\lambda_{\varepsilon}$, $y$, $k_{\varepsilon}$ and $k_{v}$, which ensures that the last piece of the modified utility function also satisfies the same \eqref{cond:SC_in_lemma}, in particular $\overline{U}^{(4)}_{y} = y^{\gamma -1}\exp \left(a(t) + b(t)v\right) D^{(4)}_{y}$.

For $\overline{U}_{4}(y)=\frac{1}{\gamma }\left( K^{\gamma }-k_{\varepsilon
}^{\gamma }+\gamma \lambda _{\varepsilon }\right) 1_{\left\{ y<k_{\varepsilon}\right\}
}$ in \eqref{eq:U_i_Fourier_Transform_SupMat} we get:%
\begin{flalign*}
\overline{U}^{(4)} &=\frac{1}{2\pi }\frac{1}{\gamma }\int \int \left( K^{\gamma
}-k_{v}^{\gamma }+\gamma \lambda _{\varepsilon }\right)
1_{\left\{ z<\ln k_{\varepsilon}-\ln y\right\} }\exp \left(
-iuz+A^{\mathbb{P}}(T-t,u)+B^{\mathbb{P}}(T-t,u)v\right) dudz\\
&= \frac{1}{2\pi }\frac{ K^{\gamma
}-k_{v}^{\gamma }+\gamma \lambda _{\varepsilon }}{\gamma }\int \limits_{-\infty}^{\ln(k_{\varepsilon}/y)} \underbrace{\int\limits_{-\infty}^{+\infty} \exp \left(
-iuz+A^{\mathbb{P}}(T-t,u)+B^{\mathbb{P}}(T-t,u)v\right) du}_{=:g(y,z)}dz\\
     &\stackrel{LIR}{=} \frac{1}{2\pi }\frac{ K^{\gamma
}-k_{v}^{\gamma }+\gamma \lambda _{\varepsilon }}{\gamma } \Biggl(g(y, \ln(k_{\varepsilon}/y)) \rBrackets{-\frac{1}{y}} - \lim_{c \downarrow -\infty} \biggl(  g(y, c) \underbrace{\frac{\partial c}{\partial y} }_{=0}\biggr) + \int \limits_{-\infty}^{\ln\rBrackets{k_{\varepsilon}/y}} \underbrace{\ddy g(y, z) }_{=0}dz\Biggr)\\
    &\stackrel{g}{=} -\frac{1}{y} \frac{K^{\gamma
}-k_{v}^{\gamma }+\gamma \lambda _{\varepsilon }}{\gamma } \frac{1}{2\pi } \int \limits_{-\infty}^{+\infty} \exp
\left( -iu\left( \ln k_{\varepsilon}-\ln y\right) +A^{\mathbb{P}}(T-t,u)+B^{\mathbb{P}}(T-t,u)v\right) du\\
&=-\frac{1}{y} \frac{K^{\gamma}-k_{v}^{\gamma }+\gamma \lambda _{\varepsilon }}{\gamma } \frac{1}{2\pi }\int\limits_{-\infty}^{+\infty} \exp \left( -iu \ln k_{\varepsilon} \right) \exp\left( iu \ln y +A^{\mathbb{P}}(T-t,u)+B^{\mathbb{P}}(T-t,u)v\right) du \\
\end{flalign*}
So:
\begin{flalign*}
   \overline{U}_{y}^{(4)}&=-\frac{1}{y} \frac{K^{\gamma}-k_{v}^{\gamma }+\gamma \lambda _{\varepsilon }}{\gamma } \underbrace{\frac{1}{2\pi }\int\limits_{-\infty}^{+\infty} \exp \left( -iu \ln k_{\varepsilon} \right) \phi^{Z^\ast(T),\mathbb{P}}(u; t, \ln(y), v) du}_{=: f^{\mathbb{P}}_{Z^\ast(T)}(\ln k_{\varepsilon})} = -\frac{1}{y} \frac{K^{\gamma}-k_{v}^{\gamma }+\gamma \lambda _{\varepsilon }}{\gamma } f^{\mathbb{P}}_{Z^\ast(T)}(\ln k_{\varepsilon}),
\end{flalign*}
where $f^{\mathbb{P}}_{Z^\ast(T)}$ denotes the $\mathbb{P}$-density of $Z^\ast(T) = \ln(Y^{\ast}(T))$.

Applying the previous result for $\gamma = 1$, $\lambda_{\varepsilon} = 0$ and working under the measure $\mathbb{Q}$ instead of $\mathbb{P}$, we get for $D_{4}(y)=\left( K-k_{v}\right) 1_{\left\{ y<k_{\varepsilon}\right\} }$ in \eqref{eq:Pi_i_Fourier_Transform_SupMat} the following:%
\begin{eqnarray*}
D^{(4)} &=&\frac{1}{2\pi }\exp\rBrackets{-r(T-t)}\int \int \left( K-k_{\varepsilon
}\right) 1_{\left\{ z<\ln k_{\varepsilon}-\ln y\right\} }\exp \left(
-iuz+A^{\mathbb{Q}}(T-t,u)+B^{\mathbb{Q}}(T-t,u)v\right) dudz\\
&=& -\frac{1}{y} \left(K-k_{v}\right) \exp\rBrackets{-r(T-t)} f^{\mathbb{Q}}_{Z^\ast(T)}(\ln k_{\varepsilon}),
\end{eqnarray*}
where $f^{\mathbb{Q}}_{Z^\ast(T)}$ denotes the $\mathbb{Q}$-density of $Z^\ast(T) = \ln(Y^{\ast}(T))$.

Hence, the condition equivalent to \eqref{cond:SC_in_lemma} in the context of the fourth piece is given by:
\begin{align}
    \cancel{\rBrackets{-\frac{1}{y}}} \frac{K^{\gamma}-k_{v}^{\gamma }+\gamma \lambda _{\varepsilon }}{\gamma } f^{\mathbb{P}}_{Z^\ast(T)}(\ln k_{\varepsilon}) &\stackrel{!}{=} y^{\gamma -1}\exp \left(a(t) + b(t)v\right) \cancel{\rBrackets{-\frac{1}{y}}} \left(K-k_{v}\right)\notag \\
    & \quad \cdot \exp\rBrackets{-r(T-t)} f^{\mathbb{Q}}_{Z^\ast(T)}(\ln k_{\varepsilon})\notag\\
    \iff \frac{K^{\gamma}-k_{v}^{\gamma }+\gamma \lambda _{\varepsilon }}{\gamma } f^{\mathbb{P}}_{Z^\ast(T)}(\ln k_{\varepsilon}) &\stackrel{!}{=} y^{\gamma -1}\exp \left(a(t) + b(t)v\right) \left(K-k_{v}\right) \tag{ESC Binary} \label{eq:sufficient_conditions_binary_SupMat} \\
    & \quad \cdot \exp\rBrackets{-r(T-t)} f^{\mathbb{Q}}_{Z^\ast(T)}(\ln k_{\varepsilon}) \notag
\end{align}
Condition \eqref{eq:sufficient_conditions_binary_SupMat} is satisfied if the following relationship among $\lambda_{\varepsilon}$, $y$, $k_{\varepsilon}$ and $k_{v}$ holds:
\begin{equation}
    \lambda_{\varepsilon } = y^{\gamma -1}\exp \left(a(t) + b(t)v\right) \left(K-k_{v}\right) \exp\rBrackets{-r(T-t)} \frac{f^{\mathbb{Q}}_{Z^{\ast}(T)}(\ln k_{\varepsilon})}{f^{\mathbb{P}}_{Z^{\ast}(T)}(\ln k_{\varepsilon})} -  \frac{K^{\gamma}-k_{v}^{\gamma }}{\gamma }.\label{eq:lambda_epsilon}
\end{equation}

So by Lemma \ref{lem:sufficient_condition}, both \eqref{cond:U_D_yy_y} and \eqref{cond:U_D_yv_y} in Theorem \ref{MainTheo} are satisfied at an arbitrary but fixed $\tin$, when Condition \eqref{cond:SC_in_lemma} holds at $\tin$. In Part 1 of this proof, we have shown that for an arbitrary but fixed $\tin$, ensuring Condition \eqref{cond:SC_in_lemma} is equivalent to ensuring $\overline{U}^{(i)}_{y} = y^{\gamma -1}\exp \left(a(t) + b(t)v\right) D^{(i)}_{y}$ $\forall i  \in \{1,2,3,4\}$ for the constructed $D$. As we have also shown, these four equalities are satisfied when $\lambda^{v}(t) =  - \sigma \sqrt{1- \rho^2} b(t)$ and \eqref{eq:lambda_epsilon} hold, imposing a specific relationship among $\lambda_{\varepsilon }$, $k_{\varepsilon}$, $k_{v}$ and $y$ at $\tin$. The optimal Lagrange multiplier is determined at $t = 0$ via
\begin{equation*}
    \lambda _{\varepsilon}^\ast = y_0^{\gamma - 1}\exp \left(a(0) + b(0)v_0 - rT\right) \left(K-k_{v,0}\right) \frac{f^{\mathbb{Q}}_{Z^\ast(T)}(\ln k_{\varepsilon, 0})}{f^{\mathbb{P}}_{Z^\ast(T)}(\ln k_{\varepsilon, 0})} -  \frac{K^{\gamma} - k_{v,0}^{\gamma }}{\gamma }
\end{equation*}
and imposes the relationship among the degrees of freedom $k_{\varepsilon,t}$, $k_{v,t}$ and $y_t$ at each $\tin$. 

\bigskip

\textbf{Part 2.}  At any $\tin$, Condition \eqref{cond:D_v} is satisfied due to the assumption that $(y_t, k_{v,t}, k_{\varepsilon,t})$ solves the vega-neutrality equation in \eqref{eq:SNLE_rho_zero}, namely
\begin{equation*}
h_{VN}(y_t, k_{v, t}, k_{\varepsilon, t}) := \widehat{D}(t,y_t,v_t;k_{v, t}, k_{\varepsilon, t}) = 0.
\end{equation*}
Note that for any $t \in (0, T]$ the system \eqref{eq:SNLE_rho_zero} has three variables and three equations. The same applies to the system \eqref{eq:SNLE_t_zero} at $t = 0$.

\bigskip

\textbf{Part 3.} As argued in Parts 1 and 2, Conditions \eqref{cond:U_D_yy_y} -- \eqref{cond:D_v} are satisfied, at $t=0$, the second equation in \eqref{eq:SNLE_t_zero} ensures that the VaR constraint is satisfied:
\begin{equation*}
    h_{VaR}(y_0, k_{v, 0}, k_{\varepsilon, 0}) := \mathbb{P}\left( Y^{\ast}(T) <  k_{\varepsilon, 0}|Y^{\ast}(0) = y_0, v(0) = v_0\right) = \varepsilon.
\end{equation*}
Thus, we can apply Theorem \ref{MainTheo} for $\lambda^{v} = - \sigma \sqrt{1- \rho^2} b(t)$ and conclude that
\begin{eqnarray*}
    X^{x, \pi^\ast_c}(T) & = & D(Y^{y, \piu}(T))\quad  \text{with} \quad x = \dqlatyv  :=  \mathbb{E}_{t,y,v}^{\mathbb{Q}(\lambda^{v})}\left[\exp\rBrackets{-r(T-t)} D(Y^{y, \piu}(T))\right];\\
    \mathcal{V}^{c}\left( t, x ,v\right) &=& \udp(t,y,v);\\
    \pi ^{\ast }_{c}(t) &=& \piu(t)  \cdot y \cdot \frac{\dq_{y}(t,y,v)}{\dq(t,y,v)},
\end{eqnarray*}
where $D$ is the derivative constructed via a continuum of contingent claims with payoffs $\widehat{D}(\cdot; k_{v,t}, k_{\varepsilon,t})$, as allowed by Proposition \ref{prop:equivalence_btw_sequence_and_single_D}.
\end{proof}

\section{Explicit formulas for the left-hand side of \ref{eq:SNLE_rho_zero}}\label{app:explicit_formulas}

In this section of the appendix, we provide representations of the equations in \eqref{eq:SNLE_rho_zero} in the spirit of \cite{Carr1999}. \\

\textbf{Budget equation.} First, we provide a formula for the price of a plain-vanilla put option. Second, we derive the formula for the price of a digital put option. Afterwards, we will provide the formula for the LHS of the budget equation, which combines the formulas obtained in the previous two steps.

\textit{Put option.} Take any $\alpha_P > 1$ and any strike $K > 0$. Denote $k = \ln\rBrackets{K}$. Analogously to Equation (3.50) in Fabrice (2013), pages 82-83, we can get:
\begin{align}
   Put(k) & := Put(Y^{\ast}(T), K) = \EVtyv{Q}{\exp\rBrackets{-r(T - t)}\rBrackets{K - Y^{\ast}(T)}^{+}} \notag \\
    & = \frac{\exp\rBrackets{\alpha k}}{\pi} \int \limits_{0}^{+\infty} \text{Real}\rBrackets{\frac{\exp\rBrackets{-r (T - t)} \exp\rBrackets{-i u k}}{\alpha_P^2 - \alpha_P - u^2 + i u (1 - 2\alpha_P)}\phi^{Z^\ast(T), \mathbb{Q}}(u + (\alpha_P - 1) i;t,\ln y,v)}du. \label{eq:put_price_CM}
\end{align}

\textit{Digital put option.} Let $K > 0$ be an arbitrary but fixed strike of a digital put option with the nominal payment of 1 monetary unit. Denote $k = \ln\rBrackets{K}$. Then the price of such a digital put option is given by:
\begin{align}
     DigPut(k) & := DigPut(Y^{\ast}(T), K) = \EVtyv{Q}{\exp\rBrackets{-r(T - t)}\mathbbm{1}_{\{Y^{\ast}(T) < K\}}} \notag \\
     & \stackrel{Def}{=}  \mathbb{E}^{\mathbb{Q}}\sBrackets{\exp\rBrackets{-r(T - t)}\mathbbm{1}_{\{Z^\ast(T) < k\}}|Z^\ast(t) = \ln\rBrackets{y}, v(t) = v} \notag \\
    & = \exp\rBrackets{-r(T - t)} \mathbb{Q}\rBrackets{Z^\ast(T) < k |Z^\ast(t) = \ln\rBrackets{y}, v(t) = v} = \exp\rBrackets{-r(T - t)} \int \limits_{-\infty}^{k} f_{Z^\ast(T)}^{\mathbb{Q}}(z)\,dz \label{eq:digital_put_price}
\end{align}

Take any $\alpha_{DP} > 0$ and consider the following dampened price of a digital put option:
\begin{equation}\label{eq:digital_put_dampened_price}
    DigPut^{\rBrackets{\alpha_{DP}}}(k) = \exp\rBrackets{-\alpha_{DP} k} DigPut(k)
\end{equation}

Then the Fourier transform of $DigPut^{\rBrackets{\alpha_{DP}}}(k)$ is given by:
\begin{align}
    \phi^{DigPut^{\rBrackets{\alpha_{DP}}}}(k) & = \int\limits_{-\infty}^{+\infty} \exp\rBrackets{i u k } DigPut^{\rBrackets{\alpha_{DP}}}(k) \, dk  \stackrel{\eqref{eq:digital_put_dampened_price}}{=}   \int\limits_{-\infty}^{+\infty} \exp\rBrackets{i u k } \exp\rBrackets{-\alpha_{DP} k} DigPut(k) \, dk \notag \\
    & \stackrel{\eqref{eq:digital_put_price}}{=} \int\limits_{-\infty}^{+\infty} \exp\rBrackets{i u k } \exp\rBrackets{-\alpha_{DP} k} \exp\rBrackets{-r(T - t)} \int \limits_{-\infty}^{k} f_{Z^\ast(T)}^{\mathbb{Q}}(z) \,dz \, dk \notag \\
    & \stackrel{(i)}{=} \int\limits_{-\infty}^{+\infty} \int \limits_{z}^{+\infty}  \exp\rBrackets{i u k } \exp\rBrackets{-\alpha_{DP} k} \exp\rBrackets{-r(T - t)} f_{Z^\ast(T)}^{\mathbb{Q}}(z) \, dk \,dz  \notag \\
    & \stackrel{}{=} \int\limits_{-\infty}^{+\infty}  \exp\rBrackets{-r(T - t)} f_{Z^\ast(T)}^{\mathbb{Q}}(z) \rBrackets{ \int \limits_{z}^{+\infty}  \exp\rBrackets{i u k } \exp\rBrackets{-\alpha_{DP} k} \, dk }\,dz  \notag \\
    & \stackrel{\alpha_{DP} > 0}{=} \int\limits_{-\infty}^{+\infty}  \exp\rBrackets{-r(T - t)} f_{Z^\ast(T)}^{\mathbb{Q}}(z) \frac{\exp\rBrackets{iuz - \alpha_{DP} z}}{\alpha_{DP} - i u}\,dz \notag\\
    & = \frac{\exp\rBrackets{-r(T - t)}}{\alpha_{DP} - i u} \int\limits_{-\infty}^{+\infty} \exp\rBrackets{iz(u - \alpha_{DP} / i)} f_{Z^\ast(T)}^{\mathbb{Q}}(z) \, dz  \notag \\
    & = \frac{\exp\rBrackets{-r(T - t)}}{\alpha_{DP} - i u} \phi^{Z^\ast(T), \mathbb{Q}}(u - \alpha_{DP} / i;t,\ln y,v) \notag  \\
    & = \frac{\exp\rBrackets{-r(T - t)}}{\alpha_{DP} - i u} \phi^{Z^\ast(T), \mathbb{Q}}(u + \alpha_{DP} i;t,\ln y,v) \label{eq:digital_put_dampened_FT}
\end{align}
where in (i) we changed the order of integration.

Therefore, the price of a digital put option is given by:
\begin{align}
    DigPut(k) & = \exp\rBrackets{\alpha_{DP} k} \exp\rBrackets{-\alpha_{DP} k} DigPut(k) \stackrel{\eqref{eq:digital_put_dampened_price}}{=} \exp\rBrackets{\alpha_{DP} k} DigPut^{\rBrackets{\alpha_{DP}}}(k) \notag\\
    &  \stackrel{IFT}{=} \exp\rBrackets{\alpha_{DP} k} \frac{1}{2 \pi} \int\limits_{-\infty}^{+\infty}\text{Real} \rBrackets{\exp\rBrackets{-i u k } \phi^{DigPut^{\rBrackets{\alpha_{DP}}}}(u)} \, du \notag\\
    & \stackrel{\eqref{eq:digital_put_dampened_FT}}{=}  \exp\rBrackets{\alpha_{DP} k} \frac{1}{2 \pi} \int\limits_{-\infty}^{+\infty} \text{Real} \rBrackets{\exp\rBrackets{-i u k } \frac{\exp\rBrackets{-r(T - t)}}{\alpha_{DP} - i u} \phi^{Z^\ast(T), \mathbb{Q}}(u + \alpha_{DP} i;t,\ln y,v)} \, du \notag \\
    & = \exp\rBrackets{\alpha_{DP} k} \frac{1}{\pi} \int\limits_{0}^{+\infty} \text{Real} \rBrackets{\exp\rBrackets{-i u k } \frac{\exp\rBrackets{-r(T - t)}}{\alpha_{DP} - i u} \phi^{Z^\ast(T), \mathbb{Q}}(u + \alpha_{DP} i;t,\ln y,v)} \, du.\label{eq:digital_put_price_CM}
\end{align}

Therefore, the budget equation in \eqref{eq:SNLE_rho_zero} can be written as follows:
\begin{align*}
    & \dqla (t,y,v) = y + \int \limits_{0}^{+\infty} \text{Real}\rBrackets{\frac{\exp\rBrackets{-r (T - t)} \exp\rBrackets{-i u \ln(K)}}{\alpha_P^2 - \alpha_P - u^2 + i u (1 - 2 \alpha_P)}\phi^{Z^\ast(T), \mathbb{Q}}(u + (\alpha_P - 1) i;t,\ln y,v)}du \\
    &\quad \cdot \frac{\exp\rBrackets{\alpha_P \ln(K)}}{\pi} - \int \limits_{0}^{+\infty} \text{Real}\rBrackets{\frac{\exp\rBrackets{-r (T - t)} \exp\rBrackets{-i u \ln(k_{v})}}{\alpha_P^2 - \alpha_P - u^2 + i u (1 - 2 \alpha_P)}\phi^{Z^\ast(T), \mathbb{Q}}(u + (\alpha_P - 1) i;t,\ln y,v)}du \\
    & \quad \cdot \frac{\exp\rBrackets{\alpha_P \ln(k_{v})}}{\pi} -  \int\limits_{0}^{+\infty} \text{Real} \rBrackets{ \exp\rBrackets{-i u \ln\rBrackets{k_{\varepsilon}} } \frac{\exp\rBrackets{-r(T - t)}}{\alpha_{DP} - i u} \phi^{Z^\ast(T), \mathbb{Q}}(u + \alpha_{DP} i;t,\ln y,v)} \, du\\
    & \quad \cdot (K - k_{v}) \exp\rBrackets{\alpha_{DP} \ln\rBrackets{k_{\varepsilon}}} \frac{1}{ \pi}
\end{align*}

\textbf{VaR equation.}
The LHS of the VaR equation can be obtained from Equation \eqref{eq:digital_put_price_CM} by considering the measure $\mathbb{P}$ instead of $\mathbb{Q}$ and setting $r = 0$.
\begin{flalign*}
    \mathbb{P}\left( Y^{\ast}(T) <  k_{\varepsilon}|Y^{\ast}(t) = y, v(t) = v\right) =& \exp\rBrackets{\alpha_{DP} \ln\rBrackets{k_{\varepsilon}}} \\
    &\cdot \frac{1}{\pi} \int\limits_{0}^{+\infty} \text{Real} \rBrackets {\frac{\exp\rBrackets{-i u \ln\rBrackets{k_{\varepsilon}} }}{\alpha_{DP} - i u}  \phi^{Z^\ast(T),\mathbb{P}}(u + \alpha_{DP} i;t,\ln y,v) } \, du.
\end{flalign*}

\textbf{Vega equation.} Differentiating the budget equation w.r.t $v$ and using Remark 3 to Corollary \ref{cor:heston_var_rho_nonzero_solution}, we get:
\begin{flalign*}
    \dqla_{v}(t,y,v) & = \frac{\exp\rBrackets{\alpha_P \ln(K)}}{\pi} \int \limits_{0}^{+\infty} \text{Real}\Biggl(\frac{\exp\rBrackets{-r (T - t)} \exp\rBrackets{-i u \ln(K)} B^{\mathbb{Q}}(T -t, u + (\alpha_P - 1) i)}{\alpha_P^2 - \alpha_P - u^2 + i u (1 - 2 \alpha_P)}\\
    & \quad \cdot \phi^{Z^\ast(T), \mathbb{Q}}(u + (\alpha_P - 1) i;t,\ln y,v)\Biggr) \, du - \frac{\exp\rBrackets{\alpha_P \ln(k_{v})}}{\pi}\\
    &\quad  \cdot  \int \limits_{0}^{+\infty} \text{Real}\Biggl(\frac{\exp\rBrackets{-r (T - t)} \exp\rBrackets{-i u \ln(k_{v})} B^{\mathbb{Q}}(T -t, u + (\alpha_P - 1) i)}{\alpha_P^2 - \alpha_P - u^2 + i u (1 - 2 \alpha_P)} \\
    & \quad \cdot \phi^{Z^\ast(T), \mathbb{Q}}(u + (\alpha_P - 1) i;t,\ln y,v)\Biggr) \, du - (K - k_{v}) \exp\rBrackets{\alpha_{DP} \ln\rBrackets{k_{\varepsilon}}} \frac{1}{2 \pi}  \\
    & \quad  \cdot \int\limits_{-0}^{+\infty} \text{Real} \Biggl( \exp\rBrackets{-i u \ln\rBrackets{k_{\varepsilon}}} \frac{\exp\rBrackets{-r(T - t)} B^{\mathbb{Q}}(T -t, u + \alpha_{DP} i)}{\alpha_{DP} - i u} \\
    & \quad \cdot \phi^{Z^\ast(T), \mathbb{Q}}(u + \alpha_{DP} i;t,\ln y,v)\Biggr) \, du.
\end{flalign*}


\section{Numerical studies for more turbulent markets}\label{app:numerical_studies_turbulent_markets}
In this subsection of the Appendix, we consider $T=3$ as in the main part of the article, but a decision maker with a smaller relative risk-aversion and who invests in a more turbulent market than we had before, i.e., higher initial value of the variance process, a higher long-term average variance, and a lower mean reversion rate. In particular, we set $\gamma = -1$ and use the values of the Heston model parameters so that they are consistent with \cite{Schoutens2004}: $v_0 = 0.0654$, $\tilde{\theta} = 0.0707$, $\tilde{\kappa} = 0.6067$, $\sigma = 0.2928$, $\rho = -0.7571$. We plot in Figure \ref{fig:impact_of_rho_sigma_VaR_Schoutens} the sensitivity of the optimal constrained investment strategy w.r.t. $\rho$, $\sigma$, and $\kappa$. 

\begin{figure}[!ht]
        \centering
        \begin{subfigure}{0.5\textwidth}
          \centering
          \includegraphics[width=\linewidth]{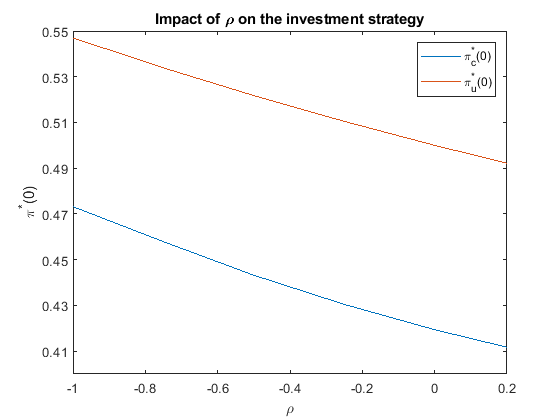}
          \caption{$\pic(0)$ \& $\piu(0)$ vs $\rho$}
          \label{sfig:piStar_vs_rho_VaR_Schoutens}
        \end{subfigure}%
        \begin{subfigure}{0.5\textwidth}
          \centering
          \includegraphics[width=\linewidth]{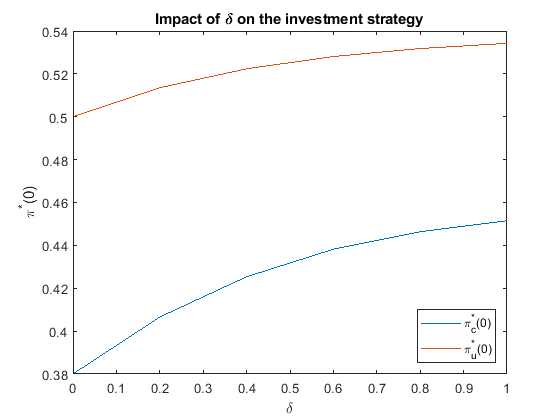}
          \caption{$\pic(0)$ \& $\piu(0)$  vs $\delta$ (influencing $\sigma_{\delta} = \sigma \cdot \delta$ and  $\kappa_{\delta} = \kappa \cdot \delta$)}
          \label{sfig:piStar_vs_sigma_VaR_Schoutens}
        \end{subfigure}
        \caption{The impact of $\rho$, $\sigma$, $\kappa$ on the optimal investment strategies in a more turbulent market}
        \label{fig:impact_of_rho_sigma_VaR_Schoutens}
\end{figure}

In contrast to the case of average parameters considered in the main text of the paper, the sensitivity of the optimal constrained investment strategies w.r.t. the correlation coefficient, mean-reversion rate, and the volatility of the variance process is higher in a more volatile market.  For example, according to the Subfigure \ref{sfig:piStar_vs_rho_VaR_Schoutens}, a decrease in the correlation coefficient from $-40\%$ to $-60\%$ leads to an increase of the initial optimal constrained investment strategy by more than $1\%$, namely from $42.7\%$ to approximately $44\%$. Looking at $\delta = 1$ and $\delta = 0.75$ in Subfigure \ref{sfig:piStar_vs_sigma_VaR_Schoutens}, we see that a decrease in volatility from $39.28\%$  to $21.96\%$ and the real-world-measure mean-reversion rate of the variance process from $0.8171$ to $0.6128$ would require a rational investor to decrease his/her initial constrained investment strategy by approximately $0.7\%$, namely, from $45.2\%$ to $44.7\%$.  The behavior is similar to that of the optimal unconstrained investment strategy.  It can have the following economic interpretation. The infinitesimal Sharpe ratio of the risky asset is $\bar{\lambda} \sqrt{v(t)}$. It is negatively correlated with the Wiener process $W_1^{\mathbb{P}}(t)$ driving the stock returns. As a result, low return \enquote{today} tends to occur when $dW_1^{\mathbb{P}}(t)$ is negative and $dW_2^{\mathbb{P}}(t)$ is positive, which in turn pushes the \enquote{tomorrow's} Sharpe ratio higher and may give hope to the investor for good investment in the risky asset. Consequently, an investor increases his/her position in the risky asset in comparison to the Black-Scholes market. The \enquote{more} incompleteness an investor sees in the market, the more chances he/she sees for making profit with the risky asset investment and the corresponding correction term will be larger.

\end{document}

%% file: HestonVaR_shared.tex
\usepackage{lipsum}
\usepackage{amsfonts}
\usepackage{graphicx}
\usepackage{epstopdf}
\usepackage{algorithmic}
\usepackage{amsmath} 
\usepackage{amssymb} 
\usepackage{stackrel} 
\usepackage{csquotes}
\usepackage[makeroom]{cancel} 
\usepackage{bbm} 
\usepackage{subcaption} 
\usepackage{setspace}
\ifpdf
  \DeclareGraphicsExtensions{.eps,.pdf,.png,.jpg}
\else
  \DeclareGraphicsExtensions{.eps}
\fi

\newcommand{\tin}{t \in [0, T]}

\newcommand{\lambar}{\overline{\lambda}}

\newcommand{\ddy}{\frac{\partial}{\partial y}}
\newcommand{\ddv}{\frac{\partial}{\partial v}}
\newcommand{\piu}{\pi^{\ast}_{u}}
\newcommand{\pic}{\pi^{\ast}_{c}}
\newcommand{\rBrackets}[1]{\left( #1 \right)}
\newcommand{\sBrackets}[1]{\left[ #1 \right]}
\newcommand{\EVtyv}[2]{\mathbb{E}^{\mathbb{#1}}_{t,y,v} \sBrackets{#2}}

\newcommand{\EVtxv}[2]{\mathbb{E}^{\mathbb{#1}}_{t,x,v} \sBrackets{#2}}
\newcommand{\udptyv}{\overline{U}^{D,\mathbb{P}}(t,y,v)}
\newcommand{\udp}{\overline{U}^{D,\mathbb{P}}}
\newcommand{\dqlatyv}{D^{\mathbb{Q}(\lambda^{v})}(t,y,v)}

\newcommand{\dqla}{D^{\mathbb{Q}(\lambda^{v})}}
\newcommand{\dq}{D^{\mathbb{Q}}}
\newcommand{\vfc}{\mathcal{V}^{c}}
\newcommand{\vfu}{\mathcal{V}^{u}}
\newcommand{\ubard}{\overline{U}^{D}}
\newcommand{\ubardi}[1]{\overline{U}^{D}_{#1}}
\newcommand{\vcansatz}{\hat{\mathcal{V}}^{c}}

\newsiamremark{remark}{Remark}
\newsiamremark{hypothesis}{Hypothesis}
\crefname{hypothesis}{Hypothesis}{Hypotheses}
\newsiamthm{claim}{Claim}

\headers{VaR-constrained portfolios in incomplete markets}{M. Escobar-Anel, Y. Havrylenko, and R. Zagst}

\title{Value-at-Risk constrained portfolios in incomplete markets: a dynamic programming approach to Heston's model}

\author{Marcos Escobar-Anel\thanks{Department of Statistical and Actuarial Sciences, Western University, 1151 Richmond street, London, Canada
  (\email{marcos.escobar@uwo.ca}).}
\and Yevhen Havrylenko\thanks{TUM School of Computation, Information and Technology, Technical University of Munich, Parkring 11, 85748 Garching bei M\"unchen, Germany. Present address: Department of Mathematical Sciences, University of Copenhagen, Universitetsparken 5, 2100, Copenhagen, Denmark (\email{yh@math.ku.dk}).}
\and Rudi Zagst\thanks{TUM School of Computation, Information and Technology, Technical University of Munich, Parkring 11, 85748 Garching bei M\"unchen, Germany (\email{zagst@tum.de}).}}

\usepackage{amsopn}

\allowdisplaybreaks
